\documentclass[a4paper,onecolumn, accepted=2024-02-24]{quantumarticle}
\pdfoutput=1

\usepackage{fullpage}

\usepackage[T1]{fontenc}
\usepackage[english]{babel}
\usepackage[utf8]{inputenc}
\usepackage[dvipsnames]{xcolor}
\usepackage[colorlinks=true,urlcolor=black,citecolor=black,linkcolor=black,anchorcolor=black]{hyperref}
\usepackage{graphics, graphicx, url, color, physics, cancel, multirow, bbm, soul, mathtools, mathrsfs, amsfonts, amssymb, amsmath, amsthm, dsfont, tabularx, bm, verbatim, theoremref, tikz, caption, subcaption, soul, array, svg, nicefrac, booktabs}
\usepackage[normalem]{ulem}
\usepackage[ruled,vlined]{algorithm2e}

\newtheorem{definition}{Definition}
\newtheorem{theorem}{Theorem}

\newtheorem{lemma}{Lemma}

\DeclareMathOperator{\msg}{MSG_3}
\DeclareMathOperator{\msgt}{MSG_2}

\def\Pax{{P_a^x}}
\def\Qby{{Q_b^y}}
\def\Fcxy{{F_c^{xy}}}
\def\Pxy{{\Pi^{xy}}}
\DeclareMathOperator{\otwo}{\omega_2}
\DeclareMathOperator{\othree}{\omega_3}
\DeclareMathOperator{\oexp}{\omega_\textup{exp}}
\DeclareMathOperator{\Prtext}{Pr}

\newcommand{\kettext}[1]{| #1 \rangle}

\newcommand{\ketbratext}[2]{\left|#1\right\rangle\kern-0.3em\left\langle#2\right|}

\def\lkey{{l_\textup{key}}}

\def\dtol{{\delta_\textup{tol}}}
\def\epscorr{{\varepsilon_\textup{corr}}}
\def\epssou{{\varepsilon_\textup{sou}}}
\def\epssec{{\varepsilon_\textup{sec}}}
\def\epscom{{\varepsilon_\textup{com}}}
\DeclareMathOperator{\EC}{EC}
\DeclareMathOperator{\PE}{PE}
\DeclareMathOperator{\PA}{PA}
\DeclareMathOperator{\key}{key}
\DeclareMathOperator{\g}{g}

\DeclareMathOperator{\hon}{hon}

\def\WAB{{W_\textup{AB}}}
\def\GE{{G_\textup{E}}}
\def\GEc{{G^c_\textup{E}}}

\DeclareMathOperator{\Max}{Max}
\DeclareMathOperator{\Min}{Min}
\DeclareMathOperator{\argmin}{arg\,min}
\DeclareMathOperator{\supp}{supp}
\DeclareMathOperator{\CPTP}{CPTP}

\begin{document}

\title{Device independent security of quantum key distribution from monogamy-of-entanglement games}

\author{Enrique Cervero-Mart\'{i}n}
\email{enrique.cervero@u.nus.edu}
\affiliation{Centre for Quantum Technologies, National University of Singapore}

\author{Marco Tomamichel}
\email{marco.tomamichel@nus.edu}
\affiliation{Centre for Quantum Technologies, National University of Singapore and\\
Department of Electrical and Computer Engineering, National University of Singapore}


\begin{abstract}
We analyse two party non-local games whose predicate requires Alice and Bob to generate matching bits, and their three party extensions where a third player receives all inputs and is required to output a bit that matches that of the original players.
We propose a general device-independent quantum key distribution protocol for the subset of such non-local games that satisfy a monogamy-of-entanglement property characterised by a gap in the maximum winning probability between the bipartite and tripartite versions of the game.
This gap is due to the optimal strategy for two players requiring entanglement, which due to its monogamy property cannot be shared with any additional players.
Based solely on the monogamy-of-entanglement property, we provide a simple proof of information theoretic security of our protocol.
Lastly, we numerically optimize the finite and asymptotic secret key rates of our protocol using the magic square game as an example, for which we provide a numerical bound on the maximal tripartite quantum winning probability which closely matches the bipartite classical winning probability. 
Further, we show that our protocol is robust for depolarizing noise up to about $2.88\%$, providing the first such bound for general attacks for magic square based quantum key distribution.
\end{abstract}

\maketitle
\tableofcontents

\section{Introduction}
Device independent quantum key distribution (DIQKD) allows two honest parties (commonly called Alice and Bob) to establish an information theoretic secure key with minimal assumptions on the inner workings of their lab equipment, and in the presence of an eavesdropper (aptly named Eve) who in particular might have distributed, or is otherwise in control of the devices of the honest parties. 
Even when using such uncharacterised devices, the security of DIQKD can be guaranteed through a large enough violation of a \emph{Bell inequality} \cite{BHK05, AMP06}, as this ensures that the correlations shared amongst the honest parties arise from a quantum behaviour.
Such violations are used in DIQKD to certify that the honest parties share an entangled state, which may not be further entangled to an eavesdropper due to the \emph{monogamy of entanglement}, a property of multi-partite quantum systems which asserts that strongly correlated two party states may only be weakly correlated to a third party \cite{Werner89, DPS04}.
Crucially, the {monogamy-of-entanglement} property asserts that meddling eavesdroppers may not be strongly correlated to the honest parties measurement results if the latter violate a Bell inequality, which implies that the eavesdropper's available information about the honest users' exchanged keys in a DIQKD protocol may be bounded.

{To make the latter statement more precise, secrecy of a QKD protocol in the finite round regime requires a lower bound on the smooth min-entropy of an honest party's raw key conditioned on the eavesdropper's side information \cite{Renner05, TL17}.
Roughly speaking, the smooth min-entropy measures the worst-case uncertainty of Alice's raw key from the perspective of Eve.
Lower bounding this quantity is difficult in practice, but dependant on the assumptions of the operations across the rounds of a QKD protocol that produce the raw key, some simplifications are possible.}
For example, using the \emph{Asymptotic Equipartition Property} \cite{TCR09} if the channels are identical and memoryless, the \emph{Entropy Accumulation Theorem} (EAT) \cite{DFR20, DF19} if the channels satisfy a Markov chain condition, or the more recent \emph{Generalised Entropy Accumulation Theorem} (GEAT) \cite{MFSR22} if the channels satisfy a no-signalling condition.
The EAT and GEAT provide the highest security guarantee as they allow the eavesdropper to hold a transcript of any public discussion and a quantum memory when performing their attack.
This is known as a \emph{coherent} or \emph{general} attack.
In contrast, \emph{memoryless} (also known as \emph{collective}) attacks force the eavesdropper to apply her attack after each round of measurements by Alice and Bob without any additional classical or quantum side information.
Security of DIQKD protocols for general attacks in the finite round regime has been analysed extensively.
See for example the framework in \cite{Renner05, TL17, T21}, the first security proofs in the finite regime \cite{VR08, SPV10, TLGR12, TFKW13, VV14}, security proofs via EAT \cite{AFDF+18, AFRV19, T21, TSG+21, STP+21, TSB+22, NDN+22}, and via the more recent GEAT in \cite{MR22}.
See also \cite{PAB+20, ZLetal23, PGTKGL23} for recent reviews on other tools and techniques, and \cite{NDN+22, Zetal22} for experimental realizations of DIQKD.

We remark that monogamy of entanglement can be used directly to prove the secrecy of QKD protocols by bounding the information available to the eavesdropper.
To do this, monogamy-of-entanglement games were introduced in \cite{TFKW13} to bound the probability with which an eavesdropper is able to make a successful guess of the key sifted by the honest parties in a semi-DIQKD protocol. 
Informally, monogamy-of-entanglement games are protocols which require the participating parties to produce the same output, but whose success probability decreases as more parties are added.
This gap in winning probabilities is due to the optimal strategy for two parties requiring entanglement, which due to monogamy-of-entanglement property cannot be shared with additional parties.
In fact, the use of monogamy of-entanglement in security proofs has recently received attention for QKD and various cryptographic primitives.
For example, in \cite{V17, JMS20, JK21} the authors use the magic square game, whose bipartite optimal strategy self-tests maximal entanglement \cite{WBMS16}, to establish secure cryptographic keys. 
Likewise, the authors in \cite{CV22, CVV22} introduce a monogamy-of-entanglement game based on coset states and use it as primitive in a continuous variable QKD protocol.
Furthermore, monogamy-of-entanglement games are used in \cite{PK18} for oblivious transfer, and in \cite{CLLZ22, KT22} for uncloneable encryption.

This begs the question: is it possible to explicitly prove secrecy \emph{of all} DIQKD protocols using arbitrary games which exhibit monogamy of entanglement?
Typically, the Clauser-Horne-Shimony-Holt (CHSH) game is chosen due to its simplicity and relative ease of implementation \cite{AMP06, PAB+09, VV14, TSG+21, TSB+22}, although more general CHSH games \cite{SBV+21} and even other games like the magic square game \cite{MS16, JMS20, V17, JK21} have been considered.
In this work, we use GEAT to show it is indeed possible to construct device independent quantum key distributions protocols from any two player non-local game which exhibits the monogamy-of-entanglement property characterised by a sufficiently large gap in the winning probability between two party and three party instances of the game.
For this class of games, we show in \thref{th_secrecy} that the entropy of Alice's key bit conditioned on Eve's side information in one instance of the non-local game is bounded below by a family of simple affine functions $g:[0,1]\rightarrow\mathbb{R}$.
One such example is given by
\begin{align}
    {g}(p):=\frac{p-\otwo}{\ln2\cdot(1-\otwo+\othree)} - \log\big(1-\otwo+\othree\big),
\end{align}
where $\omega_2, \omega_3$ are the optimal quantum winning probabilities of the bipartite and tripartite instances of the non-local game.
Further, we highlight that in some monogamy-of-entanglement games, the tripartite quantum strategy does not outperform the classical strategy. 
For example, this is the case for the monogamy-of-entanglement game in \cite{TFKW13}, for the CHSH game, and as we show in this work, for the magic square game as well.
This is because the optimal bipartite strategies require sharing maximally entangled states which due to the monogamy of entanglement, may not be additionally shared by a third party.
The collapse of the tripartite quantum value onto the classical value is of independent theoretical interest but is in fact not a requirement for our security proof, which requires only a gap between the tripartite and bipartite quantum values of the underlying game.
There might exist monogamy-of-entanglement games in which the tripartite quantum value exceeds the classical value, for example if the optimal strategy requires a (non-maximally entangled) state which remains entangled after tracing out one of the three parties.
We leave the existence of these monogamy-of-entanglement games for future work.

As a result, our framework provides an alternative viewpoint for the security of QKD protocols which relies directly on monogamy of entanglement via a gap between the bipartite-tripartite gap in winning probability.
In contrast, the security of QKD protocols often seen in the recent literature \cite{AFRV19, TSG+21, TSB+22} relies on the robustness of the shared correlations, which is more indirectly dependent on the monogamy of entanglement.
In contrast to these works, we focus mainly on spot-checking QKD protocols with identical testing and generation rounds. 
An advantage of this approach is that it is simpler to obtain entropy lower bounds for the achievable key rates. 
Unfortunately, this comes at the cost of a larger error correction cost which in some cases (see Figure \ref{fig:nlg_pos_keyrate}) irrevocably yields a trivial key rate.
To rectify this, it is possible to follow the framework in \cite{AMP06, PAB+09, AFRV19, TSB+22} and add additional measurement settings which differentiate testing rounds from key generation rounds in the QKD protocol. 
Even though we do not explore additional measurement settings in this work, we remark that our analysis and security proofs largely extends to that framework, and that it suffices to use the entropy lower bounds (specifically \emph{min-tradeoff functions}) found in \cite{AMP06, PAB+09, AFRV19, TSB+22} in place of the ones derived in this work to prove the security of the QKD protocol with additional measurement settings.

We showcase our results with the \emph{magic square game} (also known as the \emph{Mermin-Peres game}) of \cite{Mermin90, Peres90}.
For this game, we compute an analytical upper bound on the optimum tripartite winning probability which is arbitrarily close to the maximum bipartite classical winning probability.
Further, we use the NPA hierarchy of \cite{NPA08, PNA10, Wittek15} and the tools from \cite{BFF23} to compute tighter bounds on the min-entropy and von Neumann entropy, respectively, of Alice's key bit conditioned on the eavesdroppers information in one instance of the magic square game.
We use these improved bounds to numerically optimize the finite and asymptotic key rates of magic square game based DIQKD, in particular showing robustness for depolarizing noise of up to $2.88\%$.

We remark that previous works \cite{JMS20, V17, JK21} also explore the security of MSG based DIQKD from the viewpoint of parallel repetition of monogamy-of-entanglement games. 
This article differs from these previous works by extending the analysis to arbitrary monogamy-of-entanglement games and including more detailed analysis of security in the finite regime of sequential MSG based DIQKD, including on error correction, completeness and numerically optimized key rates.
Moreover, the recent work \cite{ZMZ+23} which we only became aware of while concluding our study also explores security of MSG based DIQKD using the numerical tools from \cite{BFF23} and under the memoryless attack assumption. 
We use the same numerical tools for the computation of the von Neumann entropies, but we generalise on their work by not limiting to security against memoryless eavesdroppers in the asymptotic regime, instead considering security in the finite round regime under general attacks.
Further, we also improve over the asymptotic key rates and robustness against depolarizing noise reported in \cite{ZMZ+23}.

The remainder of the paper is structured as follows.
We start Section \ref{sec_qkd} by introducing some basic definitions on quantum information, non-local games and quantum key distribution that will be used in the text. 
In this section we also present and discuss our DIQKD protocol, protocol $\mathcal{P}$ in algorithm \ref{prot_qkd_actual}.
In Section \ref{sec_main_results} we present our main theorems, \thref{th_completeness}, \thref{th_correctness} and \thref{th_secrecy} for the completeness, correctness and secrecy of $\mathcal{P}$.
The proof of these theorems is deferred to Section \ref{sec_finite_proof}, where we also include a second preliminaries section with more involved definitions and results that are needed in the proofs.
Of particular interest is \thref{th_max_entropy_bound} in Section \ref{sec_secrecy_proof}, a new result which provides an efficient bound to the smooth max-entropy terms appearing in EAT-based secrecy proofs.
Lastly, in Section \ref{sec_msg_keyrates} we present the numerically optimized finite and asymptotic key rates for the magic square game based DIQKD protocol using the tailored entropy bounds constructed in Section \ref{sec_msg}.

\section{Quantum key distribution from non-local games}\label{sec_qkd}

\subsection{Preliminaries} \label{sec_prelims_1}
In this section we introduce the basic concepts necessary for understanding our main results and refer to Section \ref{sec_prelims_2} for the definitions required in the proofs.
Table \ref{tab:notation} contains some notation that will be used through this text and Table \ref{tab:state_notation} specifies the greek letters used to denote the quantum states at different stages of a QKD protocol.

\begin{table}[h!]
\caption{Notation}
\def\arraystretch{1.5}
\setlength\tabcolsep{.28cm}
\centering
\begin{tabular}{c l}
\toprule
\textit{Symbol} & \textit{Definition} \\
\toprule

$A_i^j$ & {collection of} registers $A_i,...,A_j$  \\
\hline
$S(A)$ (resp. $S_\leq(A)$) & Set of normalized (resp. subnormalised) states on register $A$ \\
\hline
$\tau_A$ & Maximally mixed state on register $A$ \\
\hline
$\mathbb{I}_A$ & Identity operator on register $A$ \\
\hline
$L^\dagger$ & Adjoint of map $L$ \\
\hline
$\Omega^c$ & Complement of event $\Omega$ \\
\hline
$\neg 0$ & Values other than $0$ \\
\hline
$\log$ (resp. $\ln$) & Base-$2$ (resp. base-$\mathrm{e}$) logarithm \\
\hline
$\mathcal{P}(\mathcal{C})$ & Set of probability distributions on alphabet $\mathcal{C}$\\
\hline
$\norm{\cdot}_p$ & Schatten $p$-norm \\
\toprule
\end{tabular}
\def\arraystretch{1}
\label{tab:notation}
\end{table}

\begin{table}[h!]
\caption{Quantum states in QKD protocols}
\def\arraystretch{1.5}
\setlength\tabcolsep{.28cm}
\centering
\begin{tabular}{c l}
\toprule
\textit{State} & \textit{Usage} \\
\toprule

$\rho_{K_AE}$ & State at the end of QKD protocols, contains final key  \\
\hline
$\sigma_{S_AE}$ & State prior to privacy amplification, contains raw key  \\
\hline
$\nu_{Q_AQ_B}$ & State in the quantum strategy of bipartite non-local games  \\
\hline
$\tilde{\nu}_{Q_AQ_B}$ & Noisy state in the quantum strategy of bipartite non-local games \\
\hline
$\omega^i_{R}$ & State in the honest parties' devices during $i$-th measurement round \\
\toprule
\end{tabular}
\def\arraystretch{1}
\label{tab:state_notation}
\end{table}

A \emph{quantum state} $\rho$ is a unit trace, positive-definite operator acting on a Hilbert space.
The collection of quantum states on Hilbert space $A$ is denoted $S(A)$.
We say $\rho\in S(A)$ is pure if it has unit rank, and mixed otherwise.
On the other hand, a subnormalised state is a positive-definite operator on Hilbert space $A$ with trace smaller than $1$.
The set of subnormalised states on $A$ is denoted $S_\leq (A)$.
In this text we will only consider finite dimensional Hilbert spaces and quantum systems.

To quantify the distance between any two quantum states $\rho, \sigma\in S(A)$ we use the \emph{trace distance}
\begin{align}
    \norm{\rho-\sigma}_{\Tr} := \frac{1}{2}\norm{\rho-\sigma}_1,
\end{align}
where $\norm{\;\cdot\;}_1$ is the Schatten $1$-norm which equals the sum of the singular values of its argument.

A \emph{register} is a quantum system with a fixed orthonormal basis which describes a discrete random variable.
For example, if random variable $X$ takes values over alphabet $\mathcal{X}$ with probabilities $x\mapsto p_x$ then a joint \emph{classical-quantum} state $\rho_{XA}\in S(XA)$ is of the form
\begin{align}
    \rho_{XA} = \sum_{x\in\mathcal{X}}\ketbra{x}\otimes \rho_A^x,
\end{align}
where $\rho_A^x\in S_\leq (E)$ is a subnormalised state with $\Tr[\rho_A^x]=p_x$.
Further, if $\Omega\subset\mathcal{X}$ is some event, we define the following conditioned states
 \begin{align}
    \rho_{AX\wedge\Omega} &:=\sum_{x\in\Omega}\ketbra{x}\otimes\rho_A^x,\\
    \rho_{AX|\Omega} &:= \frac{\Tr[\rho_{AX}]}{\Tr[\rho_{AX\wedge\Omega}]}\rho_{AX\wedge\Omega},
\end{align}
where of course $\Pr[\Omega]_\rho = \Tr[\sum_{x\in\Omega}\rho_A^x] = \sum_{x\in\Omega}p_x$.
Lastly, we use $X_1^n$ to denote a collection of classical registers $X_1,...,X_n$ over the same alphabet $\mathcal{X}$.

Maps from states in $S(A)$ to states in $S(B)$ are described by \emph{quantum channels}, which are completely positive and trace preserving mappings (CPTP).
We use $\CPTP(A,B)$ to denote the set of quantum channels from $A$ to $B$.
Quantum channels that are of particular interest are \emph{measurements} and \emph{instruments}.
Given a positive operator valued measure (POVM) $\{P_x\}_{x\in\mathcal{X}}$ on system $A$, that is, a set of positive semi-definite operators on $A$ satisfying $\sum_x P_x = \mathbb{I}_A$, then a measurement $\mathcal{M}\in \CPTP(A,X)$ applied to state $\rho_{A}\in S(A)$ is of the form
\begin{align}
    \mathcal{M}:\rho_{A} \mapsto \sum_{x\in\mathcal{X}}\ketbra{x}_X \Tr[\rho_{A}P_x].
\end{align}
If the operators $\{P_x\}_x$ additionally satisfy $P_x^2=P_x$ for all $x\in\mathcal{X}$ then the corresponding set $\{P_x\}_x$ is called a projection valued measure (PVM).
On the other hand, if $\mathcal{M}_x\in \textup{CP}(A,B)$ are trace non-increasing maps that satisfy $\sum_{x\in\mathcal{X}}\Tr[\mathcal{M}_x(\rho_A)] = \Tr[\rho_A]$ then a quantum instrument is the map $\mathcal{M}\in\CPTP(A,XB)$ given by
\begin{align}
    \mathcal{M}:\rho_A \mapsto \sum_{x\in\mathcal{X}}\ketbra{x}_X \otimes \mathcal{M}_x(\rho_A).
\end{align}
If the register $X$ is subsequently measured and outcome $x\in\mathcal{X}$ is observed, we call $\frac{\mathcal{M}_x(\rho_A)}{\Tr[\mathcal{M}_x(\rho_A)]}$ the \emph{post-measurement} state.

\subsection{Non-local games} \label{sec_NLG}
We introduce some basic definitions of \emph{non-local games} and refer to \cite{PV16} for a review.
\begin{definition}[Non-local game]
    A two player non-local game is a tuple $\mathcal{G}=(\pi,\mathcal{X},\mathcal{Y},\mathcal{A},\mathcal{B}, V)$ of
    \begin{itemize}
        \item Finite sets of questions $\mathcal{X}, \mathcal{Y}$ and answers $\mathcal{A},\mathcal{B}$,
        \item Probability distribution $\pi:\mathcal{X}\times\mathcal{Y}\rightarrow[0,1]$,
        \item Predicate $V:\mathcal{X}\times\mathcal{Y}\times\mathcal{A}\times\mathcal{B}\rightarrow\{0,1\}$.
    \end{itemize}
    If $\pi$ is uniform, $\mathcal{G}$ is known as a {free} game.
\end{definition}

The definitions of two player non-local games can be generalised to multi-player non-local games in the obvious way.

In an instance of $\mathcal{G}$, the players respectively receive questions $x\in\mathcal{X}, y\in\mathcal{Y}$ with probability $\pi(x,y)$ and output answers $a\in\mathcal{A}, b\in\mathcal{B}$.
The players win if $V(x,y,a,b)=1$ and lose otherwise.

Prior to the start of the game, the players are allowed to prepare a strategy that dictates their behaviour through the game. 
Once the game begins, they are no longer allowed to communicate.

\begin{definition}[Strategy]
    A strategy between two honest players is a conditional distribution $\Pr[a,b|x,y]$.
    This strategy may be:
    \begin{itemize}
        \item Classical: in which there exists a hidden variable $\lambda$ such that
        \begin{align}
            \Pr[a,b|x,y]=\sum_\lambda \Pr[\lambda]\Pr[a|x,\lambda]\Pr[b|y,\lambda].
        \end{align}
        \item Quantum: in which the conditional distribution is 
        \begin{align}
            \Pr[a,b|x,y] = \Tr[\nu_{AB}(P_a^x\otimes Q_b^y)],
        \end{align}
        for a pre-agreed quantum state $\nu_{AB}\in S(AB)$ and choice of POVMs $\big\{\{\Pax\}_a\big\}_x$ in Alice's lab and $\big\{\{\Qby\}_b\big\}_y$ in Bob's lab.
        \item No-signalling: in which the distribution $\Pr[a,b|x,y]$ satisfies the constraints
        \begin{align}
            &\sum_{a\in\mathcal{A}}\Pr[a,b|x,y] = \sum_{a\in\mathcal{A}}\Pr[a,b|x',y] \qquad \text{for all } a,b,y,\\
            &\sum_{b\in\mathcal{B}}\Pr[a,b|x,y] = \sum_{b\in\mathcal{B}}\Pr[a,b|x,y'] \qquad \text{for all } a,b,x.
        \end{align}
    \end{itemize}
\end{definition}

\begin{definition}[Value of the game]
    The value of a two player non-local game $\mathcal{G}$ is
    \begin{align}
        \omega(\mathcal{G}) &= \sup_{\Pr} \sum_{x,y;a,b} \pi(x,y) \Pr[a,b|x,y] V(x,y;a,b),
    \end{align}
    where the supremum is over permissible strategies (classical, quantum or no-signalling).
\end{definition}

For classical strategies, we use $\omega_\textup{C}$ to denote the \emph{classical value} of $\mathcal{G}$ when $\mathcal{G}$ is inferred from context.
Likewise, we use $\omega_\textup{Q}$ and $\omega_\textup{NS}$ to denote the \emph{quantum} and \emph{no-signalling} values, respectively.
For any non-local game, it is clear that $\omega_\textup{C}\leq\omega_\textup{Q}\leq\omega_\textup{NS}$.

For the task of QKD, we focus on two player non-local games $\mathcal{G}_2=(\pi,\mathcal{X},\mathcal{Y},\mathcal{A},\mathcal{B},V)$ where $X$ and $Y$ are independent such that $\pi(x,y)$ is a product and the predicate is of the form
\begin{align}\label{predicate}
    V(x,y,a,b) = [SK_A(a,x,y) = SK_B(b,x,y)].
\end{align}
Here the square brackets evaluate to $1$ if its argument is true, and to $0$ otherwise. 
The function $SK_A:\mathcal{X}\times\mathcal{Y}\times\mathcal{A}\rightarrow\{0,1\}$ (respectively $SK_B:\mathcal{X}\times\mathcal{Y}\times\mathcal{B}\rightarrow\{0,1\}$) is a deterministic function that outputs a single bit on input of $x, y$ and $a$ (respectively, $x, y$ and $b$). 
For example, in the well known CHSH game Alice and Bob receive respective binary bits $x,y$, output respective bits $a,b$, and win if $x\cdot y = a\oplus b$.
With our notation, this is equivalent to setting $SK_A(a,x,y)=a$ and $SK_B(b,x,y) = b\oplus(x\cdot y)$ (see Figure \ref{fig:nlg_pos_keyrate} for additional context).
On the other hand, for the magic square game introduced in Section \ref{sec_msg_keyrates}, Alice and Bob receive respective inputs $x,y\in\{0,1,2\}$, produce respective bit triples $a,b\in\{0,1\}^3$ satisfying some parity conditions, and win if Alice's $y$-th bit equals Bob's $x$-th bit. 
In our framework, this is captured by functions $SK_A(a,x,y)=a[y]$ and $SK_B(a,x,y)=b[x]$.

We consider the three player extension $\mathcal{G}_3=(\pi,\mathcal{X},\mathcal{Y},\mathcal{A},\mathcal{B},\{0,1\},\widetilde{V})$ in which the third player receives both the inputs of the original players $(x,y)$ and outputs a single bit $c\in\{0,1\}$. 
The new predicate in $\mathcal{G}_3$ is of the form
\begin{align}
    \widetilde{V}(x,y,a,b,c) = [SK_A(a,x,y) = SK_B(b,x,y)] \cdot [SK_A(a,x,y) = c].
\end{align}

In the context of QKD, the game $\mathcal{G}_2$ corresponds to the protocol `played' by the \emph{honest} parties and the functions $SK_A, SK_B$ correspond to the functions used by Alice and Bob to obtain their respective raw keys and the parameter estimation registers. 
Here, \emph{honest} refers to the behaviour of the players and in particular means that they are playing the game as per the instructions without trying to cheat, disrupt or break the protocol in any malicious way.
Consequently, the third player in the game $\mathcal{G}_3$ corresponds to the eavesdropper in a QKD protocol who is attempting to guess Alice's sifted key bits $SK_A(a,x,y)$ during every round after all the inputs have been announced by the honest players.

For the rest of the article, we let $\omega_2 = \omega_\textup{Q}(\mathcal{G}_2)$, $\omega_3 = \omega_\textup{Q}(\mathcal{G}_3)$.
Lastly, we introduce $\omega_\textup{exp}\in[\omega_3, \omega_2]$, the expected winning probability of the two honest parties in an instance of $\mathcal{G}_2$.
For example, the discrepancy between $\omega_2$ and $\omega_\textup{exp}$ may be due to the noise present in the physical quantum devices which Alice and Bob use to distribute or measure their shared quantum states.

\subsection{Security definitions of QKD}
In this text we focus on entanglement based device independent QKD protocols between two honest parties, Alice and Bob, and a potential eavesdropper, Eve.

A QKD protocol is an interactive multi-round protocol between Alice and Bob, characterised by a sequence of local (classical or quantum) operations and classical communication (LOCC).
In \emph{entanglement based} protocols, a source produces a joint quantum system and distributes it amongst the honest parties Alice and Bob, who subsequently perform measurements on it with randomly and independently chosen measurement settings.

Before discussing the meaning of \emph{device independent} protocols, we state some necessary assumptions of QKD protocols without which security guarantees are not applicable.
\begin{itemize}
    \item All parties are bound by quantum mechanics\footnote{We remark that existing results suggest security of device independent QKD is possible if quantum theory is replaced by a more general no-signalling theory \cite{MRC+14}};
    \item The devices of Alice and Bob are spatially separated.
    In particular, this allows us to describe operations in Alice's, Bob's or Eve's labs as a tensor product of local operations in the respective Hilbert spaces;
    \item The devices used by the honest parties do not communicate or leak information to any eavesdropper or to each other. 
    Further, the honest parties have control over what their respective devices receive as inputs.
    In practice, these assumption can be enforced by shielding the devices and/or by ensuring space-like separation between the honest parties;
    \item The honest parties have access to trusted sources of randomness;
    \item The honest parties have access to an authenticated classical channel. 
    In practice this is not a problem if the honest parties share a small amount of secret information \cite{RW04, DW09}.
    We remark that the communication passing through this channel is public and as such must be included in the eavesdropper's side information.
\end{itemize}
In addition to the above and for the sake of technical convenience, we will also assume all Hilbert spaces are finite dimensional.

In \emph{device independent} quantum key distribution (DIKD) protocols, we further assume that the devices used by Alice and Bob are uncharacterized.
Formally, a device $\mathcal{D}$ is an object capable of holding a quantum state and performing quantum instruments dependent on a classical input, producing a classical output and updating this state.
\begin{definition}[Quantum Device]
    An $(\mathcal{X}, \mathcal{A}, R)$-device is characterised by a set $\mathcal{X}$ (the inputs), a set  $\mathcal{A}$ (the outputs), a finite-dimensional quantum system $R$ (the memory), as well as a collection of quantum instruments $\{ \mathcal{M}^{a|x} : a \in \mathcal{A},\ x \in \mathcal{X} \}$ where $\mathcal{M}^{a|x}$ are completely positive maps such that $\sum_{a \in \mathcal{A}} \mathcal{M}^{a|x} \in \CPTP(R, R)$ for every $x \in \mathcal{X}$. The device takes as input a state $\omega_R^0$ and then behaves in the following way when interacted with for the $i$-th time:
    \begin{itemize}
        \item Upon receiving an input $x_i \in \mathcal{X}$ it outputs $a_i \in \mathcal{A}$ with probability $\Tr[ \mathcal{M}^{a_i|x_i} (\omega_R^{i-1}) ]$.
        \item The internal memory is updated to $\omega_R^{i} = \frac{\mathcal{M}^{a_i|x_i}(\omega_{R}^{i-1})}{\Tr[ \mathcal{M}^{a_i|x_i} (\omega_R^{i-1}) ]}$.
    \end{itemize}
\end{definition}
Note that this device model is very general. For honest implementations the quantum system $R$ may contain an index register that is updated every time the device is used so that fresh entangled states can be used each time the game is played. 
Moreover, via this index register and memory registers, the instrument can be made dependent on all the inputs and outputs of previous rounds.
In particular, this notation removes the need of an additional label $i$ on the register $R$ or the quantum instruments $\mathcal{M}^{a|x}$.
While the internal state and the quantum instruments are usually specified for honest implementations our security guarantees are meant to be device-independent, that is, they do not depend on the state $\omega_R$ (and its potential extensions) or on the quantum instrument used to update this state and produce the classical outputs.\footnote{We could even define a universal device that determines its behaviour after reading out a program from a classical register in $R$; however, this perspective does not seem to simplify our arguments significantly.}

A (device independent) QKD protocol takes as input the pair of devices $\mathcal{D}_A$ for Alice and $\mathcal{D}_B$ for Bob, which hold respective (and possibly entangled) quantum registers $R_A$ and $R_B$.

\begin{definition}[QKD Protocol]
    A $(n,\ell)$-QKD protocol $\mathcal{P}_{\mathcal{D}_A, \mathcal{D}_B}$ for a game $\mathcal{G} = (\mathcal{X},\mathcal{Y},\mathcal{A},\mathcal{B},V)$ for two parties (Alice and Bob) takes as argument a $(\mathcal{X}, \mathcal{A}, R_A)$-device $\mathcal{D}_A$ for Alice and a $(\mathcal{Y}, \mathcal{B}, R_B)$-device $\mathcal{D}_B$ for Bob. 
    The protocol $\mathcal{P}_{\mathcal{D}_A, \mathcal{D}_B}$ is a CPTP map taking as input a state $\omega^0_{R_AR_B}\in\mathcal{S}(R_A R_B)$ that initialises the devices. 
    At the end of the protocol, either Alice and Bob output a string of $\ell$ bits in respective registers $K_A$ and $K_B$, or they abort the protocol with output $K_A=K_B=\ketbratext{\perp}{\perp}$. 
    The QKD protocol is comprised of local operations and classical communication (LOCC) over the authenticated channel. 
    Both parties also have access to trusted sources of randomness. 
    Alice can interact $n$ times with $\mathcal{D}_A$ and Bob can interact $n$ times with $\mathcal{D}_B$, where $n$ is the number of rounds. 
\end{definition}

\begin{figure}[h]
    \centering
    \scalebox{0.8}{\tikzset{every picture/.style={line width=0.75pt}} 

\begin{tikzpicture}[x=0.75pt,y=0.75pt,yscale=-1,xscale=1]

\draw [fill={rgb, 255:red, 208; green, 2; blue, 27 }  ,fill opacity=0.2 ]   (180,250) -- (180,80) -- (210,80) -- (210,120) -- (260,120) -- (260,140) -- (210,140) -- (210,200) -- (260,200) -- (260,220) -- (210,220) -- (210,250) ;
\draw [fill={rgb, 255:red, 208; green, 2; blue, 27 }  ,fill opacity=0.2 ] [dash pattern={on 0.75pt off 0.75pt on 0.75pt off 0.75pt}]  (180,250) -- (180,270) ;
\draw [fill={rgb, 255:red, 208; green, 2; blue, 27 }  ,fill opacity=0.2 ] [dash pattern={on 0.75pt off 0.75pt on 0.75pt off 0.75pt}]  (210,250) -- (210,270) ;

\draw [fill={rgb, 255:red, 208; green, 2; blue, 27 }  ,fill opacity=0.2 ]   (210,270) -- (210,300) -- (260,300) -- (260,320) -- (180,320) -- (180,270) ;

\draw [fill={rgb, 255:red, 208; green, 2; blue, 27 }  ,fill opacity=0.2 ]   (410,250) -- (410,80) -- (380,80) -- (380,120) -- (330,120) -- (330,140) -- (380,140) -- (380,200) -- (330,200) -- (330,220) -- (380,220) -- (380,250) ;
\draw [fill={rgb, 255:red, 208; green, 2; blue, 27 }  ,fill opacity=0.2 ] [dash pattern={on 0.75pt off 0.75pt on 0.75pt off 0.75pt}]  (410,250) -- (410,270) ;
\draw [fill={rgb, 255:red, 208; green, 2; blue, 27 }  ,fill opacity=0.2 ] [dash pattern={on 0.75pt off 0.75pt on 0.75pt off 0.75pt}]  (380,250) -- (380,270) ;

\draw [fill={rgb, 255:red, 208; green, 2; blue, 27 }  ,fill opacity=0.2 ]   (380,270) -- (380,300) -- (330,300) -- (330,320) -- (410,320) -- (410,270) ;

\draw [fill={rgb, 255:red, 80; green, 227; blue, 194 }  ,fill opacity=0.02 ]   (290,80) -- (360,80) -- (360,100) -- (310,100) -- (310,160) -- (360,160) -- (360,180) -- (310,180) -- (310,240) -- (360,240) -- (360,250) ;
\draw [fill={rgb, 255:red, 80; green, 227; blue, 194 }  ,fill opacity=0.02 ]   (360,270) -- (360,280) -- (310,280) -- (310,340) -- (410,340) -- (410,370) -- (290,370) ;
\draw [fill={rgb, 255:red, 80; green, 227; blue, 194 }  ,fill opacity=0.02 ] [dash pattern={on 0.75pt off 0.75pt on 0.75pt off 0.75pt}]  (360,250) -- (360,270) ;
\draw [fill={rgb, 255:red, 80; green, 227; blue, 194 }  ,fill opacity=0.02 ]   (290,80) -- (230,80) -- (230,100) -- (280,100) -- (280,160) -- (230,160) -- (230,180) -- (280,180) -- (280,240) -- (230,240) -- (230,250) ;
\draw [fill={rgb, 255:red, 80; green, 227; blue, 194 }  ,fill opacity=0.02 ]   (230,270) -- (230,280) -- (280,280) -- (280,340) -- (180,340) -- (180,370) -- (290,370) ;
\draw [fill={rgb, 255:red, 80; green, 227; blue, 194 }  ,fill opacity=0.02 ] [dash pattern={on 0.75pt off 0.75pt on 0.75pt off 0.75pt}]  (230,250) -- (230,270) ;
\draw  [draw opacity=0][fill={rgb, 255:red, 208; green, 2; blue, 27 }  ,fill opacity=0.2 ] (180,250) -- (180,270) -- (210,270) -- (210,250) -- (180,250) -- cycle ;
\draw  [draw opacity=0][fill={rgb, 255:red, 208; green, 2; blue, 27 }  ,fill opacity=0.2 ] (380,250) -- (380,270) -- (410,270) -- (410,250) -- (380,250) -- cycle ;
\draw  [draw opacity=0][fill={rgb, 255:red, 80; green, 227; blue, 194 }  ,fill opacity=0.2 ] (230,80) -- (360,80) -- (360,100) -- (230,100) -- (230,80) -- cycle ;
\draw  [draw opacity=0][fill={rgb, 255:red, 80; green, 227; blue, 194 }  ,fill opacity=0.2 ] (230,160) -- (360,160) -- (360,180) -- (230,180) -- (230,160) -- cycle ;
\draw  [draw opacity=0][fill={rgb, 255:red, 80; green, 227; blue, 194 }  ,fill opacity=0.2 ] (230,240) -- (360,240) -- (360,260) -- (230,260) -- (230,240) -- cycle ;
\draw  [draw opacity=0][fill={rgb, 255:red, 80; green, 227; blue, 194 }  ,fill opacity=0.2 ] (230,260) -- (360,260) -- (360,280) -- (230,280) -- (230,260) -- cycle ;
\draw  [draw opacity=0][fill={rgb, 255:red, 80; green, 227; blue, 194 }  ,fill opacity=0.2 ] (180,340) -- (410,340) -- (410,370) -- (180,370) -- (180,340) -- cycle ;
\draw  [draw opacity=0][fill={rgb, 255:red, 80; green, 227; blue, 194 }  ,fill opacity=0.2 ] (280,280) -- (310,280) -- (310,340) -- (280,340) -- (280,280) -- cycle ;
\draw  [draw opacity=0][fill={rgb, 255:red, 80; green, 227; blue, 194 }  ,fill opacity=0.2 ] (280,180) -- (310,180) -- (310,240) -- (280,240) -- (280,180) -- cycle ;
\draw  [draw opacity=0][fill={rgb, 255:red, 80; green, 227; blue, 194 }  ,fill opacity=0.2 ] (280,100) -- (310,100) -- (310,160) -- (280,160) -- (280,100) -- cycle ;
\draw    (195,60) -- (195,80) ;
\draw    (395,60) -- (395,80) ;
\draw  [dash pattern={on 4.5pt off 4.5pt}]  (295,40) -- (295,410) ;
\draw    (245,100) -- (245,120) ;
\draw    (245,140) -- (245,160) ;
\draw    (245,180) -- (245,200) ;
\draw    (245,220) -- (245,240) ;
\draw    (245,280) -- (245,300) ;
\draw    (345,280) -- (345,300) ;
\draw    (345,220) -- (345,240) ;
\draw    (345,140) -- (345,160) ;
\draw    (345,100) -- (345,120) ;
\draw    (245,320) -- (245,340) ;
\draw    (345,320) -- (345,340) ;
\draw    (345,180) -- (345,200) ;
\draw    (230,370) -- (230,390) ;
\draw    (360,370) -- (360,390) ;

\draw (186,42.4) node [anchor=north west][inner sep=0.75pt]  [font=\scriptsize]  {$R_{A}$};
\draw (386,42.4) node [anchor=north west][inner sep=0.75pt]  [font=\scriptsize]  {$R_{B}$};
\draw (248,33) node [anchor=north west][inner sep=0.75pt]  [font=\small] [align=left] {Alice};
\draw (311,33) node [anchor=north west][inner sep=0.75pt]  [font=\small] [align=left] {Bob};
\draw (184,202.4) node [anchor=north west][inner sep=0.75pt]  [font=\footnotesize]  {$\mathcal{D}_{A}$};
\draw (386,202.4) node [anchor=north west][inner sep=0.75pt]  [font=\footnotesize]  {$\mathcal{D}_{B}$};
\draw (227,104.4) node [anchor=north west][inner sep=0.75pt]  [font=\scriptsize]  {$X_{1}$};
\draw (347,104.4) node [anchor=north west][inner sep=0.75pt]  [font=\scriptsize]  {$Y_{1}$};
\draw (227,144.4) node [anchor=north west][inner sep=0.75pt]  [font=\scriptsize]  {$A_{1}$};
\draw (227,184.4) node [anchor=north west][inner sep=0.75pt]  [font=\scriptsize]  {$X_{2}$};
\draw (227,224.4) node [anchor=north west][inner sep=0.75pt]  [font=\scriptsize]  {$A_{2}$};
\draw (227,284.4) node [anchor=north west][inner sep=0.75pt]  [font=\scriptsize]  {$X_{n}$};
\draw (227,325.4) node [anchor=north west][inner sep=0.75pt]  [font=\scriptsize]  {$A_{n}$};
\draw (347,144.4) node [anchor=north west][inner sep=0.75pt]  [font=\scriptsize]  {$B_{1}$};
\draw (347,225.4) node [anchor=north west][inner sep=0.75pt]  [font=\scriptsize]  {$B_{2}$};
\draw (347,324.4) node [anchor=north west][inner sep=0.75pt]  [font=\scriptsize]  {$B_{n}$};
\draw (347,185.4) node [anchor=north west][inner sep=0.75pt]  [font=\scriptsize]  {$Y_{2}$};
\draw (347,284.4) node [anchor=north west][inner sep=0.75pt]  [font=\scriptsize]  {$Y_{n}$};
\draw (221,392.4) node [anchor=north west][inner sep=0.75pt]  [font=\scriptsize]  {$K_{A}$};
\draw (352,392.4) node [anchor=north west][inner sep=0.75pt]  [font=\scriptsize]  {$K_{B}$};

\end{tikzpicture}}
    \caption{An $(n,\ell)$-QKD protocol $\mathcal{P}_{\mathcal{D}_A,\mathcal{D}_B}(\cdot)$. 
    The `combs' on the left and right correspond to the devices used by Alice and Bob, respectively.
    On the other hand, the central part represents the operations by the honest parties, which include generation of the inputs fed to the combs during the measurement phase and the post-measurement classical processing after receipt of outputs $A_n,B_n$ to produce the final keys $K_A, K_B$.}
    \label{fig:qkd_comb}
\end{figure}
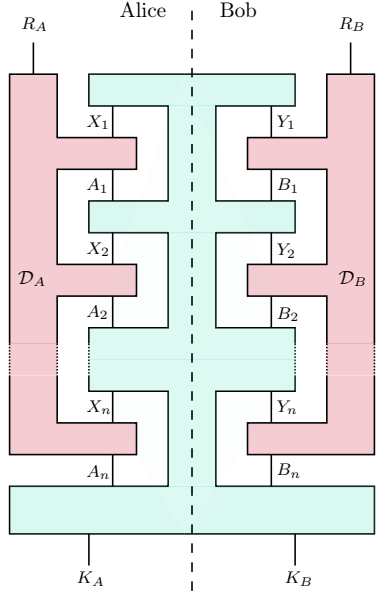

While the above is quite generic, usual protocols proceed in well-defined rounds and are then followed by a post-processing phase comprised of parameter estimation, error correction and privacy amplification. 
In each of the $n$ measurement rounds of a DIQKD protocol, Alice and Bob use their local sources of randomness to independently generate measurement settings $X_i\in\mathcal{X}$ and $Y_i\in\mathcal{Y}$ which they input to their respective devices $\mathcal{D}_A$ and $\mathcal{D}_B$ to receive classical outputs $A_i\in\mathcal{A}$ and $B_i\in\mathcal{B}$.
See Figure \ref{fig:qkd_comb} for a diagram of a QKD protocol using a $(\mathcal{X}, \mathcal{A}, R_A)$-device $\mathcal{D}_A$ for Alice and a $(\mathcal{Y}, \mathcal{B}, R_B)$-device $\mathcal{D}_B$ for Bob.
Usual protocols may abort in the parameter estimation or error correction phases, and throughout the text we use $\Omega$ to denote the event that the protocol \emph{does not abort}.
Consequently, $\Omega^c$ is the event that the protocol \emph{aborts}.

We call a pair of devices $\mathcal{D}^\textup{hon}_A$ and $\mathcal{D}^\textup{hon}_B$ together with an initial state $\omega_{R_A R_B}^\textup{hon}$ \emph{honest} with respect to a game $\mathcal{G}$ if the rounds are independent and, in each round, they win $\mathcal{G}$ with sufficiently high probability $\oexp$.
With these definitions, and following \cite{AFRV19}, we say that a QKD protocol is complete if the following holds.

\begin{definition}[Completeness]\thlabel{def_completeness}
    An $(n, \ell)$-QKD protocol $\mathcal{P}_{\cdot, \cdot}$ for a game $\mathcal{G}$ is $\varepsilon_\textup{com}$-complete if, for any honest triple $(\mathcal{D}_A^{\rm hon}, \mathcal{D}_B^{\rm hon}, \omega_{R_A R_B}^{\rm hon})$, we have
    \begin{align}
        \Pr[\Omega^c]_{\rho} \leq\varepsilon_\textup{com} \,, \qquad \text{where} \qquad 
        \rho_{K_A K_B} = \mathcal{P}_{\mathcal{D}_A^{\rm hon}, \mathcal{D}_B^{\rm hon}}(\omega_{R_A R_B}^{\rm hon}).
    \end{align}
\end{definition}

Completeness is purely a statement about the noise resilience of the protocol. 
Intuitively, a protocol is complete if the error correction and parameter estimation phases are designed so that in an honest implementation of the protocol subject to noise specified by some noise model, the probability of not aborting is higher than $1-\epscom$.

For the remaining security definitions we need to consider the scenario in which the devices are not honest and in which a potential eavesdropper is actively tampering with the QKD protocol in order to guess the keys generated by the honest parties. 
In particular, we consider devices $\mathcal{D}_A$ and $\mathcal{D}_B$ that are prepared by the eavesdropper, and denote the initial state shared amongst the two devices by $\omega^0_{R_AR_BE}$, where $E$ is additionally held by the eavesdropper\footnote{Without loss of generality we can assume that this state is pure, but this only clutters notation here.}.
The normalised state at the end of a QKD protocol $\mathcal{P}$ using such devices is
\begin{align}
    \rho_{K_AK_BE} = \mathcal{P}_{\mathcal{D}_A, \mathcal{D}_B}(\omega^0_{R_A R_B E}) = \Pr[\Omega]_\rho \cdot\rho_{K_AK_BE|\Omega} + \Pr[\Omega^c]_\rho \cdot\ketbra{\perp\perp}_{K_AK_B}\otimes\rho_{E|\Omega^c}. \label{eq:post-state}
\end{align}
Without loss of generality, we also assume that the eavesdropper is passive.
That is, the eavesdropper stores all public communication and maintains her own quantum register, but does not perform any operations on her quantum state until after the protocol is terminated.
Indeed, a data processing inequality asserts that her available (quantum or classical) information is maximised in this way.

Having established this notation, we define what it means for a QKD protocol receiving uncharacterised devices $\mathcal{D}_A$ and $\mathcal{D}_B$ as inputs to be \emph{secure}.
We adopt the security definitions from \cite{AFRV19} which are based on the definitions from \cite{PR14}.

\begin{definition}[Soundness]\thlabel{def_soundness}
    An $(n, \ell)$-QKD protocol $\mathcal{P}_{\cdot, \cdot}$ for a game $\mathcal{G}$ is $\varepsilon_\textup{sou}$-sound if for any triple $(\mathcal{D}_A, \mathcal{D}_B, \omega^0_{R_A R_B E})$, the post-protocol state~$\rho_{K_A K_B E}$ as defined in Eq.~\eqref{eq:post-state} satisfies
    \begin{align}
        \norm{\rho_{K_AK_BE\wedge\Omega}-\frac{1}{2^{\ell}}\sum_{i=0}^{2^{\ell}-1}\ketbra{ii}_{K_AK_B}\otimes\rho_{E\wedge\Omega}}_{\Tr}\leq\varepsilon_\textup{sou}\,.
    \end{align}
\end{definition}

In \cite{TL17} and \cite{AFRV19} it is proven that soundness is guaranteed by the following simpler conditions:
\begin{definition}[Correctness]\thlabel{def_correctness}
    An $(n, \ell)$-QKD protocol $\mathcal{P}_{\cdot, \cdot}$ for a game $\mathcal{G}$ is $\epscorr$-correct if for any triple $(\mathcal{D}_A, \mathcal{D}_B, \omega^0_{R_A R_B E})$, the post-protocol state~$\rho_{K_A K_B E}$ as defined in Eq.~\eqref{eq:post-state} satisfies
    \begin{align}
        \Pr[K_A\neq K_B \wedge \Omega]_\rho\leq\epscorr \,.
    \end{align}
\end{definition}

\begin{definition}[Secrecy]\thlabel{def_secrecy}
    An $(n, \ell)$-QKD protocol $\mathcal{P}_{\cdot, \cdot}$ for a game $\mathcal{G}$ is $\epssec$-secret if for any triple $(\mathcal{D}_A, \mathcal{D}_B, \omega_{R_A R_B E})$, the post-protocol state~$\rho_{K_A K_B E}$ as defined in Eq.~\eqref{eq:post-state} satisfies
    \begin{align}
        \norm{\rho_{K_AE\wedge\Omega}-\tau_{K_A}\otimes\rho_{E\wedge\Omega}}_{\Tr}\leq\epssec
    \end{align}
    where $\tau_{K_A}$ is the maximally mixed state of dimension $2^{\ell}$.
\end{definition}

In particular, if a protocol is $\epscorr$-correct and $\epssec$-secret then it is $(\epscorr+\epssec)$-sound.

\subsection{The NLG-DIQKD protocol} \label{sec_NLGDIQKD}
In this section, we introduce a device independent quantum key distribution protocol for arbitrary non-local games with suitable extensions as defined in Section \ref{sec_NLG}.
We prove security of this protocol for the family of non-local games satisfying $\omega_2>\omega_3$, a \emph{monogamy-of-entanglement} property arising from the fact that the optimal strategy for $\mathcal{G}_2$ requires the players to distribute entanglement and as such may not be further strongly correlated to the third player in $\mathcal{G}_3$ \cite{Werner89, DPS04}.

Before describing the protocol, we introduce the following definition.
\begin{definition}[2-universal hash function]\thlabel{def_hash_function}
    Let $\mathcal{H}:=\{H:\{0,1\}^m\rightarrow\{0,1\}^l\}$ be a family of hash functions. 
    The family $\mathcal{H}$ is $2$-universal if for any $H\in\mathcal{H}$ chosen uniformly
    \begin{align}
        \Pr[H(x)=H(x')]_H\leq \frac{1}{2^l},
    \end{align}
    for all distinct $x,x'\in\{0,1\}^m$.
    We refer to $l$ as the length of the outputs of the hash family.
\end{definition}

For the DIQKD protocol, fix a game $\mathcal{G}:=(\pi, \mathcal{X},\mathcal{Y},\mathcal{A},\mathcal{B}, V)$, where $\pi(x,y)$ is a product on distributions on $\mathcal{X}$ and $\mathcal{Y}$, and the predicate is of the form in Eq.~\eqref{predicate} for functions $SK_A:\mathcal{X}\times\mathcal{Y}\times\mathcal{A}\rightarrow\{0,1\}$ and $SK_B:\mathcal{X}\times\mathcal{Y}\times\mathcal{B}\rightarrow\{0,1\}$. 
Let $\otwo$ be the quantum value of $\mathcal{G}$ and let $\othree$ be the quantum value of its tripartite extension as described in Section \ref{sec_NLG}.
Protocol $\mathcal{P}$ is presented in algorithm \ref{prot_qkd_actual} below. 
We discuss its components separately.

\paragraph*{Measurement:}
In each round of the measurement phase, Alice and Bob use their respective devices $\mathcal{D}_A, \mathcal{D}_B$ to measure the state in the quantum register $R$ with independently chosen measurement settings $X_i\in\mathcal{X}, Y_i\in\mathcal{Y}$ to obtain classical outputs $A_i\in\mathcal{A}, B_i\in\mathcal{B}_i$.
We focus on protocols where the measurements inputs in testing rounds (when $T_i=1$) and generation rounds (when $T_i=0$) are the same (see discussion following Figure \ref{fig:nlg_pos_keyrate});

\paragraph*{Sifting:}
In the sifting step, Alice and Bob announce all their measurement settings to subsequently compute raw keys as given by the respective assignments $SK_A:\mathcal{X}\times\mathcal{Y}\times\mathcal{A}\rightarrow\{0,1\}$ and $SK_B:\mathcal{X}\times\mathcal{Y}\times\mathcal{B}\rightarrow\{0,1\}$.
Recall that the non-local game's winning condition was defined in Eq.~\eqref{predicate} as the Boolean $[SK_A(x,y,a)=SK_B(x,y,b)]$ which captures the intuition that the sifted keys match whenever the non local game is won.

\begin{figure}
    \centering
    \begin{algorithm}[H]\label{prot_qkd_actual}
        \SetAlgoLined
        \caption{A sequential protocol for DIQKD from non-local games: ${\mathcal{P}}$}
        \SetKwInOut{Input}{Input}
        \SetKwInOut{Parameters}{Parameters}
        \SetKwInOut{Output}{Output}
        \SetKwInOut{Init}{Initialization}
        \SetKwRepeat{Repeat}{repeat}{until}
        \Input{Devices $\mathcal{D}_{A},\mathcal{D}_{B}$}
        \Parameters{Number of rounds $n$\\
        PE parameters, $\mathbf{pe}=(\gamma, \oexp, \delta_\textup{tol})$\\
        EC parameters, $\mathbf{ec}=(\textup{Synd},\textup{Corr},\mathcal{H}_{\EC},l_{\EC})$\\
        PA parameters, $\mathbf{pa}=(\mathcal{H}_{\PA},\lkey)$}
        \textbf{1. Measurement}\;
        \For{$i=1$ to $n$}{
        Bob chooses $T_i\in\{0,1\}$ with $\Pr[T_i=1]=\gamma$\;
        Alice chooses $X_i\in\mathcal{X}$ and inputs it to $\mathcal{D}_A$ to receive $A_i\in\mathcal{A}$\;
        Bob chooses $Y_i\in\mathcal{Y}$ and inputs it to $\mathcal{D}_B$ to receive $B_i\in\mathcal{B}$\;
        }
        \textbf{2. Sifting}\;
        Bob announces $T_1^n$\;
        Alice and Bob respectively announce $X_1^n, Y_1^n$\;
        Alice sets $S_A = \{SK_A(X_i,Y_i,A_i)\}_{i=1}^n$\;
        Bob sets ${S}_B = \{SK_B(X_i,Y_i,B_i)\}_{i=1}^n$\;
        \textbf{3. Parameter Estimation}\;
        Alice sets $T_\textup{test}=\{i\;:\;T_i=1\}$ and sends $S_{A,T_\textup{test}}=\{S_{A,i}\,:\,i\in T_\textup{test}\}$ to Bob \;
        Bob sets $C_i=V(A_i,B_i,X_i,Y_i)$ for all $i\in T_\textup{test}$ and $C_i=\bot$ for all $i\in \{1,...,n\}\backslash T_\textup{test}$\;
        \eIf{$|\{i\;:\; C_i=0\}| \leq (1-\omega_\textup{exp}+\delta_\textup{tol})\cdot\gamma n$ }{
        Bob announces $F_{\PE}=\checkmark$\;
        }{
        Bob announces $F_{\PE}=\times$, both parties output $K_A=K_B=\ketbratext{\bot}{\bot}$\;
        }
        \textbf{4. Error Correction}\;
        Alice sends $Z=\text{Synd}(S_A)$, $H_A=H_{\EC}(S_A)$ and choice of $H_{\EC}\in\mathcal{H}_{\EC}$ to Bob\;
        Bob computes guess $\hat{S}_B = \text{Corr}({S}_B,Z)$ and $H_B=H_{\EC}(\hat{S}_B)$\;
        \eIf{$H_A=H_B$}{
        Bob announces $F_{\EC}=\checkmark$\;
        }{
        Bob announces $F_{\EC}=\times$, both parties output $K_A=K_B=\ketbratext{\bot}{\bot}$\;
        }
        \textbf{5. Privacy Amplification}\;
        Alice and Bob agree on $H_{\PA}\in\mathcal{H}_{\PA}$\;
        Alice outputs $K_A=H_{PA}(S_A)$ of length $l_{\key}$\;
        Bob outputs $K_B=H_{PA}(\hat{S}_B)$ of length $l_{\key}$
    \end{algorithm}
\end{figure}

\paragraph*{Parameter Estimation (PE):}
Here, Alice announces her output for the measurement rounds in which $T_i=1$ which Bob uses to test the predicate of the non-local game in order to detect a potential eavesdropper.
At this stage, Bob can abort if the number of unsuccessful instances of the non-local game is above a certain threshold from the tolerated winning probability of the game $\mathcal{G}$ due to noise $\omega_\textup{exp}$, up to certain tolerance $\delta_\textup{tol}$.
We use $\Omega_{\PE}$ to denote the event that the PE check passes.

The parameter estimation step is characterised by the tuple $\mathbf{pe} = (\gamma, \oexp, \delta_\textup{tol})$.

\paragraph*{Error Correction (EC):}
In this step, Alice and Bob apply an error correction scheme characterised by functions $\textup{Synd}:\{0,1\}^n\rightarrow \{0,1\}^{\lambda_{\EC}}$ and $\textup{Corr}:\{0,1\}^n\times\{0,1\}^{\lambda_{\EC}}\rightarrow\{0,1\}^n$.
Alice computes $\textup{Synd}: S_A \mapsto Z=\textup{Synd}(S_A)\in\{0,1\}^{\lambda_{\EC}}$ and sends the resultant syndrome to Bob, who applies $\textup{Corr}: (S_B,Z)\mapsto \hat{S}_B$ to obtain an estimate for Alice's raw key. 
The length $\lambda_{\EC}$ of the syndrome is a parameter of the error correction protocol which needs to be chosen so that for a protocol instantiated with the honest tuple $(\mathcal{D}_A^\textup{hon}, \mathcal{D}_B^\textup{hon}, \omega^{\hon}_{R_AR_B})$ winning $\mathcal{G}$ with a probability $\oexp$ which depends on the noise model, then
\begin{align}\label{EC_com_bound}
    \Pr[S_A\neq \hat{S}_B]_{\hon}\leq\varepsilon_\textup{com}^{\EC},
\end{align}
for a given completeness parameter $\varepsilon_\textup{com}^{\EC}\in(0,1]$.
In the above, the probability is over the distribution generated by the honest triple $(\mathcal{D}_A^\textup{hon}, \mathcal{D}_B^\textup{hon},\omega_{R_AR_B}^\textup{hon})$ input to the protocol.
Theoretical lower bounds for the length $\lambda_{\EC}$ achieving the condition above are explored in \cite{RR12} and \cite[Theorem 1]{TMPE17}.
In practice, for the values of $n$ and $\varepsilon_\textup{com}^{\EC}$ typically considered in QKD protocols, it suffices to take $\lambda_{\EC} = \xi n h(Q)$ for $\xi\in [1.05,1.12]$ since this achieves the condition in Eq.~\eqref{EC_com_bound} and it is what practical error correcting codes communicate ---see the discussion in \cite{TMPE17} or \cite[Footnote 10]{TL17}.
Here $h(p):=-p\log(p)-(1-p)\log(1-p)$ is the binary entropy of a Bernoulli variable, and $Q$ is known as the \emph{quantum bit error rate} (QBER) and corresponds to the probability that a bit needs to be corrected---in a QKD protocol this occurs when the $i$-th key bits do not match, i.e., when $S_{A,i}\neq \hat{S}_{B,i}$.

The hash function $H_{\EC}\in\mathcal{H}_{\EC}$ agreed upon by Alice and Bob is subsequently used to produce hashes of length $l_{\EC}$ to verify whether the error correction was successful and $S_A=\hat{S}_B$.
Indeed, the defining property of $2$-universal hash functions (\thref{def_hash_function}) asserts that 
\begin{align}
    \Pr[H_{\EC}(S_A)=H_{\EC}(\hat{S}_B)\; |\: S_A \neq \hat{S}_B] \leq \frac{1}{2^{l_{\EC}}},
\end{align}
where the probability is over the choice of the hash function ${H_{\EC}}\in\mathcal{H}_{\EC}$.
We use $\Omega_{\EC}$ to denote the event that this verification passes.

Overall, the error correction is characterised by the tuple $\mathbf{ec}=(\textup{Synd}, \textup{Corr}, \mathcal{H}_{\EC}, l_{\EC})$, where the function $\textup{Synd}$ produces a syndrome of length $\lambda_{\EC}$.
As an important remark, we note that the error correction scheme is applied to the raw key bits of \emph{all} rounds, not just generation rounds.

\paragraph*{Privacy Amplification (PA):}
Finally, Alice and Bob apply an agreed upon hash function $H_{\PA}\in\mathcal{H}_{\PA}$ to produce final keys $K_A, K_B$ of length $\lkey$.
The privacy amplification step is crucial to the secrecy of the protocol as is inferred from the \emph{Leftover Hashing Lemma} \cite{Renner05} stated later in \thref{th_leftover_hashing}.
Intuitively, the Leftover Hashing Lemma guarantees that the key $K_A$ output by Alice is close to a perfect key provided that the sifted key $S_A$ is sufficiently random from Eve's point of view, as quantified by a conditional entropy measure.

The privacy amplification is characterised by the tuple $\mathbf{pa} = (\mathcal{H}_{\PA}, \lkey)$.

Before stating our main theorems, let us discuss some important notation of the protocol.
We use $\rho = \mathcal{P}_{(\mathcal{D}_A, \mathcal{D}_B),(n,\mathbf{pe}, \mathbf{ec}, \mathbf{pa})}(\omega^0_{R_AR_BQ_E})$ to denote the state shared by Alice, Bob and Eve at the end of the protocol initialised with input devices $\mathcal{D}_A, \mathcal{D}_B$ sharing initial state $\omega^0_{R_AR_BQ_E}\in S(R_AR_BE)$, as well as parameters $\mathbf{pe}, \mathbf{ec}, \mathbf{pa}$ respectively for parameter estimation, error correction and privacy amplification.
The output state $\rho$ contains the following registers:
\begin{align}\label{prot_end_registers}
    \underbrace{A_1^nB_1^nX_1^nY_1^n}_{\text{game}} \underbrace{S_AS_B\hat{S}_B}_{\text{raw keys}} \underbrace{C_1^nT_1^nS_{A,T_\textup{test}}}_{\PE} \underbrace{  H_AH_BZH_{\EC}}_{\EC} \underbrace{H_{\PA}K_AK_B}_{\text{secret keys}} \underbrace{F_{\PE}F_{\EC}}_{\text{flags}} {Q_{E}},
\end{align}
where $Q_{E}$ is Eve's quantum side information.
Note that register $S_{A,T_\textup{test}}$ containing Alice's classical outputs in testing rounds is redundant as $S_{A,T_\textup{test}}\subset S_A$ but we include it to emphasize that it is announced and as such, Eve has access to it.
In general, Eve's total side information in an instance of protocol $\mathcal{P}_{(\mathcal{D}_A, \mathcal{D}_B),(n,\mathbf{pe}, \mathbf{ec}, \mathbf{pa})}$ is
\begin{align}
    E = X_1^nY_1^n T_1^n S_{A,T_\textup{test}}H_AZH_{\EC}H_{\PA} F_{\PE}F_{\EC} Q_{E}.
\end{align}

Further, we use $\sigma$ to denote the state prior to the privacy amplification step. 
The state $\sigma$ contains the same registers as the state $\rho$ at the end of the protocol minus $K_A$ and $K_B$.

\section{Main results}\label{sec_main_results}
\subsection{Security for general monogamy-of-entanglement games}
In this section we state and discuss our main results, which we prove later in Section \ref{sec_finite_proof}.
We remark that the only \emph{inputs} to the protocol are the physical uncharacterised devices $\mathcal{D}_A$ and $\mathcal{D}_B$ of the honest parties, whereas all variables appearing in the theorems are \emph{parameters} of the protocol, chosen to achieve appropriate security and/or maximize key rate.

\begin{theorem}[Completeness]\thlabel{th_completeness}
    Let $n\in\mathbb{N}$, $\epsilon_\textup{com}^{\PE},\epsilon_\textup{com}^{\EC}\in(0,1]$.
    For any honest triple $(\mathcal{D}_A^{\hon}, \mathcal{D}_B^{\hon},\omega^{\hon}_{R_AR_B})$ which plays the game $\mathcal{G}$ independently in every round and wins with probability $\oexp$, the protocol $\mathcal{P}_{(\mathcal{D}_A^{\hon}, \mathcal{D}_B^{\hon}),(n,\mathbf{pe}, \mathbf{ec}, \mathbf{pa})}$ with parameters $\mathbf{ec}=(\textup{Synd}, \textup{Corr}, \mathcal{H}_{\EC}, l_{\EC})$ and $\mathbf{pe}=(\gamma, \oexp, \delta_\textup{tol})$ is $(\varepsilon_\textup{com}^{\PE}+\varepsilon_\textup{com}^{\EC})$-complete if $\mathbf{ec}$ is chosen so that
    \begin{align}\label{EC_com_bound_th}
        \Pr[S_A\neq \hat{S}_B]_{\rho^{\hon}}\leq\varepsilon_\textup{com}^{\EC}, \qquad \text{where}\qquad {\rho^{\hon}} = \mathcal{P}_{(\mathcal{D}_A^{\hon}, \mathcal{D}_B^{\hon}),(n,\mathbf{pe}, \mathbf{ec}, \mathbf{pa})}(\omega^{\hon}_{R_AR_B})
    \end{align}
    and if 
    \begin{align}\label{PE_com_bound_th}
        \gamma \geq \frac{2(1-\oexp) + \dtol}{\dtol^2 n} \ln\frac{1}{\varepsilon_\textup{com}^{\PE}}.
    \end{align}
\end{theorem}

Intuitively, Eq.~\eqref{EC_com_bound_th} induces a requirement on the error correction protocol specified by the tuple $\mathbf{ec}=(\textup{Synd}, \textup{Corr}, \mathcal{H}_{\EC}, l_{\EC})$, given the honest triple $(\mathcal{D}_A^{\hon}, \mathcal{D}_B^{\hon},\omega^{\hon}_{R_AR_B})$.
The value $\oexp$ differs from the maximal two player winning probability $\otwo$ by an amount characterised by the underlying noise in the devices used to distribute and interact with the quantum states over registers $R_AR_B$. 
The overall quantum bit error rate for any noise model is defined by $Q=1-\oexp$, such that the number of errors follows a binomial distribution with probability $Q$ and overall expected number of errors $nQ$.
As such, the error correction protocol $\mathbf{ec}$ is designed to correct this random number of errors with high probability (in this case, $1-\varepsilon_\textup{com}^{\EC}$).
In practice, the choice of error correcting code is heuristic, although it is shown in \cite[Theorem 1]{TMPE17} that for any parameter $\varepsilon^{\EC}_{\textup{com}}$, there exists an error correction scheme satisfying Eq.~\eqref{EC_com_bound_th}.

On the other hand, Eq.~\eqref{PE_com_bound_th} is a requirement on the parameter estimation stage.
In particular, for the triple $(\mathcal{D}_A^{\hon}, \mathcal{D}_B^{\hon},\omega^{\hon}_{R_AR_B})$ as above, we require the testing probability $\gamma$ and tolerance $\delta_\textup{tol}$ to be chosen so that the number of losses in an instance of $\mathcal{P}$ does not fall far above the expected number of loses $(1-\oexp)n\gamma$ except with small probability.
In essence, Eq.~\eqref{PE_com_bound_th} is derived from a `tail bound' on the probability that the tolerated number of loses deviates from the expected number of losses, and $\delta_\textup{tol}$ is chosen to make this probability small (by $\varepsilon_\textup{com}^{\PE}$). 
Given the value of $\gamma$ in the theorem, the value of $\dtol$ is chosen so that $\dtol$ and $\gamma$ vanish as $n$ approaches infinity. 
In practice, we set $\dtol={n^{-1/3}}$ so that asymptotically, $\dtol,\gamma\in O(n^{-1/3})$.

\begin{theorem}[Correctness]\thlabel{th_correctness}
    Let $n\in\mathbb{N}$ and $\epscorr\in(0,1]$.
    For any triple $(\mathcal{D}_A, \mathcal{D}_B,\omega^0_{R_AR_BQ_E})$, the protocol $\mathcal{P}_{(\mathcal{D}_A, \mathcal{D}_B),(n,\mathbf{pe}, \mathbf{ec}, \mathbf{pa})}$ with $\mathbf{ec}=(\textup{Synd}, \textup{Corr}, \mathcal{H}_{\EC}, l_{\EC})$ is $\epscorr$-correct if $l_{\EC}\geq \log\frac{1}{\epscorr}$.
\end{theorem}

In our protocol, correctness is ensured by comparing random hashes of Alice's raw key $S_A$ and Bob's guess $\hat{S}_B$.
Thus, this theorem does not need to refer to the actual error correction protocol $(\textup{Synd},\textup{Corr})$.

\begin{theorem}[Secrecy]\thlabel{th_secrecy}
    Let $n\in\mathbb{N}$ and $\epssec\in(0,1]$.
    For any triple $(\mathcal{D}_A, \mathcal{D}_B,\omega^0_{R_AR_BQ_E})$, protocol $\mathcal{P}_{(\mathcal{D}_A, \mathcal{D}_B),(n,\mathbf{pe}, \mathbf{ec}, \mathbf{pa})}$ with $\mathbf{pe}=(\gamma, \oexp, \delta_\textup{tol})$, $\mathbf{ec}=(\textup{Synd}, \textup{Corr}, \mathcal{H}_{\EC}, l_{\EC})$ and $\mathbf{pa} = (\mathcal{H}_{\PA}, \lkey)$ is $\epssec$-secret if
    \begin{align}\label{l_key}
        l_\textup{key} \leq \;&n{g}(\omega_\textup{exp}-\delta_\textup{tol}) -d_1\sqrt{n}-d_0 - (2-\oexp+\delta_\textup{tol})n(\gamma+\kappa)   - 2\vartheta\left(\frac{\varepsilon_s}{4}\right)\\
        &-l_{\EC}-\lambda_{\EC} 
        - 2\log\frac{1}{\epssec-2\varepsilon_s},
    \end{align}
    where $\varepsilon_s\in\big(0,\frac{1}{2}\epssec\big)$ is a variable to be optimised, $\kappa\geq\frac{1}{\sqrt{n}}\sqrt{\ln{8}-\ln{\varepsilon_s}}$, $g:[0,1]\rightarrow\mathbb{R}$ is the affine function defined by
    \begin{align}\label{affine_th}
        {g}(p):=\frac{p-\beta}{\ln2\cdot(1-\beta+\omega_3)} - \log\big(1-\beta+\omega_3\big)
    \end{align}
    for $\beta\in[\omega_3, \omega_2]$ to be optimised, and the coefficients $d_1, d_0$ depend only on $\epssec$, $\varepsilon_s$ and $\beta$ and are respectively as in Eqs.~\eqref{d1_def_final} and~\eqref{d0_def_final}.
    Lastly, $\vartheta(\delta) = -\log\big(1-\sqrt{1-\delta^2}\big)$.
\end{theorem}

In particular, the protocol is $(\epscorr + \epssec)$-sound, provided $\mathbf{pe}, \mathbf{ec}, \mathbf{pa}$ are chosen as per \thref{th_correctness} and \thref{th_secrecy}.

We now discuss the variables appearing in the these theorems.
The variable $\varepsilon_s$ corresponds to a bound on the probability that the error correction is successful (i.e. that $S_A=\hat{S}_B$) provided that the parameter estimation check is successful.
The variable $\kappa$ is necessary to ensure the number of testing rounds does not exceed $n(\gamma+\kappa)$ except with negligible probability.
The variables $\varepsilon_s\in(0,\frac{1}{2}\epssec)$ and $\beta\in[\omega_3,\omega_2]$ are to be optimised to maximize the key rate $\frac{1}{n} \lkey$.

On the other hand, the function $g$ appearing in Eq.~\eqref{l_key} intuitively represents the amount of randomness that is secret from Eve which Alice is able to generate in every round of the protocol.
The function $g$ presented in Eq.~\eqref{affine_th} is derived solely from the fact that the non-local game satisfies the monogamy-of-entanglement property $\omega_3<\omega_2$ (see Section \ref{sec_guessing_probability}).
In principle, it is possible to tailor $g$ to specific non-local games to achieve better key rates without altering the underlying secrecy proof. 
For example, see Section \ref{sec_msg_keyrates} for improved key rates for the magic square game, and Section \ref{sec_msg} for the respective constructions.

\subsubsection{Asymptotic analysis}

In the following, we briefly analyse the asymptotic behaviour of $l_\textup{key}$ in Eq.~\eqref{l_key} for an honest implementation of the protocol.

Firstly, we note that in an honest implementation the $n$ measurement rounds are independently and identically distributed (i.i.d.) and that the devices utilise the optimal quantum strategy for $\mathcal{G}_2$.
Formally, let $\mathcal{S}=\{\nu_{Q_AQ_B}, \{\{\Pax\}_a\}_x, \{\{\Qby\}_b\}_y\}$ be the strategy that maximizes $\omega_2$, where $\nu_{Q_AQ_B}\in S(Q_AQ_B)$ is Alice's and Bob's bipartite state and $\{\Pax\}_a, \{\Qby\}_b$ are POVMs on $Q_A, Q_B$ with outputs over sets $\mathcal{A}, \mathcal{B}$ and inputs over sets $\mathcal{X}, \mathcal{Y}$, respectively.
The state $\nu_{Q_AQ_B}$ may be subject to noise which in QKD is usually modelled using the depolarizing channel
\begin{align}
    \mathcal{E}_q:\nu_{Q_AQ_B}\mapsto \tilde{\nu}_{Q_AQ_B} = (1-2q)\nu_{Q_AQ_B} +2q \tau_{Q_AQ_B},
\end{align}
where $q\in[0,\frac{1}{2}]$ and $\tau_{Q_AQ_B}=\frac{\mathbb{I}_{Q_AQ_B}}{\dim Q_A \dim Q_B}$ is the maximally mixed state in $Q_AQ_B$.
Therefore, in an honest implementation of the QKD protocol the devices $\mathcal{D}_A^{\hon}, \mathcal{D}_B^{\hon}$ share the i.i.d. state $\omega_{R_AR_B}^{\hon}=\tilde{\nu}_{Q_AQ_B}^{\otimes n}$.
In the following, we omit the round indexes since the behaviour is i.i.d..
Then, on input of $X$ by Alice and $Y$ by Bob, the respective devices $\mathcal{D}_A^{\hon}$ and $\mathcal{D}_B^{\hon}$ apply the local measurements given by 
\begin{align}
    \mathcal{M}_{A}^{x,\textup{hon}} \otimes \mathcal{M}_{B}^{y,\textup{hon}} : \tilde{\nu}_{Q_AQ_B} \mapsto \sum_{a\in\mathcal{A}}\sum_{b\in\mathcal{B}} \ketbra{ab}_{AB}\otimes \Tr[\tilde{\nu}_{Q_AQ_B}(\Pax\otimes \Qby)].
\end{align}
It follows that the expected winning probability of Alice and Bob is precisely
\begin{align}
    \omega_\textup{exp} = \sum_{x,y,a,b} \pi(x,y) V(a,b,x,y)\Tr[\tilde{\nu}_{Q_AQ_B}(\Pax\otimes \Qby)],
\end{align}
where $V(a,b,x,y)$ is the non-local game's predicate as in Eq.~\eqref{predicate}.
The quantum bit error rate is therefore given by $Q:=1-\omega_\textup{exp}$.
To determine the length of the error correction syndrome $\lambda_{\EC}$, note that \cite[Theorem 1]{TMPE17} implies the existence of error correction protocol $\mathbf{ec}=(\textup{Synd}, \textup{Corr}, \mathcal{H}_{\EC},l_{\EC})$ which asymptotically satisfies\footnote{Here, we use $g(n)\in\Theta(f(n))$ to denote both $g(n)\in O(f(n))$ and $g(n)\in\Omega(f(n)).$} $\frac{\lambda_\textup{EC}}{n} = h(Q) + \Theta\left(\frac{\log n}{n}\right)$ and achieves $\Pr\big[S_A\neq \hat{S}_B\big]_{\rho^{\hon}}\leq\varepsilon_\textup{com}^{\EC}$.
Further, note that choosing 
\begin{align}\label{gamma_delta}
    \dtol = \frac{1}{n^{1/3}} \qquad\textup{and}\qquad \gamma = \frac{2(1-\oexp) + \dtol}{\dtol^2 n} \ln\frac{1}{\varepsilon_\textup{com}^{\PE}}
\end{align}
as per \thref{th_completeness}, we asymptotically obtain $\dtol, \gamma\in O\big({n^{-1/3}}\big)$.

With this choice of $\lambda_{\EC}, \gamma$ and $\delta_\textup{tol}$ we have $(\varepsilon_\textup{com}^{\EC}+\varepsilon_\textup{com}^{\PE})$-completeness, and from Eq.~\eqref{l_key} we arrive at
\begin{align}\label{asymptotic_keyrate}
    \lim_{n\rightarrow\infty}\frac{l_\textup{key}}{n} = {g(\omega_\textup{exp})} - h(Q).
\end{align}
The first term encapsulates the amount of Alice's key that is secret from Eve, whereas the second relates to the cost of information reconciliation between Alice and Bob.
We remark that this asymptotic key rate retrieves the \emph{Devetak-Winter rate} \cite{DW05} if $g$ is chosen to be a sufficiently tight bound on the von Neumann entropy.

In Figure \ref{fig:nlg_pos_keyrate} we compute the range of values of $\omega_2$ and $\omega_3$ for which Eq.~\eqref{asymptotic_keyrate} using $g$ as in Eq.~\eqref{affine_th} gives positive asymptotic key rates.
In particular, we highlight two games: the \emph{magic square game} (MSG) which we consider in the next section with better optimized functions $g$, and the \emph{Clauser-Horne-Simone-Holt} (CHSH) game which most DIQKD protocols use as primitive.
In the latter game, Alice and Bob produce binary inputs $x,y$ and generate binary outputs $a,b$, and win or lose if $a=b\oplus(x\cdot y)$.
In the tripartite extension, the third player receives both $x,y$ and produces bit $c$. 
The three parties win if $a=b\oplus(x\cdot y)=c$.
In this case it is known that $\omega_2=\frac{2+\sqrt{2}}{4}\approx 0.85$ and it can be shown that $\omega_3\approx\frac{3}{4}$ (see \cite{Cervero23}).
Using our framework, the difference between $\omega_2$ and $\omega_3$ in the CHSH game is not sufficient to offset the cost of error correction.
Indeed, in the noiseless case practical codes require $1.1h(1-\frac{2+\sqrt{2}}{4})\approx 0.661$ bits of error correction per round of the protocol.
To circumvent this problem, it is possible to use different measurement settings in testing and generation rounds in such a way that the generation settings produce highly correlated outcomes for Alice and Bob, thereby reducing the cost of error correction---see for example the framework in \cite{AMP06, PAB+09, AFRV19, TSB+22}.
We do not explore this route here as it has already been extensively analysed in the aforementioned works, but remark that the functions which bound the amount of randomness of Alice which is secret from Eve obtained in the works of \cite{AFRV19, TSB+22} may directly replace the function $g$ in Eq.~\eqref{l_key} of \thref{th_secrecy} to yield equivalent secrecy guarantees and positive key rates for the CHSH game.

\begin{figure}[h]
    \centering
    \includegraphics[width=0.6\textwidth]{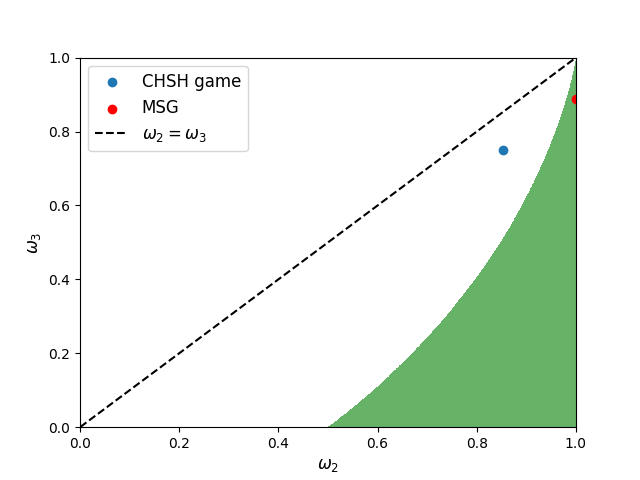}
    \caption{Shaded green area corresponds to the range of bipartite winning probabilities $\omega_2$ and tripartite winning probabilities $\omega_3$ from which Eq.~\eqref{asymptotic_keyrate} yields a positive asymptotic key rate.} 
    \label{fig:nlg_pos_keyrate}
\end{figure}

\subsection{Application with the magic square game}\label{sec_msg_keyrates}
In this section, we apply our NLG-DIQKD protocol to the magic square game which we briefly describe in the following.

The two party magic square game $\msgt$ proceeds as follows:
\begin{itemize}
    \item Alice receives $x\in \{0,1,2\}$ with uniform probability and outputs $a \in\{0,1\}^3$ s.t. $\bigoplus_k a[k] = 0$.
    \item Bob receives $y\in \{0,1,2\}$ with uniform probability and outputs $b \in \{0,1\}^3$ s.t. $\bigoplus_k b[k] =1$.
\end{itemize}
Alice and Bob win if $a[y]=b[x]$, where $a[k]$ (respectively $b[k]$) denotes the $k$-th bit in $a$ ($b$).

It is known \cite{Mermin90, Peres90} that $\omega_\textup{C}(\msgt)=\frac{8}{9}$ and $\omega_\textup{Q}(\msgt)=1$. 
The best classical strategy consists of preparing a $3\times 3$ grid of bits such that the rows satisfy Alice's parity condition and the columns satisfy Bob's parity condition. 
Any $3\times3$ grid filled in this way always has one square which is impossible to fill, so Alice and Bob always lose when their inputs $x,y$ correspond to that square.
The quantum strategy achieving $\omega_\textup{Q}(\msgt)=1$ requires the honest parties to share $\kettext{\Phi^+}_{Q_{A_1}Q_{B_1}}\otimes \kettext{\Phi^+}_{Q_{A_2}Q_{B_2}}$ ---that is, two copies of the maximally entangled Bell state $\kettext{\Phi^+}=\frac{1}{\sqrt{2}}(\kettext{00}+\kettext{11})$--- and perform the following observables on their respective halves:

\begin{align}
\label{msg2_measurements}
\begin{tabular}{  m{1.5cm} | m{1.5cm} | m{1.5cm} | m{1.5cm} } 
    &  
    \begin{center} $y=0$ \end{center} &  
    \begin{center} $y=1$ \end{center} &  
    \begin{center} $y=2$ \end{center} \\  
  \hline
    \begin{center} $x=0$ \end{center} &
    \begin{center} $\mathbb{I}\otimes Z$ \end{center} &  
    \begin{center} $Z\otimes \mathbb{I}$ \end{center} &  
    \begin{center} $Z\otimes Z$ \end{center} \\ 
  \hline
    \begin{center} $x=1$ \end{center} &
    \begin{center} $X\otimes \mathbb{I}$ \end{center} &  
    \begin{center} $\mathbb{I}\otimes X$ \end{center} &  
    \begin{center} $X\otimes X$ \end{center} \\ 
  \hline
    \begin{center} $x=2$ \end{center} &
    \begin{center} $-X\otimes Z$ \end{center} &  
    \begin{center} $-Z\otimes X$ \end{center} &  
    \begin{center} $Y\otimes Y$ \end{center} \\ 
  \hline
\end{tabular}
\end{align}

Namely, on input of $x$ for Alice (respectively, $y$ for Bob), she applies the three commuting observables in the $x$-th row ($y$-th column) on her halves of $\kettext{\Phi^+}_{Q_{A_1}Q_{B_1}}\otimes \kettext{\Phi^+}_{Q_{A_2}Q_{B_2}}$.
Since the overlap of the $x$-th row and $y$-th column contains the same observables for Alice and Bob, they are guaranteed to obtain the same bit.
Further, since the product of each row (respectively, each column) equals $\mathbb{I}\otimes\mathbb{I}$ (respectively $-\mathbb{I}\otimes\mathbb{I}$), the parity conditions are satisfied.

In the tripartite extension $\msg$, the third player receives both inputs $x,y\in\{0,1,2\}$ and produces the single bit $c$.
The new predicate checks whether $a[y]=c$ in addition to $a[y]=b[x]$.
It is easy to see that $\omega_\textup{C}(\msg)=\frac{8}{9}$ and in Section \ref{sec_msg} we use numerical optimization tools to show\footnote{In addition, we use the linear program specified in \cite[Section 3]{Toner08} to show $\omega_\textup{NS}(\msg)=1$.} $\omega_\textup{Q}(\msg)\lesssim\frac{8.00077}{9}$.

Now, let us discuss the honest behaviour of MSG-DIQKD to compute finite and asymptotic key rates.
In an honest implementation, the devices $\mathcal{D}_A^{\hon}, \mathcal{D}_B^{\hon}$ of Alice and Bob behave in an i.i.d.\ way and adhere to the optimal strategy of the bipartite magic square game: $\mathcal{S}=\{\nu_{Q_AQ_B}, \{\{\Pax\}_a\}_x, \{\{\Qby\}_b\}_y\}$.
In this strategy $\nu_{Q_AQ_B}=\ketbratext{\Phi^+}{\Phi^+}_{Q_{A_1}Q_{B_1}}\otimes \ketbratext{\Phi^+}{\Phi^+}_{Q_{A_2}Q_{B_2}}$ and $\{\Pax\}_{a\in\mathcal{A}}$, $\{\Qby\}_{b\in\mathcal{B}}$ for $\mathcal{A}=\{000,011,101,110\}$, $\mathcal{B}=\{001,010,100,111\}$, are the POVMs associated to the observables in Eq.~\eqref{msg2_measurements}.

For the noise model, we assume that each of the two Bell states in $\nu_{Q_AQ_B}$ suffers from identical depolarising noise, resulting in the state
\begin{align}\label{msg_init_state}
    \tilde{\nu}_{Q_AQ_B} = \big((1-2q)\ketbra{\Phi^+}+2q\tau\big)_{Q_{A_1}Q_{B_1}} \otimes \big((1-2q)\ketbra{\Phi^+}+2q\tau\big)_{Q_{A_2}Q_{B_2}}, 
\end{align}
where $q\in[0,\frac{1}{2}]$ and $\tau_{Q_{A_1}Q_{B_1}}, \tau_{Q_{A_2}Q_{B_2}}$ are the maximally mixed states of dimension four.
The devices $\mathcal{D}_A^{\hon}, \mathcal{D}_B^{\hon}$ therefore share the state $\omega_{R_AR_B}^{\hon}=\tilde{\nu}_{Q_AQ_B}^{\otimes n}$.
Standard calculations then show that overall QBER and expected winning probability for the bipartite magic square game are respectively 
\begin{align}\label{msg_qber}
    Q=\frac{2}{9}q(7-5q) \qquad \textup{and}\qquad \omega_\textup{exp} = 1-\frac{2}{9}q(7-5q).
\end{align}
Formally, this is because the state $\tilde{\nu}_{Q_AQ_B}$ is a mixture containing the state $\ketbra{\Phi^+}^{\otimes 2}$ with probability $(1-2q)^2$, the state $\tau^{\otimes 2}$ with probability $4q^2$, and the states $\ketbra{\Phi^+}\otimes\tau$ and $\tau\otimes\ketbra{\Phi^+}$ with probability $2q(1-2q)$ each.
It's clear that the probabilities of winning when sharing $\ketbra{\Phi^+}^{\otimes 2}$ and $\tau^{\otimes 2}$ are $1$ and $\frac{1}{2}$, respectively. 
On the other hand, when sharing the state $\ketbra{\Phi^+}\otimes\tau$, it is only possible to win with probability $1$ when $(x,y)\in\{(0,1),(1,0)\}$ (corresponding to product observables $Z\otimes \mathbb{I}$ and $\mathbb{I}\otimes X$), which occurs with probability $\frac{2}{9}$ over the choice of inputs.
All remaining inputs result in a winning probability of $\frac{1}{2}$.
Similarly, when sharing the sate $\tau\otimes\ketbra{\Phi^+}$, inputs $(x,y)\in\{(0,0),(1,1)\}$ yield winning probability $1$, whereas all other inputs yield $\frac{1}{2}$.
Overall, linearity of the trace over this mixture of states gives
\begin{align}
    \omega_\textup{exp} 
    = (1-2q)^2\cdot 1 + 2\cdot 2q(1-2q)\bigg(\frac{7}{9}\cdot\frac{1}{2}+\frac{2}{9}\cdot1\bigg) + (2q)^2\cdot\frac{1}{2}
    = 1-\frac{2}{9}q(7-5q).
\end{align}

Now, we apply \thref{th_completeness}, \thref{th_correctness} and \thref{th_secrecy} to compute the achievable finite sized key rates for the MSG-DIQKD protocol under this noise model and with optimized affine functions $g$ derived in Section \ref{sec_msg}.
Results are shown in Figure \ref{fig:MSG-DIQKD_keyrates}.

\begin{figure}[h]
    \centering
    \includegraphics[width=0.6\textwidth]{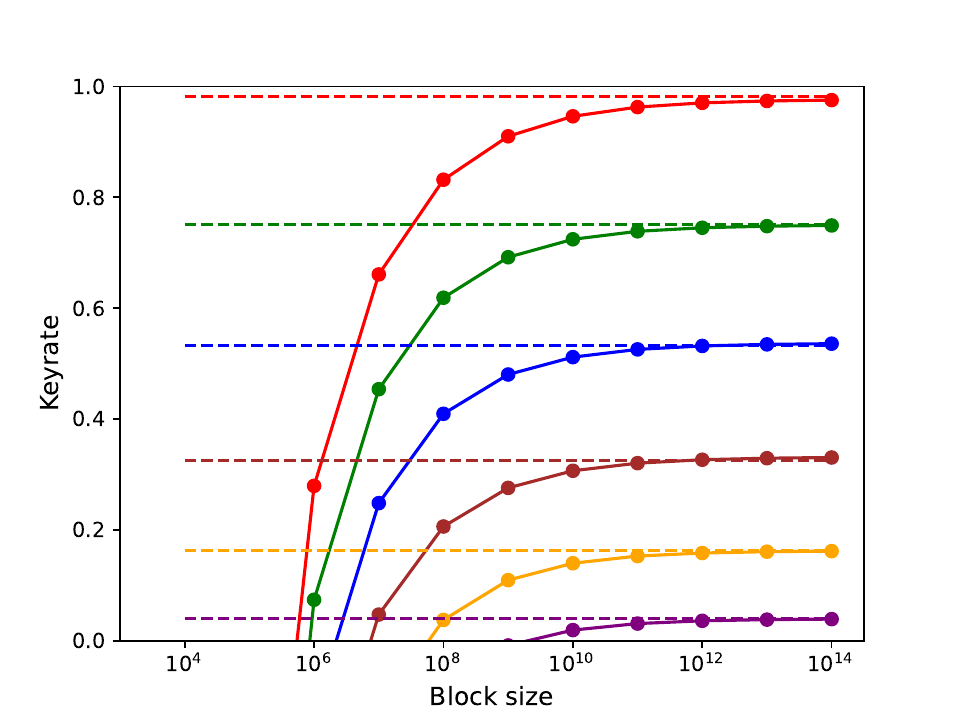}
    \caption{Keyrates for MSG-DIQKD protocol as a function of block size (i.e. number of rounds) under depolarizing noise. 
    The curves correspond to depolarizing noise for $q \in \{ 0, 0.005, 0.01, 0.015, 0.02, 0.025\}$ (red, green, blue, brown, orange and purple, respectively), testing probability $\gamma$ and noise tolerance $\delta_\textup{tol}$ chosen as per Eq.~\eqref{gamma_delta}. 
    The key rates for $q \in \{0, 0.005, 0.01, 0.015\}$ (respectively $q \in \{0.02, 0.025\}$) were obtained using the min-tradeoff functions computed with the von Neumann entropy bounds of Eq.~\eqref{msg_vn_lower_bound} (respectively min-entropy bounds of Eqs.~\eqref{affine_constrained},\eqref{msg_prob_upper_bound}).
    The parameter $\lambda_\textup{EC}$ is set to {$1.1h(Q)$} and the security parameters are chosen as $\epscom = 10^{-2}$ and $\epssou = 10^{-6}$.
    The dashed horizontal lines correspond to the asymptotic behaviour given by Eq.~\eqref{asymptotic_keyrate}.
    See \cite{Cervero23} for Python script used to generate this figure.}
    \label{fig:MSG-DIQKD_keyrates}
\end{figure}

Following the finite key rates, we turn our attention to the asymptotic key rates of the MSG-DIQKD protocol.
For the remainder of this section, we ignore round indexes and use $X, Y, A, B$ to respectively denote the inputs and the outputs of Alice and Bob in an instance of the MSG. 
Further, we use $S_A$ and $S_B$ for the singular raw key bits generated from $A$ and $Y$, and $B$ and $X$, respectively.
Since Alice's key bits are generated using arbitrary inputs of both Alice and Bob, and the error correction in the protocol is one-way, the asymptotic key rate of MSG-DIQKD is given by the Devetak-Winter bound \cite{DW05}
\begin{align}\label{DW_rate}
    r\geq H(S_A|XYE)_{\hat{\nu}}-H(S_A|XYS_B)_{\hat{\nu}},
\end{align}
where $\hat{\nu}\in S(S_AS_BXYE)$ is the post-measurement classical-quantum state shared by Alice, Bob and the eavesdropper Eve after one instance of the magic square game with the initial quantum state $\tilde{\nu}\in S(Q_AQ_BE)$ which extends the state in Eq.~\eqref{msg_init_state}.
Further, $H(\cdot|\cdot)$ is the conditional von Neumann entropy which for a bipartite state $\rho\in S(AB)$ is defined as $H(A|B)_\rho:=H(AB)_\rho - H(B)_\rho$, for $H(A)_\rho:=-\Tr[\rho\log\rho]$.
In the Devetak-Winter rate, the first entropy term intuitively corresponds to the amount of Alice's randomness that is secret from Eve, whereas the second entropy term encapsulates the cost of error correction for Alice and Bob and equals the binary entropy of the QBER.

We present two bounds to the Devetak-Winter rate in Figure \ref{fig:MSG_asymptotics}. 
In Figure \ref{fig:DW_min_ent} we use a bound on the min-entropy (defined in Section \ref{sec_prelims_2}) to bound the first term in the Devetak-Winter rate.
On the other hand, in Figure \ref{fig:DW_vn_ent} we use the numerical bounds on the von Neumann entropy $H(S_A|X\in\{0,1,2\},Y=0;E)_{\hat{\nu}}$ specified in Eq.~\eqref{MSG_di_vn_smol}\footnote{The reduced choice of measurement settings in the von Neumann entropy is done to make the underlying optimization problem tractable.
We remark that $H(S_A|X\in\{0,1,2\},Y=0;E)_{\hat{\nu}}\leq H(S_A|XYE)_{\hat{\nu}}$.}.
Further, in Figure \ref{fig:MSG_asymptotics_noise} we compare our asymptotic rates with those reported in the recent paper \cite{ZMZ+23}.

\begin{figure}[h]
    \centering
    \begin{subfigure}{0.45\textwidth}
        \includegraphics[width=\textwidth]{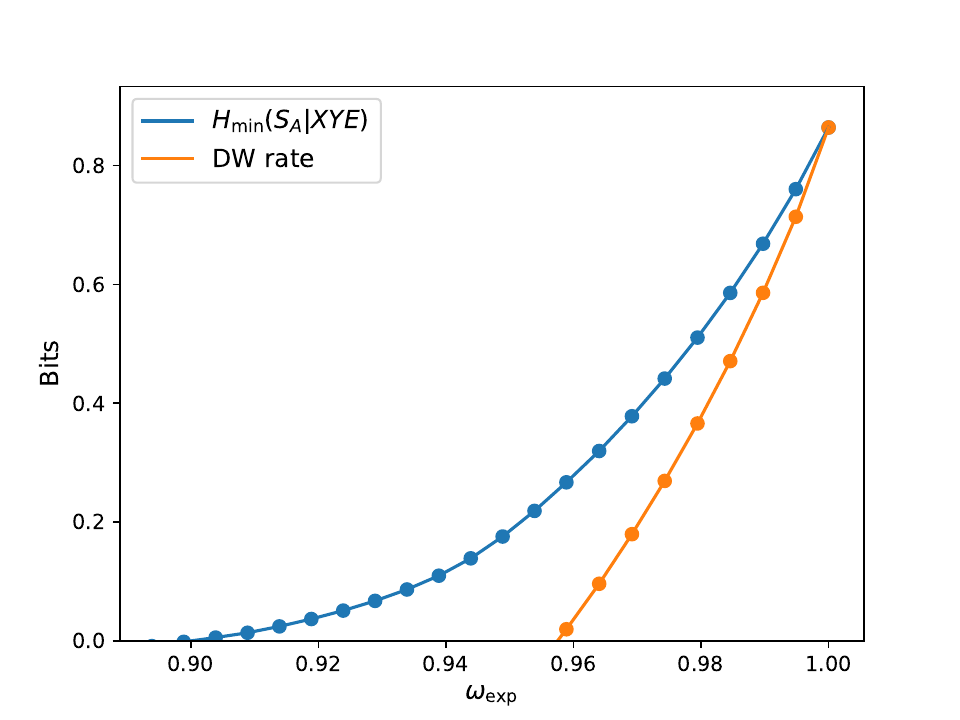}
        \caption{Bounds with min-entropy.}
        \label{fig:DW_min_ent}
    \end{subfigure}
    \begin{subfigure}{0.45\textwidth}
        \includegraphics[width=\textwidth]{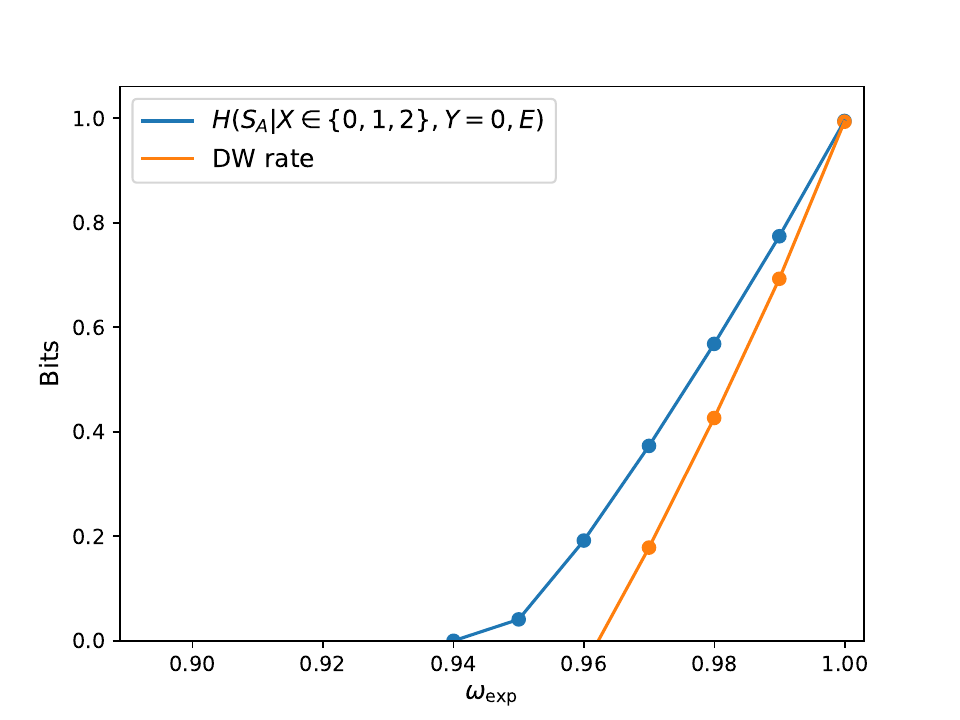}
        \caption{Bounds with von Neumann entropy.}
        \label{fig:DW_vn_ent}
    \end{subfigure}
    \caption{Device independent bounds for MSG-DIQKD. Blue and orange lines respectively correspond to bounds on the conditional entropy and the asymptotic key rate.
    See \cite{Cervero23} for Python script used to generate this figure.}
    \label{fig:MSG_asymptotics}
\end{figure}

\begin{figure}
    \centering
    \includegraphics[width=0.6\textwidth]{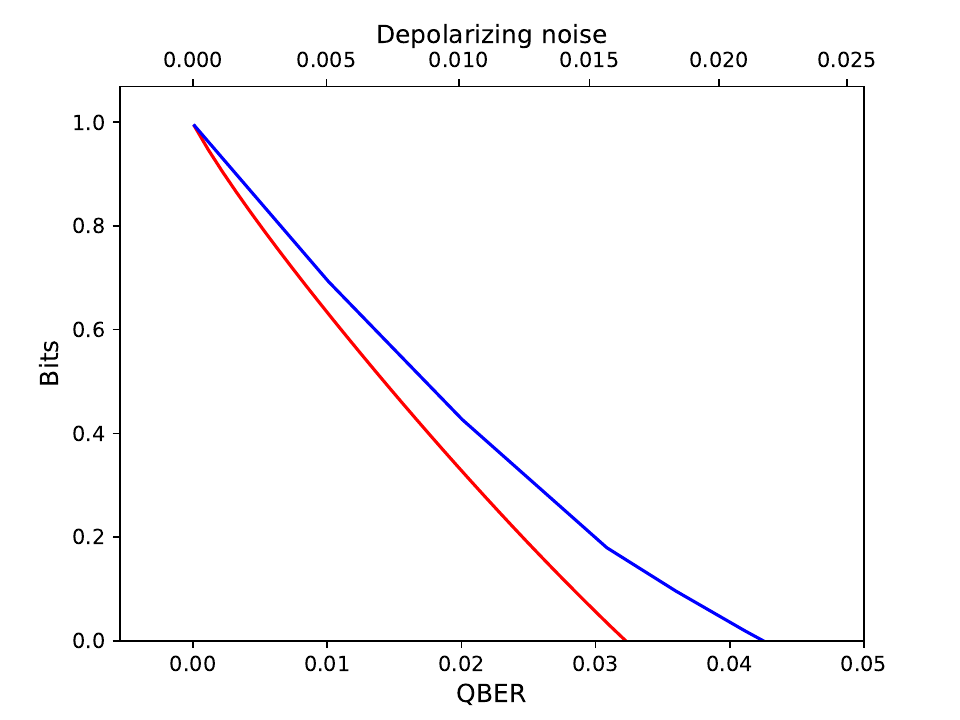}
    \caption{Comparison of Devetak-Winter rates as a function of QBER (bottom horizontal axis) and its corresponding depolarizing noise (top horizontal axis).
    Blue line is computed using $\max\{H_{\min}(S|XYE), H(S|X\in\{0,1,2\},Y=0;E)\}$ and red line is computed with $H(S|X=0,Y=0;E)$ and matches the key rates reported in \cite{ZMZ+23} (see footnote \ref{footnote}).}
    \label{fig:MSG_asymptotics_noise}
\end{figure}

Note that in these figures the min-entropy bound yields better key rates in the regime with higher quantum bit error rate, whereas the von Neumann bounds perform better in the regime where $\omega_\textup{exp}\gtrsim 0.985$.
This is purely a consequence of the reduced choice of measurement inputs in the optimization of Eq.~\eqref{MSG_di_vn_smol} due to computational complexity of the optimization problem ---indeed it is expected that $H(S|XYE)\geq H_{\min}(S|XYE)$.

The asymptotic rates using the min-entropy bounds from Section \ref{sec_msg_min_ent} suggest that positive key rate is only possible for quantum bit error rate lower than approximately $0.044$, which corresponds to depolarizing noise of about $2.88\%$ as per Eq.~\eqref{msg_qber}.
The recent work \cite[Figure 1 in Supp. Material]{ZMZ+23} applies identical techniques to lower bound the entropy $H(S|X=0,Y=0;E)$, which they use to claim security of MSG-DIQKD against collective attacks in the asymptotic regime with robustness against depolarizing noise of about $1.6\%$\footnote{\label{footnote}We remark that the constraint originally used in \cite{ZMZ+23} to bound $H(S|X=0,Y=0;E)$ correspond to Alice and Bob's winning probability take away Alice and Bob's losing probability.
Hence, the original claimed results of \cite[Figure 1 in Supp. Material]{ZMZ+23} were horizontally scaled by a factor of $2x -1$.
We calculate the noise robustness of $1.6\%$ of the protocol in \cite{ZMZ+23} taking this error into consideration.
We have notified the authors and these mistakes have been corrected in the latest version of \cite{ZMZ+23}.
}.

\section{Finite key security proof}\label{sec_finite_proof}

Before delving into the proofs of completeness, correctness and secrecy of $\mathcal{P}$ (algorithm \ref{prot_qkd_actual}) we will need various definitions and prior results, presented in the following Preliminaries section.

\subsection{Preliminaries}\label{sec_prelims_2}

We will require some measures of closeness between different states. 
The \emph{fidelity} between quantum states $\rho,\sigma\in S(A)$ is given by
\begin{align}
    F(\rho,\sigma) := \Big(\Tr\left|\sqrt{\rho}\sqrt{\sigma}\right|\Big)^2.
\end{align}
On the other hand, the \emph{generalised fidelity} for subnormalised states $\rho,\sigma\in S_{\leq}(A)$ is given by
\begin{align}
    F_{*}(\rho,\sigma) := \Big(\Tr\big|\sqrt{\rho}\sqrt{\sigma}\big| + \sqrt{(1-\Tr\rho)(1-\Tr\sigma)}\Big)^2.
\end{align}
From the generalised fidelity we define the \emph{purified distance} 
\begin{align}
    P(\rho,\sigma):= {\sqrt{1-F_{*}(\rho,\sigma)}}.
\end{align}

Following the formalism from \cite{Tomamichel16} we also introduce various entropies.
For a subnormalized state $\rho_{AB}\in S_\leq(AB)$, the \emph{min-entropy} and the \emph{max-entropy} of $A$ given $B$ are
\begin{align}
    H_{\min}(A|B)_\rho &:= \max_{\substack{\sigma_B\in S_\leq(B): \\ \ker(\sigma_B)\subseteq\ker(\rho_B)}}-\log\norm{\Big(\mathbb{I}_A\otimes\sigma_B^{-\frac{1}{2}}\Big)\rho_{AB}\Big(\mathbb{I}_A\otimes\sigma_B^{-\frac{1}{2}}\Big)}_\infty\\
    H_{\max}(A|B)_\rho &:= \max_{\sigma_B\in S_\leq(B)}\log F(\rho_{AB},\mathbb{I}_A\otimes\sigma_B).
\end{align}
In the case $\rho_{XE}=\sum_{x\in\mathcal{X}} p_x\ketbratext{x}{x}\otimes\rho_E^x$ is a classical-quantum state in $S(XE)$ then 
\begin{align}
    H_{\min}(X|E)_\rho = -\log p_\textup{guess}(X|E)_\rho,
\end{align}
where $p_\textup{guess}(X|E)_\rho:=\sup_{\{M_x\}_x}\sum_{x} p_x \Tr[\rho_{E}^x M_x]$ is the guessing probability and the supremum is over POVMs on register $E$.

We also define the \emph{smooth-min entropy} and the \emph{smooth-max entropy} of $A$ given $B$ for a subnormalized state $\rho_{AB}\in S_\leq(AB)$ and a smoothing parameter $\varepsilon\in[0,\sqrt{\Tr\rho_{AB}}]$: 
\begin{align}
    H_{\min}^\varepsilon(A|B)_\rho &:= \max_{\substack{\tilde{\rho}_{AB}\in S_{\leq}(AB):\\ P(\rho_{AB},\tilde{\rho}_{AB})\leq\varepsilon}} H_{\min}(A|B)_{\tilde{\rho}}\\
    H_{\max}^\varepsilon(A|B)_\rho &:= \min_{\substack{\tilde{\rho}_{AB}\in S_{\leq}(AB):\\ P(\rho_{AB},\tilde{\rho}_{AB})\leq\varepsilon}} H_{\max}(A|B)_{\tilde{\rho}}.
\end{align}
The smooth min-entropy is a generalisation of the min-entropy which takes into account states that are close to $\rho_{AB}$ in purified distance.
Likewise for the smooth max-entropy.
Both of these smooth entropies satisfy the following \emph{data processing inequality} \cite[Section 6.3.3]{Tomamichel16}:
for $\mathcal{E}\in\CPTP(B,B')$ and $\tilde{\rho}=(\mathbb{I}_A\otimes\mathcal{E})(\rho_{AB})$, then
\begin{align}
    H^\varepsilon_{\min}(A|B)_\rho\leq H^\varepsilon_{\min}(A|B')_{\tilde{\rho}}, \qquad\textup{and}\qquad H^\varepsilon_{\max}(A|B)_\rho\leq H^\varepsilon_{\max}(A|B')_{\tilde{\rho}}.
\end{align}

\subsubsection{Min-tradeoff functions}\label{sec_min_tradeoff}
Min-tradeoff functions bound the amount of randomness generated by a party, Alice, that is secret from an eavesdropper during any round in a protocol comprised of $n$ measurement rounds.

Formally, let $\mathcal{M}\in\CPTP(RE,SCRE)$ be a quantum instrument which updates quantum registers $R$ and $E$ and produces classical information $S,C$ over alphabets $\mathcal{S,C}$.
Let $\mathcal{P}(\mathcal{C})$ denote the set of probability distributions over alphabet $\mathcal{C}$ and let $F\cong RE$.

\begin{definition}[Min-tradeoff function]\thlabel{def_min-tradeoff}
    A function $f:\mathcal{P}(\mathcal{C})\rightarrow\mathbb{R}$ is a min-tradeoff function for the channel $\mathcal{M} \in \CPTP(RE, SCRE)$ if for any $q\in\mathcal{P}(\mathcal{C})$
    \begin{align}
        f(q)\leq \min_{\omega_{REF}} H(S|EF)_{\mathcal{M}(\omega_{REF})},
    \end{align}
    where the minimum is over states $\omega_{REF}\in S(REF)$ such that $\Tr_{SREF}[\mathcal{M}(\omega_{REF})]=\sum_c q_c\ketbratext{c}{c}_C$, that is over output states whose $C$ register is distributed as $q\in\mathcal{P}(\mathcal{C}).$
\end{definition}

Now, consider instead the sequence of channels $\{\mathcal{M}_i \in \CPTP(RE, S_iC_iRE)\}_{i=1}^n$ over quantum registers $R,E$ and classical registers $S_i,C_i$ over alphabets $\mathcal{S, C}$. 
We say that a function $f:\mathcal{P}(\mathcal{C})\rightarrow\mathbb{R}$ is a min-tradeoff function for the sequence $\{\mathcal{M}_i \in \CPTP(RE, S_iC_iRE)\}_{i=1}^n$ if $f$ is a min-tradeoff function for each channel $\mathcal{M}_i\in \CPTP(RE, S_iC_iRE)$ in the sequence. 

Henceforth, let $f$ be a min-tradeoff function for a sequence of channels $\{\mathcal{M}_i \in \CPTP(RE, S_iC_iRE)\}$.
Following the exposition from \cite[Section 4]{MFSR22}, for any $q\in\mathcal{P}(\mathcal{C})$ define
\begin{align}
    \Sigma_i(q) := \{\mathcal{M}_i(\omega_{REF})\;|\; \omega\in S(REF),\; \Tr_{S_iREF}[\mathcal{M}(\omega_{REF})]=\sum_{c\in\mathcal{C}} q_c\ketbratext{c}{c}_{C_i}\}
\end{align}
as the set of output states of the channel $\mathcal{M}_i$ whose $C_i$ register is distributed as $q$.
Note that the minimisation in \thref{def_min-tradeoff} is precisely over states in $\Sigma_i(q)$.
Further define $\Sigma(q):=\bigcup_i\Sigma_i(q)$.
We will need the following properties of the min-tradeoff function $f$ as defined in \cite{DF19,DFR20,MFSR22}: 
\begin{align}
    \Max(f) &:= \max_{q\in\mathcal{P}(\mathcal{C})} f(q),\\
    \Min(f) &:= \min_{q\in\mathcal{P}(\mathcal{C})} f(q),\\
    \Min_\Sigma (f) &:= \min_{q:\Sigma(q)\neq\emptyset} f(q),\\
    \textup{Var}(f)&:= \max_{q:\Sigma(q)\neq\emptyset} \sum_{x\in\mathcal{C}}q(x)f(\delta_x)^2-\bigg(\sum_{x\in\mathcal{C}}q(x)f(\delta_x)\bigg)^2,
\end{align}
where $\delta_x$ is the distribution over $\mathcal{C}$ with $\Pr[C=x]=1$.

Lastly, we will require the following result from \cite[Lemma V.5]{DF19}.
\begin{lemma}\thlabel{th_infrequent_sampling}
    Suppose that for each $i$, the channels $\mathcal{M}_i\in\CPTP(RE,S_iC_iRE)$ for $S_i\in\{0,1\}$ and $C_i\in\{0,1,\perp\}$ are of the form
    \begin{align}
        \mathcal{M}_i = \gamma \mathcal{M}_i^{\textup{test}} + (1-\gamma)\mathcal{M}_i^{\textup{gen}},
    \end{align}
    where $\gamma\in(0,1)$ and $\mathcal{M}_i^{\textup{gen}}$ sets $C_i=\perp$.
    Further suppose that the function $g:\mathcal{P}(\{0,1\})\rightarrow\mathbb{R}$ is an affine function which for all $i\in\{1,...,n\}$ and for all $[q_0,q_1]\in\mathcal{P}(\{0,1\})$ satisfies
    \begin{align}\label{g_function_def}
        g(q)\leq \min_{\omega_{REF}} \Big\{H(S_i|EF)_{\mathcal{M}_i(\omega_{REF})} \,:\, \Tr_{SREF}[\mathcal{M}_i^\textup{test}(\omega_{REF})]=q_0\ketbra{0}_{C_i}+q_1\ketbra{1}_{C_i}\Big\}.
    \end{align}
    Then the affine function defined by 
    \begin{align}\label{affine_f}
        f(\delta_0)=g(\delta_1)+\frac{1}{\gamma}(g(\delta_0)-g(\delta_1)),\qquad f(\delta_1)=f(\delta_\bot)=g(\delta_1)
    \end{align}
    is a min-tradeoff function for the sequence $\{\mathcal{M}_i \in \CPTP(RE, S_iC_iRE)\}_{i=1}^n$ which additionally satisfies
    \begin{align}
        \textup{Max}(f) &= \textup{Max}(g),\\
        \textup{Min}(f) &= \Big(1-\frac{1}{\gamma}\Big)\textup{Max}(g)+\frac{1}{\gamma}\textup{Min}(g),\\
        \textup{Min}_{\Sigma}(f)&\geq \textup{Min}(g),\\
        \textup{Var}(f) &\leq\frac{1}{\gamma}\big(\textup{Max}(g)-\textup{Min}(g)\big)^2.
    \end{align}
\end{lemma}
Channels of the form as in this lemma are known as \emph{infrequent sampling channels}.

\subsection{Completeness}
In this section, we prove that for any $(\mathcal{D}_A^\textup{hon}$, $\mathcal{D}_B^\textup{hon}, \omega_{R_AR_B}^\textup{hon})$ which plays the game $\mathcal{G}$ independently in every round and wins with probability $\oexp$, the protocol $\mathcal{P}_{(\mathcal{D}_A^{\hon}, \mathcal{D}_B^{\hon}),(n,\mathbf{pe}, \mathbf{ec}, \mathbf{pa})}$ with parameters $\mathbf{ec}=(\textup{Synd}, \textup{Corr}, \mathcal{H}_{\EC}, l_{\EC})$ and $\mathbf{pe}=(\gamma, \oexp, \delta_\textup{tol})$ is $\big(\epsilon_\textup{com}^{\PE}+\epsilon_\textup{com}^{\EC}\big)$-complete if
\begin{align}
    \Pr[S_A\neq \hat{S}_B]_{\rho^{\hon}}\leq\varepsilon_\textup{com}^{\EC}, \qquad \text{where}\qquad {\rho^{\hon}} = \mathcal{P}_{(\mathcal{D}_A^{\hon}, \mathcal{D}_B^{\hon}),(n,\mathbf{pe}, \mathbf{ec}, \mathbf{pa})}(\omega^{\hon}_{R_AR_B})
\end{align}
and if
\begin{align}\label{PE_com_bound}
    \gamma \geq \frac{2(1-\oexp) + \dtol}{\dtol^2 n} \ln\frac{1}{\varepsilon_\textup{com}^{\PE}}.
\end{align}

To this end, let $\Omega$ denote the event that the protocol does not abort.
We have $\Omega = \Omega_{\PE} \wedge \Omega_{\EC}$ where $\Omega_{\PE}$ is the event $F_{\PE}=\checkmark$ and $\Omega_{\EC}$ is the event $F_{\EC}=\checkmark$.
Define the additional event 
\begin{align}
    &\Omega_\textup{g} := \big\{S_A = \hat{S}_B\big\},
\end{align}
and note $\Omega\wedge\Omega_\textup{g} = \Omega_{\PE}\wedge\Omega_{\EC}\wedge\Omega_\textup{g} = \Omega_{\PE}\wedge\Omega_\textup{g}$.
The probability of abort is given by event $\Omega^c=\Omega_{\EC}^c\vee\Omega_{\PE}^c$, wherefore 
\begin{align}
    \Pr[\Omega^c_{\EC}\vee\Omega^c_{\PE}]_{\rho^{\hon}}
    \leq \Pr[\Omega_{\EC}^c]_{\rho^{\hon}}+ \Pr[\Omega^c_{\PE}]_{\rho^{\hon}} \leq \Pr[\Omega_\textup{g}^c]_{\rho^{\hon}} + \Pr[\Omega_{\PE}^c]_{\rho^{\hon}}. \label{comp_two_terms}
\end{align}
The first inequality follows from the union bound and the second as event $\Omega_\textup{g}$ implies event $\Omega_{\EC}$.
Informally, the term $\Pr[\Omega_\textup{g}^c]_{\rho^{\hon}}$ is the probability that the error correction in the protocol fails (not to be confused with the error correction \emph{check} which uses hash functions and whose failure is given by event $\Omega^c_{\EC}$), and the term $\Pr[\Omega^c_{\PE}]_{\rho^{\hon}}$ is the probability that the parameter estimation fails.
We tackle each term in Eq.~\eqref{comp_two_terms} separately.

To bound the term $\Pr[\Omega_\textup{g}^c]_{\rho^{\hon}}$ we note that the error correction scheme aims to correct all instances $S_{A,i}\neq S_{B,i}$.
In an honest implementation of the protocol, the expected number of such occurrences is $n(1-\oexp)$.
For $n$ large enough \cite[Theorem 1]{TMPE17} shows that for any parameter $\varepsilon_\textup{com}^{\EC}\in(0,1]$ independent of $n$, there exists an error correction scheme given by functions
\begin{align}
    \textup{Synd}:\{0,1\}^n\rightarrow\{0,1\}^{\lambda_{\EC}}, \qquad \textup{Corr}:\{0,1\}^n\times\{0,1\}^{\lambda_{\EC}}\rightarrow\{0,1\}^{n}
\end{align}
satisfying $\Pr [\Omega_g]_{\rho^{\hon}}=\Prtext[S_A=\hat{S}_B]_{\rho^{\hon}}\geq 1- \varepsilon_\textup{com}^{\EC}$ and $\frac{\lambda_{\EC}}{n}=h(Q) + \Theta\big(\frac{\log n}{n}\big)$, where $h$ is the binary entropy and $Q = 1-\oexp$ is the quantum bit error rate.

For the remaining term $\Pr[\Omega^c_{\PE}]_\rho$ we note that in the honest behaviour, $C_i$ is a sequence of i.i.d variables with
\begin{align}
    \Pr[C_i=0]=\gamma(1-\omega_\textup{exp}),\qquad \Pr[C_i=1]=\gamma\omega_\textup{exp}, \qquad \Pr[C_i=\bot]=1-\gamma,
\end{align}
where the probabilities are over the distribution generated when the honest devices are interacted with for the $i$-th time.
Then,
\begin{align}\label{Chernoff_bound}
    \Pr\big[\Omega^c_{\PE}\big]_{\rho^{\hon}}
    &= \Pr\big[|\{i \;:\; C_i=0\}|>(1-\omega_\textup{exp}+\delta_\textup{tol})\gamma n\big]_{\rho^{\hon}}\\
    &= \Pr[|\{i \;:\; C_i=0\}|>\bigg(1+\frac{\delta_\textup{tol}}{1-\omega_\textup{exp}}\bigg)\mathbb{E}\big[|\{i \;:\; C_i=0\}|\big]]_{\rho^{\hon}}\\
    &\leq \exp{-\frac{\delta_\textup{tol}^2}{2(1-\omega_\textup{exp})+\delta_\textup{tol}}\gamma n},
\end{align}
where in the last inequality we used the multiplicative Chernoff bound in \cite[Theorem 4.4]{ME05}.
Setting $\gamma$ as in the statement of the theorem (Eq.~\eqref{PE_com_bound}), we directly obtain
\begin{align}
    \Pr\big[\Omega^c_{\PE}\big]_{\rho^{\hon}}
    \leq \varepsilon_\textup{com}^{\PE},
\end{align}
as is required to complete the proof.

\subsection{Correctness}
In this section we tackle the correctness as per \thref{th_correctness}, which states that for any triple $(\mathcal{D}_A, \mathcal{D}_B,\omega^0_{R_AR_BE})$, the state $\rho = \mathcal{P}_{(\mathcal{D}_A, \mathcal{D}_B),(n,\mathbf{pe}, \mathbf{ec}, \mathbf{pa})}(\omega^0_{R_AR_BE})$ with error correction tuple $\mathbf{ec}=(\textup{Synd}, \textup{Corr}, \mathcal{H}_{\EC}, l_{\EC})$ is such that $\Pr[K_A\neq K_B \wedge \Omega_{\PE} \wedge \Omega_{\EC}]_\rho\leq\epscorr$ if $l_{\EC}\geq \log\frac{1}{\epscorr}$, for some $\epscorr\in(0,1]$.

The proof is analogous to that of \cite{TL17}, but we present it here for completeness.
As with before, let $\Omega=\Omega_{\PE}\wedge\Omega_{\EC}$ denote the event that the protocol does not abort.
Indeed:
\begin{align}
    \Pr\big[K_A\neq K_B \wedge \Omega_{\PE} \wedge \Omega_{\EC}\big]_\rho
    &\leq \Pr\big[K_A\neq K_B \wedge \Omega_{\EC}\big]_\rho\\
    &= \Pr\big[H_{\PA}(S_A)\neq H_{\PA}(\hat{S}_B) \wedge H_{\EC}(S_A)=H_{\EC}(\hat{S}_B)\big]_\rho\\
    &\leq \Pr\big[S_A\neq \hat{S}_B \wedge H_{\EC}(S_A)=H_{\EC}(\hat{S}_B)\big]_\sigma\\
    &= \Pr\big[S_A\neq \hat{S}_B\big]_\sigma \cdot \Pr\big[H_{\EC}(S_A)=H_{\EC}(\hat{S}_B) \;|\; S_A\neq \hat{S}_B\big]_\sigma\\
    &\leq \Pr\big[H_{\EC}(S_A)=H_{\EC}(\hat{S}_B) \;|\; S_A\neq \hat{S}_B\big]_\sigma\\
    &\leq \frac{1}{2^{l_{\EC}}}.
\end{align}
In the above, the first inequality follows from the union bound, the second inequality follows from the fact that $H_{\PA}(S_A)\neq H_{\PA}(\hat{S}_B)$ implies $S_A\neq \hat{S}_B$ and the last inequality is the defining property of $2$-universal hashing.
Choosing $l_{\EC}\geq\log\frac{1}{\epscorr}$ completes the proof.

\subsection{Secrecy}\label{sec_secrecy_proof}
Finally, we prove \thref{th_secrecy} through the use of the Leftover Hashing Lemma and the generalised entropy accumulation theorem.
Given $n\in\mathbb{N}$ and $\epssec\in(0,1]$, we want to show that for any triple $(\mathcal{D}_A, \mathcal{D}_B,\omega^0_{R_AR_BQ_E})$, protocol $\mathcal{P}_{(\mathcal{D}_A, \mathcal{D}_B),(n,\mathbf{pe}, \mathbf{ec}, \mathbf{pa})}$ with $\mathbf{pe}=(\gamma, \oexp, \delta_\textup{tol})$, $\mathbf{ec}=(\textup{Synd}, \textup{Corr}, \mathcal{H}_{\EC}, l_{\EC})$ and $\mathbf{pa} = (\mathcal{H}_{\PA}, \lkey)$ is $\epssec$-secret if
\begin{align}
    l_\textup{key} \leq \;&n{g}(\omega_\textup{exp}-\delta_\textup{tol}) -d_1\sqrt{n}-d_0 - (2-\oexp+\delta_\textup{tol})n(\gamma+\kappa)   - 2\vartheta\left(\frac{\varepsilon_s}{4}\right)\\
        &-l_{\EC}-\lambda_{\EC} 
        - 2\log\frac{1}{\epssec-2\varepsilon_s},
\end{align}
where $\varepsilon_s\in\big(0,\frac{1}{2}\epssec\big)$ is a variable to be optimised, $\kappa\geq\frac{1}{\sqrt{n}}\sqrt{\ln{8}-\ln{\varepsilon_s}}$, $g:[0,1]\rightarrow\mathbb{R}$ is the affine function defined by
\begin{align}\label{affine_th}
    {g}(p):=\frac{p-\beta}{\ln2\cdot(1-\beta+\omega_3)} - \log\big(1-\beta+\omega_3\big),
\end{align}
for $\beta\in[\omega_3, \omega_2]$ to be optimised.
Further, the coefficients $d_1, d_0$ depend only on $\epssec$, $\varepsilon_s$ and $\beta$ and are respectively as in Eqs.~\eqref{d1_def_final} and~\eqref{d0_def_final}.
Lastly, $\vartheta(\delta) = -\log\big(1-\sqrt{1-\delta^2}\big)$.

To prove this, we fix parameters $\mathbf{pe}, \mathbf{ec}, \mathbf{pa}$ and consider a different protocol $\widetilde{\mathcal{P}}$ initialised with devices $\widetilde{\mathcal{D}}_A, \widetilde{\mathcal{D}}_B$ and state $\widetilde{\omega}^0_{R_AR_BQ_E}\in S(R_AR_BQ_E)$.
Protocol $\widetilde{\mathcal{P}}$ is presented in algorithm \ref{prot_qkd_virtual} and differs from the original protocol ${\mathcal{P}}$ in algorithm \ref{prot_qkd_actual} only in the measurement stage and in the parameter estimation stage (error correction and privacy amplification proceed identically in either protocol).
Additionally, protocol $\widetilde{\mathcal{P}}$ maintains the register $\bar{S}_B$, which is only used later in the proof to bound a smooth max-entropy term. 
Crucially, we remark that Eve never has access to $\bar{S}_B$.
\begin{figure}
    \centering
    \begin{algorithm}[H]\label{prot_qkd_virtual}
        \SetAlgoLined
        \caption{A virtual protocol for DIQKD from non-local games: $\widetilde{\mathcal{P}}$}
        \SetKwInOut{Input}{Input}
        \SetKwInOut{Parameters}{Parameters}
        \SetKwInOut{Output}{Output}
        \SetKwInOut{Init}{Initialization}
        \SetKwRepeat{Repeat}{repeat}{until}
        \Input{Devices $\mathcal{D}_{A},\mathcal{D}_{B}$}
        \Parameters{Number of rounds $n$\\
        PE parameters, $\mathbf{pe}=(\gamma, \oexp, \delta_\textup{tol})$\\
        EC parameters, $\mathbf{ec}=(\textup{Synd},\textup{Corr},\mathcal{H}_{\EC},l_{\EC})$\\
        PA parameters, $\mathbf{pa}=(\mathcal{H}_{\PA},\lkey)$}
        \textbf{1. Measurement}\;
        \For{$i=1$ to $n$}{
        Bob chooses and announces $T_i\in\{0,1\}$ with $\Pr[T_i=1]=\gamma$\;
        Alice chooses and announces $X_i\in\mathcal{X}$ and inputs it to $\mathcal{D}_A$ to receive $A_i\in\mathcal{A}$\;
        Bob chooses and announces $Y_i\in\mathcal{Y}$ and inputs it to $\mathcal{D}_B$ to receive $B_i\in\mathcal{B}$\;
        Alice sets $S_{A,i}=SK_A(X_i,Y_i,A_i)$\;
        Bob sets $S_{B,i}=SK_B(X_i,Y_i,A_i)$\;}
        \eIf{$T_i=0$}{
        Bob sets $C_i=\perp$ and $\bar{S}_{B,i}=0$\;
        }{
        Alice announces $A_i$\;
        Bob sets $C_i=V(X_i,Y_i,A_i,X_i)$ and $\bar{S}_{B,i}=S_{B,i}$\;
        }
        \textbf{2. Parameter Estimation}\;
        \eIf{$|\{i\;:\; C_i=0\}| \leq (1-\omega_\textup{exp}+\delta_\textup{tol})\cdot\gamma n$ }{
        Bob announces $F_{\PE}=\checkmark$\;
        }{
        Bob announces $F_{\PE}=\times$, both parties output $K_A=K_B=\ketbratext{\bot}{\bot}$\;
        }
        \textbf{3. Error Correction}\;
        \textbf{4. Privacy Amplification}\;
    \end{algorithm}
\end{figure}

It is straightforward to see that
\begin{align}
    {\mathcal{P}}_{(\mathcal{D}_A, \mathcal{D}_B),(n,\mathbf{pe}, \mathbf{ec}, \mathbf{pa})}(\omega^0_{R_AR_BQ_E}) = \Tr_{\bar{S}_B}\circ\widetilde{\mathcal{P}}_{(\widetilde{\mathcal{D}}_A, \widetilde{\mathcal{D}}_B),(n,\mathbf{pe}, \mathbf{ec}, \mathbf{pa})}(\widetilde{\omega}^0_{R_AR_BQ_E})
\end{align}
whenever $\widetilde{\omega}^0_{R_AR_BQ_E} = \omega^0_{R_AR_BQ_E}$ and the device's quantum instruments through an instance of $\mathcal{P}$ and $\widetilde{\mathcal{P}}$ are identical.
Further, since in protocol $\widetilde{\mathcal{P}}$ the eavesdropper has access to the additional information in registers $X_1^iY_1^iT_1^i$ prior to the $i$-th interaction with the devices, secrecy of protocol $\widetilde{\mathcal{P}}_{(\mathcal{D}_A, \mathcal{D}_B),(n,\mathbf{pe}, \mathbf{ec}, \mathbf{pa})}$ with parameters $(\mathbf{pe}, \mathbf{ec}, \mathbf{pa})$ implies secrecy of ${\mathcal{P}}_{(\mathcal{D}_A, \mathcal{D}_B),(n,\mathbf{pe}, \mathbf{ec}, \mathbf{pa})}$ with identical parameters.
Therefore, without loss of generality we prove secrecy of protocol $\widetilde{\mathcal{P}}$ instantiated with $(\mathcal{D}_A, \mathcal{D}_B,\omega^0_{R_AR_BQ_E})$.

We begin by recollecting some notation:
\begin{align}
    &\Omega_{\PE} := \{|\{i\;:\; C_i=0\}| \leq (1-\omega_\textup{exp}+\delta_\textup{tol})\cdot\gamma n\},\\
    &\Omega_{\EC} := \{H_{\EC}(S_A)=H_{\EC}(\hat{S}_B)\},
\end{align}
are the events that the parameter estimation test passes and that the error correction test passes, such that $\Omega=\Omega_{\PE}\wedge\Omega_{\EC}$ is the event that the protocol does not abort.
Let 
\begin{align}
    \widetilde{\mathcal{P}}_{(\mathcal{D}_A, \mathcal{D}_B),(n,\mathbf{pe}, \mathbf{ec}, \mathbf{pa})}
({\omega}^0_{R_AR_BQ_E})=\rho_{A_1^nB_1^nK_AK_BS_AS_B\hat{S}_B\bar{S}_B C_1^nE}
\end{align}
be the state output by the protocol, where Eve's total side information is 
\begin{align}
    E = X_1^nY_1^n T_1^n S_{A,T_\textup{test}}H_AZ F_{\PE}F_{\EC} Q_{E}.
\end{align}
Lastly, let $\sigma_{A_1^nB_1^nS_AS_B\hat{S}_B\bar{S}_BC_1^nE}$ be the state prior to the privacy amplification phase.

We start the proof by noting that if $\Pr[\Omega_{\EC}|{\Omega}_{\PE}]< \varepsilon_s^2$ or if $\Pr[{\Omega}_{\PE}]<\varepsilon_a$ for some $\varepsilon_a\in\big(0,\frac{1}{2}\epssec\big]$, $\varepsilon_s\in(0,\frac{1}{2}\epssec)$, we have that
\begin{align}
    \Pr[{\Omega}_{\PE}\wedge\Omega_{\EC}]_\rho\leq\min\Big\{\Pr[\Omega_{\EC}|{\Omega}_{\PE}]_\rho, \Pr[{\Omega}_{\PE}]_\rho\Big\}\leq \max\left\{\varepsilon_s^2, \varepsilon_a\right\}
\end{align}
which implies
\begin{align}
    \norm{\rho_{K_AE \wedge {\Omega}_{\PE}\wedge\Omega_{\EC}} - \tau_{K_A}\otimes\rho_{E\wedge {\Omega}_{\PE}\wedge\Omega_{\EC}}}_{\Tr} 
    &\leq \norm{\rho_{K_AE \wedge {\Omega}_{\PE}\wedge\Omega_{\EC}}}_{\Tr} +\norm{\tau_{K_A}\otimes\rho_{E\wedge {\Omega}_{\PE}\wedge\Omega_{\EC}}}_{\Tr}\\
    &= 2\Pr[{\Omega}_{\PE}\wedge\Omega_{\EC}]_\rho \label{sececy_trivial}
\end{align}
and so secrecy is given directly by $2\max\{\varepsilon_s^2, \varepsilon_a\}\leq 2\max\{\varepsilon_s, \varepsilon_a\}\leq\epssec$.

Therefore, from now on we assume $\Pr[\Omega_{\EC}|{\Omega}_{\PE}]_\rho\geq \varepsilon_s^2$ and $\Pr[{\Omega}_{\PE}]_\rho\geq\varepsilon_a$ such that
\begin{align}
    \norm{\rho_{K_AE \wedge {\Omega}_{\PE}\wedge\Omega_{\EC}} - \tau_{K_A}\otimes\rho_{E\wedge {\Omega}_{\PE}\wedge\Omega_{\EC}}}_{\Tr}
    &= \Pr[{\Omega}_{\PE}]_\rho \norm{\rho_{K_AE | {\Omega}_{\PE}\wedge\Omega_{\EC} } - \tau_{K_A}\otimes\rho_{E| {\Omega}_{\PE}\wedge\Omega_{\EC} }}_{\Tr}.\label{pre_lhl}
\end{align}
To bound the trace distance on the right, we use the following result from \cite[Proposition 9]{TL17}:
\begin{lemma}[Leftover Hashing]\thlabel{th_leftover_hashing}
    Let $\sigma\in S_{\leq}(S_1^nE)$ be classical-quantum and let $\mathcal{H}:=\big\{H:\{0,1\}^n\rightarrow \{0,1\}^{l}\big\}$ be a family of $2$-universal functions.
    Further let
    \begin{align}
        \rho_{K_1^lS_1^nE} := \mathcal{E}_{\mathcal{H}}(\sigma_{S_1^nE}\otimes \tau_H),
    \end{align}
    where $\tau_H$ is the maximally mixed state of dimension $|\mathcal{H}|$.
    The CPTP map $\mathcal{E}_{\mathcal{H}}$ chooses a hash function from register $H$, applies it to the classical string in registers $S_1^n$ and stores the hash in registers $K_1^l$.
    
    Then, for any $\varepsilon\in\Big[0,\sqrt{\Tr\big[\sigma_{S_1^nE}\big]}\Big]$:
    \begin{align}
        \norm{\rho_{K_1^lEH}-\tau_{K_1^l}\otimes\rho_{EH}}_{\Tr}\leq \frac{1}{2}2^{-\frac{1}{2}(H_{\min}^\varepsilon(S_1^n|E)_\sigma-l)}+2\varepsilon,
    \end{align}
    where $\tau_{K_1^l}$ is the maximally mixed state of dimension $2^l$.
\end{lemma}
Intuitively, the Leftover Hashing Lemma guarantees that the key $K_A$ output by Alice in protocol $\mathcal{P}$ (algorithm \ref{prot_qkd_actual}) is close to a perfect key provided the smooth min-entropy of the sifted key $S_A$ conditioned on Eve's side information $E$ is large enough. 
Continuing from Eq.~\eqref{pre_lhl}, the Leftover Hashing Lemma implies
\begin{align}
    \norm{\rho_{K_AE | {\Omega}_{\PE}\wedge\Omega_{\EC} } - \tau_{K_A}\otimes\rho_{E| {\Omega}_{\PE}\wedge\Omega_{\EC} }}_{\Tr} 
    \leq \frac{1}{2}\cdot 2^{-\frac{1}{2}\big(H_\text{min}^{\varepsilon_s}(S_A|E)_{\sigma_{| {\Omega}_{\PE}\wedge\Omega_{\EC}}}-\lkey\big)}+2{\varepsilon_s}.
\end{align}
The smooth min-entropy term in the exponent is further bounded by
\begin{align}
    H_\text{min}^{\varepsilon_s}(S_A|E)_{\sigma_{|{\Omega}_{\PE}\wedge\Omega_{\EC}}} 
    &\geq H_\text{min}^{\varepsilon_s}(S_A|E)_{\sigma_{|{\Omega}_{\PE}}}\\
    &\geq H_\text{min}^{\varepsilon_s}(S_A|X_1^nY_1^n T_1^n S_{A,T_\textup{test}} Q_{E})_{\sigma_{{|\Omega}_{\PE}}}-l_{\EC}-\lambda_{\EC}-2.\label{min_entropy_bound_1}
\end{align}
The first inequality follows from \cite[Lemma 10]{TL17} and for the second inequality we used the chain rule in \cite[Lemma 6.8]{Tomamichel16} to remove the error correction and flag registers $H_AZ F_{\PE}F_{\EC}$ from Eve's side information.

All that remains is to bound $H_\text{min}^{\varepsilon_s}(S_A|X_1^nY_1^n T_1^n S_{A,T_\textup{test}} Q_{E})_{\sigma_{{|\Omega}_{\PE}}}$.
However, there are two last technical hurdles we need to address.
Firstly, the conditioning event $|{\Omega}_{\PE}$ imposes a condition on the registers ${C}_i$ which are themselves computed from the registers $S_{A,i}S_{B,i}T_i$.
We need to ensure that all these registers are included in the smooth-min entropy term to bound. 
The conditioning side already includes $X_1^nY_1^nT_1^n S_{A,T_\textup{test}}$ so it remains to add the classical registers $S_{B,T_\textup{test}}=\{S_{B,i}\,:\,i\in T_\textup{test}\}$ for $T_\textup{test}=\{i:T_i=1\}$ which are not publicly announced in $\widetilde{\mathcal{P}}$. 
Secondly, the classical register $S_{A,T_\textup{test}}$ which is available to the eavesdropper needs to be removed from the conditioning side of the smooth min-entropy.
This is to ensure Alice's outcome $A_i$ on round $i$ does not depend on her previous announced outcomes, which would violate the no-signalling assumption of our protocol.

With this in mind, let $\bar{S}_B$ be such that $\bar{S}_{B,T_\textup{test}}=S_{B,T_\textup{test}}$ and all other elements are set deterministically to $0$. Letting $\bar{E}_n=X_1^nY_1^n T_1^n  Q_{E}$, we have that for $\varepsilon_s'\in\big[0,\frac{1}{2}\epssec\big)$ and $\varepsilon_s''\in\big[0,\frac{1}{4}\epssec\big)$ such that $\delta=\varepsilon_s-\varepsilon_s'-2\varepsilon_s''>0$:
\begin{align}
    H_{\min}^{\varepsilon_s}(S_A| S_{A,T_\textup{test}}\bar{E}_n)_{\sigma_{{|\Omega}_{\PE}}}
    &\geq H_{\min}^{\varepsilon_s}(S_A| S_{A,T_\textup{test}}\bar{S}_B\bar{E}_n)_{\sigma_{{|\Omega}_{\PE}}}\\
    &\geq H_{\min}^{\varepsilon_s'}(S_A\bar{S}_B|\bar{E}_n)_{\sigma_{{|\Omega}_{\PE}}}  - H_{\max}^{\varepsilon_s''}(S_{A,T_\textup{test}}\bar{S}_B|\bar{E}_n)_{\sigma_{{|\Omega}_{\PE}}} - 2\vartheta(\delta)\\
    &\geq H_{\min}^{\varepsilon_s'}(S_A\bar{S}_B|\bar{E}_n)_{\sigma_{{|\Omega}_{\PE}}}  - H_{\max}^{\varepsilon_s''}(S_{A,T_\textup{test}}\bar{S}_B|T_1^n)_{\sigma_{{|\Omega}_{\PE}}} - 2\vartheta(\delta)\label{min_max_entropy_chain_rule0},
\end{align}
where the first and last inequalities follow from the data processing inequality of the smooth min/max-entropy, and in the second inequality we used Eq.~(6.60) in \cite[Lemma 6.16]{Tomamichel16} as well as the fact that $H_{\min}(AB)\geq H_{\min}(A)$ when $A$ and $B$ are classical to remove the additional $S_{A,T_\textup{test}}$ on the left of the min-entropy term.
The function $\vartheta(\delta)=-\log\big(1-\sqrt{1-\delta^2}\big)$.

We now want to bound the smooth max-entropy term in the right hand side of Eq.~\eqref{min_max_entropy_chain_rule0}.
In general, the max-entropy $H_{\max}(X|B)_\rho$ of a classical-quantum state $\rho\in S(XB)$ is upper bounded by the logarithm of the support of $\rho_X$.
However, in our case it is not clear what the support of $S_{A,T_\textup{test}}\bar{S}_B$ on $\sigma_{|\Omega_{\PE}}$ is, since it depends on the values of the Bernoulli variables $T_i$.
Indeed, when conditioned on $\Omega_{\PE}$, we expect $S_{A,T_\textup{test}}$ and $\bar{S}_{B,T_\textup{test}}$ to equal in all places except with probability $(1-\oexp+\dtol)$.
This implies that the expected size (over the random variable $T_i$ with $\Pr[T_i=1]=\gamma$) of $\log\supp\big(\sigma_{S_{A,T_\textup{test}}\bar{S}_B|\Omega_{\PE}}\big)$ is $(2-\oexp+\delta_\textup{tol})n\gamma$.
Thus, for a small parameter $\kappa$, we want to replace the state $\sigma_{|\Omega_{\PE}}\in S_{\leq}(S_{A,T_\textup{test}}\bar{S}_BT_1^n)$ in which $\log\supp\big(\sigma_{S_{A,T_\textup{test}}\bar{S}_B|\Omega_{\PE}}\big)$ does not exceed $(2-\oexp+\delta_\textup{tol})n(\gamma+\kappa)$ except with small probability (certified by some tail bound), with a state $\tilde{\sigma}\in S_{\leq}(S_{A,T_\textup{test}}\bar{S}_BT_1^n)$ in which $\log\supp\big(\tilde{\sigma}_{S_{A,T_\textup{test}}\bar{S}_B}\big)$ never exceeds $(2-\oexp+\delta_\textup{tol})n(\gamma+\kappa)$.
This will subsequently allow us to bound the smooth max-entropy term with the logarithm of the support of $\tilde{\sigma}_{S_A\bar{S}_B}$.
Formally:

\begin{lemma}\thlabel{th_max_entropy_bound}
    Let $n, T_1^n, S_{A,T_\textup{test}}, \bar{S}_B, \Omega_{\PE}$ and $\sigma_{|\Omega_{\PE}}\in S(T_1^n S_{A,T_\textup{test}} \bar{S}_B)$ be as above.
    Further let $\varepsilon_s''\in(0,1]$.
    Then, for any $\kappa\geq\frac{1}{\sqrt{n}}\sqrt{-\ln{\varepsilon_s''}}$ it holds that
    \begin{align}
        H_{\max}^{\varepsilon_s''}(S_{A,T_\textup{test}}\bar{S}_B|T_1^n)_{\sigma_{{|\Omega}_{\PE}}} 
        \leq (2-\oexp+\delta_\textup{tol})n(\gamma+\kappa).
    \end{align}
\end{lemma}
\begin{proof}
    For $\kappa$ as in the statement of the lemma, define the event 
    \begin{align}
        \Omega_{T,\kappa}&:=\big\{|\{i\;:\;T_i=1\}|\geq n(\gamma+\kappa)\big\}.
    \end{align}
    Hoeffding's inequality implies
    \begin{align}
        \varepsilon :=\Pr[\Omega^c_{T,\kappa}]_{\sigma_{|\Omega_{\PE}}} \leq \textrm{e}^{-2n\kappa^2}.
    \end{align}
    Since $\kappa\geq\frac{1}{\sqrt{n}}\sqrt{-\ln{\varepsilon_s''}}$ and $\Tr[\sigma_{|\Omega_{\PE}}]=1$ we have that
    \begin{align}
        \textrm{e}^{-2n\kappa^2} \leq (\varepsilon_s'')^2 \leq \Tr[\sigma_{|\Omega_{\PE}}].
    \end{align}
    Hence, Lemma 7 in \cite{TL17} implies that the state $\tilde{\sigma}\in S_{\leq}(S_{A,T_\textup{test}} \bar{S}_B T_1^n)$ defined by 
    \begin{align}
        \tilde{\sigma} := \frac{\sin^2 \phi}{1 - \varepsilon}\sigma_{|\Omega_{\PE}\wedge\Omega^c_{T,\kappa}}
    \end{align}
    for some normalization factor $\phi\in[0,\frac{\pi}{2}]$, satisfies 
    \begin{align}
        \Pr\big[\Omega_{T,\kappa}\big]_{\tilde{\sigma}} = 0 \qquad\textup{and}\qquad P(\sigma_{|\Omega_{\PE}}, \tilde{\sigma}) \leq \sqrt{\varepsilon} \leq \textrm{e}^{-n\kappa^2}.
    \end{align}
    Using the definition of the smooth max-entropy together with the fact that $P(\sigma_{|\Omega_{\PE}}, \tilde{\sigma}) \leq \textrm{e}^{-n\kappa^2} \leq \varepsilon_s''$ we have that
    \begin{align}
        H_{\max}^{\varepsilon_s''}(S_{A,T_\textup{test}}\bar{S}_B|T_1^n)_{\sigma_{{|\Omega}_{\PE}}} 
        &\leq H_{\max}(S_{A,T_\textup{test}}\bar{S}_B|T_1^n)_{\tilde{\sigma}}.
    \end{align}
    Finally, re-writing $\tilde{\sigma}=\sum_{t_1^n\in\Omega^c_{T,\kappa}}\Pr[T_1^n=t_1^n]\ketbratext{t_1^n}{t_1^n}\otimes\tilde{\sigma}_{|t_1^n}$ and using the expression for the max-entropy conditioned on a classical variable given in \cite[Eq.~(6.26)]{Tomamichel16}, the following holds
    \begin{align}
        H_{\max}(S_{A,T_\textup{test}}\bar{S}_B|T_1^n)_{\tilde{\sigma}}
        &= \log\left[\sum_{t_1^n\in\Omega^c_{T,\kappa}}\Pr[T_1^n=t_1^n]\cdot2^{H_{\max}(S_{A,T_\textup{test}}\bar{S}_B)_{\tilde{\sigma}_{|t_1^n}}}\right] \\
        &\leq (2-\oexp+\delta_\textup{tol})n(\gamma+\kappa),
    \end{align}
    where in the last inequality we bounded the max entropy with an upper bound on $\log \supp (\tilde{\sigma}_{S_{A,T_\textup{test}} \bar{S}_B|t_1^n})$, noting that $\log \supp (\tilde{\sigma}_{S_{A,T_\textup{test}} |t_1^n}) \leq n(\gamma+\kappa)$ and that $\Omega_{\PE}$ implies that ${S}_{A,T_\textup{test}}$ does not match $\bar{S}_{B,T_\textup{test}}$ in at most $(1-\oexp+\dtol)n\gamma\leq(1-\oexp+\dtol)n(\gamma+\kappa)$ places.
\end{proof}
Applying the bound from \thref{th_max_entropy_bound} to Eq.~\eqref{min_max_entropy_chain_rule0} gives
\begin{align}\label{min_max_entropy_chain_rule}
    H_{\min}^{\varepsilon_s}(S_A| S_{A,T_\textup{test}}\bar{E})_{\sigma_{{|\Omega}_{\PE}}}
    \geq H_{\min}^{\varepsilon_s'}(S_A\bar{S}_B|\bar{E})_{\sigma_{{|\Omega}_{\PE}}} - (2-\oexp+\delta_\textup{tol})n(\gamma+\kappa) - 2\vartheta(\delta),
\end{align}
where $\delta=\varepsilon_s-\varepsilon_s'-2\varepsilon_s''>0$.

All that remains is to bound the smooth min-entropy $H_{\min}^{\varepsilon_s'}(S_A\bar{S}_B|\bar{E}_n)_{\sigma_{|{\Omega}_{\PE}}}$, where $\bar{E}_n=X_1^nY_1^n T_1^n Q_{E}$ contains only Eve's quantum information and the classical registers contained in the state at termination of protocol $\widetilde{\mathcal{P}}$.
This is done via the Generalised Entropy Accumulation Theorem (GEAT) \cite{MFSR22, MR22}, which roughly speaking bounds the smooth min-entropy with a sum over the von Neumann entropies of each round of the virtual protocol $\widetilde{\mathcal{P}}$, which are themselves bounded using the min-tradeoff functions introduced in Section \ref{sec_min_tradeoff}.

We split the application of the GEAT over the remaining subsections

\subsubsection{GEAT channels}\label{sec_app_geat}

Consider the sequence of CPTP maps $\{\mathcal{N}_i:R\bar{E}_{i-1}\rightarrow S_{A,i}\bar{S}_{B,i}C_i{R}\bar{E}_i\}_{i=1}^n$ where $S_{A,i}$ is Alice's sifted key bit in the $i$-th round, $R=R_AR_B$ denotes the internal register of Alice's and Bob's devices and $\bar{E}_i = X_1^iY_1^i T_1^i Q_{E}$ is Eve's side information in round $i$. 
The maps $\mathcal{N}_i$ receives as input the state $\omega^{i-1}_{R\bar{E}_{i-1}}$ and proceed as follows:
\begin{enumerate}
    \item Generate $X_i,Y_i,T_i$ as in protocol $\widetilde{\mathcal{P}}$. 
    The resultant state is $\omega^{i-1}_{RX_iY_iT_i\bar{E}_{i-1}}$.
    \item Apply map 
    \begin{align}
        \omega^{i-1}_{RX_iY_i}\mapsto \widetilde{\omega}^i_{RX_iY_iA_iB_i}:= \sum_{x\in\mathcal{X}}\sum_{y\in\mathcal{Y}} (P_i^{xy}\otimes\mathcal{M}_A^x\otimes\mathcal{M}_B^y)(\omega^{i-1}_{RX_iY_i}),
    \end{align}
    where 
    \begin{align}
        P_i^{xy}(\cdot) &= \ketbra{xy}_{X_iY_i} (\cdot) \ketbra{xy}_{X_iY_i},\\
        \mathcal{M}_A^x(\cdot) &= \sum_{a\in\mathcal{A}}\ketbra{a}_{A_i}\otimes \mathcal{M}_A^{a|x}(\cdot),\\
        \mathcal{M}_B^y(\cdot) &= \sum_{b\in\mathcal{B}}\ketbra{b}_{B_i}\otimes \mathcal{M}_B^{b|y}(\cdot),
    \end{align}
    for $\mathcal{M}_A^{a|x}\in\textup{CP}(R_A,R_AA_i)$ and $\mathcal{M}_B^{b|y}\in\textup{CP}(R_B,R_BB_i)$ trace non increasing maps satisfying $\sum_a \Tr\big[\mathcal{M}_A^{a|x}(\omega_{R_A})\big]=1$ and $\sum_b \Tr\big[\mathcal{M}_B^{b|y}(\omega_{R_B})\big]=1$ for all $\omega_{R_A}\in S(R_A)$, $\omega_{R_B}\in S(R_B)$ and for all $x\in\mathcal{X}$, $y\in\mathcal{Y}$.
    \item Apply map
    \begin{align}\label{infrequent_sampling_setup}
        \widetilde{\omega}^i_{X_iY_iT_iA_iB_i}\mapsto & \sum_{x,y,a,b} \Big(Q_i^{xyab}\otimes [SK_A(x,y,a)]_{S_{A,i}} \otimes 
        \big(P_i^{t=0}\otimes[0]_{\bar{S}_{B,i}}\otimes[\perp]_{C_i} + \\
        &P_i^{t=1}\otimes[SK_B(x,y,b)]_{\bar{S}_{B,i}}\otimes[V(x,y,a,b)]_{C_i}\big)\Big)(\widetilde{\omega}^i_{X_iY_iT_iA_iB_i}),
    \end{align}
    whose output we label $\omega^i_{X_iY_iS_{A,i}\bar{S}_{B,i}C_i}$.
    In the above, we use $[\cdot]$ as shorthand for $\ketbra{\cdot}$, recall $V(x,y,a,b)$ is the predicate of the game $\mathcal{G}$, and
    \begin{align}
        Q_i^{xyab}(\cdot) &= \ketbra{xy}_{X_iY_i} \otimes \Tr[\ketbra{xyab}_{X_iY_iA_iB_i} (\cdot)],\\
        P_i^{t}(\cdot) &= \ketbra{t}_{T_i} (\cdot) \ketbra{t}_{T_i}.
    \end{align}
    \item Apply eavesdropper's attack
    \begin{align}
        \mathcal{E}:\omega^i_{X_iY_iT_i\bar{E}_{i-1}}\mapsto \omega^i_{\bar{E}_{i}}.
    \end{align}
    Without loss of generality, this attack simply passes all of Eve's available information onto the next round of the protocol.
    Otherwise, the data processing inequality for the smooth min-entropy guarantees that Eve's available information decreases. 
\end{enumerate}

We notice that the sum over operations $P_i^{t}$ for $t\in\{0,1\}$ on the right hand side of Eq.~\eqref{infrequent_sampling_setup}
induces the following \emph{infrequent sampling channel} form on each $\mathcal{N}_i$
\begin{align}\label{infrequent_sampling_channel}
    \mathcal{N}_i = \gamma \mathcal{N}_i^{\textup{test}} + (1-\gamma)\mathcal{N}_i^{\textup{gen}}.
\end{align}
The channel $\mathcal{N}_i^\textup{gen}\in\CPTP(R\bar{E}_{i-1}, S_{A,i}\bar{S}_{B,i}C_iR\bar{E}_i)$ sets register $C_i$ to $\perp$ and register $\bar{S}_{B,i}$ to $0$, whereas the channel $\mathcal{N}_i^\textup{test}\in\CPTP(R\bar{E}_{i-1}, S_{A,i}\bar{S}_{B,i}C_iR\bar{E}_i)$ sets $\bar{S}_{B,i}$ to Bob's sifted key bit and $C_i$ to the predicate of the game $\mathcal{G}$.

The sequence of channels $\{\mathcal{N}_i\in\CPTP(R\bar{E}_{i-1}, S_{A,i}\bar{S}_{B,i}C_i{R}\bar{E}_i)\}_{i=1}^n$ is no-signalling in the sense that Eve's operations are unable to act on the registers $R_AR_B$ in the devices of Alice and Bob.
Formally, the channel $\mathcal{R}_i\in\CPTP(\bar{E}_{i-1},\bar{E}_i)$ which samples $X_i,Y_i$ and $T_i$ and subsequently updates $\bar{E}_{i-1}$ as in steps $1$ and $4$ of the map $\mathcal{N}_i$ above, trivially satisfies
\begin{align}\label{no-signalling}
    \Tr_{S_{A,i}\bar{S}_{B, I} C_iR}\circ \mathcal{N}_i = \mathcal{R}_i\circ\Tr_R.
\end{align}

\subsubsection{Min-tradeoff function and GEAT}\label{sec_min_tradeoff}
Now, let $g:\mathcal{P}(\{0,1\})\rightarrow\mathbb{R}$ be an arbitrary affine function satisfying
\begin{align}\label{affine_g}
    g(q) \leq \min_{\omega_{R\bar{E}_{i-1}F}} \big\{H(S_{A,i}|\bar{E}_{i-1}F)_{\mathcal{N}_i(\omega_{R\bar{E}_{i-1}F})}\,:\, \Tr_{S_{A,i}R\bar{E}_{i}F}[\mathcal{N}_i^\textup{test}(\omega_{R\bar{E}_{i-1}F})]=q_0\ketbra{0}+q_1\ketbra{1}\big\},
\end{align}
for all $i\in\{1,...,n\}$ and all $[q_0,q_1]\in\mathcal{P}(\{0,1\})$, as well as $\Max(g)=g(\delta_1)$\footnote{We remark that the following proof is independent of the choice of min-tradeoff function, which in particular may be optimized when considering specific non-local game primitives. 
To retrieve the statement in \thref{th_secrecy} we use the affine function in Eq.~\eqref{affine} which works for arbitrary non-local games for which $\othree<\otwo$, but to obtain tighter key rates for the MSG in Section \ref{sec_msg_keyrates} we use the affine functions constructed in Sections \ref{sec_msg_min_ent} and \ref{sec_msg_vn}.}.
Proceeding as in \cite{DF19, MFSR22, MR22}, \thref{th_infrequent_sampling} implies that the function $f:\mathcal{P}(\{0,1,\perp\})\rightarrow\mathbb{R}$ defined by 
\begin{align}\label{tradeoff_f}
    f(\delta_0)=g(\delta_1)+\frac{1}{\gamma}(g(\delta_0)-g(\delta_1)),\qquad f(\delta_1)=f(\delta_\bot)=g(\delta_1)
\end{align}
is an affine min-tradeoff function for the infrequent sampling channels\footnote{This follows directly from the observation that if ${S}_{A,i}$ and $\bar{S}_{B,i}$ are classical, then $H({S}_{A,i}\bar{S}_{B,i}|\bar{E}_i)\geq H({S}_{A,i}|\bar{E}_i).$} $\{\mathcal{N}_i=\gamma\mathcal{N}_i^\textup{test} + (1-\gamma)\mathcal{N}_i^\textup{gen}\}_i$.
Using the affinity of $f$, for any $c_1^n\in{\Omega}_{\PE}$:
\begin{align}
    f(\textup{Freq}_{{c}_1^n}(\cdot))
    &=\textup{Freq}_{{c}_1^n}(0)f(\delta_0) + (1-\textup{Freq}_{{c}_1^n}(0))f(\delta_1)\\
    &=\frac{\textup{Freq}_{{c}_1^n}(0)}{\gamma}g(\delta_0)+\left(1-\frac{\textup{Freq}_{{c}_1^n}(0)}{\gamma}\right)g(\delta_1),
\end{align}
where in the last equality we use the definition of $f$ in Eq.~\eqref{affine_f}.
By definition of ${\Omega}_{\PE}$, we must have $\textup{Freq}_{{c}_1^n}(0)\leq (1-\omega_\textup{exp}+\delta_\textup{tol})\cdot\gamma$ which implies $1-\frac{1}{\gamma}\cdot\textup{Freq}_{{c}_1^n}(0)\geq\omega_\textup{exp}-\delta_\textup{tol}$.
The affinity of $f$ then gives
\begin{align}
    f(\textup{Freq}_{{c}_1^n}(\cdot))
    &\geq \frac{\textup{Freq}_{{c}_1^n}(0)}{\gamma}g(\delta_0)+(\omega_\textup{exp}-\delta_\textup{tol})g(\delta_1)\\
    &= g\Bigg(\frac{\textup{Freq}_{{c}_1^n}(0)}{\gamma}\cdot\delta_0+(\omega_\textup{exp}-\delta_\textup{tol})\cdot\delta_1\Bigg)\\
    &= g(\omega_\textup{exp}-\delta_\textup{tol})\label{min_tradeoff_bound},
\end{align}
where in the last equality we slightly abuse notation to turn the argument to $g$ from a distribution $[q_0,q_1]\in\mathcal{P}(\{0,1\})$ into the winning probability $p$, which equals one when $[q_0,q_1]=\delta_1$ and zero when $[q_0,q_1]=\delta_0$.

The generalised entropy accumulation theorem \cite[Corollary 4.6]{MFSR22} then implies:
\begin{lemma}\thlabel{th_GEAT_testing}
    Let $\omega^0_{RQ_E}\in S(RQ_E)$ be the input state to the sequence of channels 
    \begin{align}
        \{\mathcal{N}_i\in\CPTP(R\bar{E}_{i-1}, S_{A,i}\bar{S}_{B,i}C_i{R}\bar{E}_i)\}_{i=1}^n
    \end{align}
    defined in Section \ref{sec_app_geat}.
    Define the event $\Omega_{\PE}=\{|\{i:C_i=0\}|\leq(1-\oexp+\dtol)\cdot\gamma n\}$ and let $f:\mathcal{P}(\{0,1,\perp\})\rightarrow\mathbb{R}$ and $g:\mathcal{P}(\{0,1\})\rightarrow\mathbb{R}$ be the functions respectively defined in Eq.~\eqref{tradeoff_f} and Eq.~\eqref{affine_g}.
    
    Then for any $\varepsilon_s'\in\big[0,\frac{1}{2}\epssec\big)$:
    \begin{align}
        H^{\varepsilon_s'}_{\min}(S_A\bar{S}_B|\bar{E}_n)_{\sigma} \geq ng(\oexp-\delta_\textup{tol}) -d_1\sqrt{n} - d_0,
    \end{align}
    where $\sigma_{S_A\bar{S}_BC_1^nR\bar{E}_n}=\mathcal{N}_n\circ\cdots\circ\mathcal{N}_1(\omega^0_{RQ_E})$ and the constants $d_1$ and $d_0$ are 
    \begin{align}
        d_1 
        &:= \sqrt{\frac{2\ln (2)V^2}{\eta}\left( \vartheta({\varepsilon_s'}) + (2-\eta)\log(\frac{1}{\Pr[\Omega_{\PE}]_{\sigma}})\right)},\\
        d_0
        &:= \frac{(2-\eta)\eta^2\log(\frac{1}{\Pr[\Omega_{\PE}]_{\sigma}})+\eta^2 \vartheta({\varepsilon_s'})}{3(\ln 2)^2 V^2 (2\eta-1)^3} 2^{\frac{1-\eta}{\eta}(2 +{\Max}(f)-{\Min}_\Sigma(f))} \ln^3\Big(2^{2+ {\Max}(f)-{\Min}_\Sigma(f) } + \textrm{e}^2\Big)
    \end{align}
    with 
    \begin{align}
        \eta = \frac{2\ln2}{1+2\ln2},\qquad \vartheta(\varepsilon_s')=-\log(1-\sqrt{1-(\varepsilon_s')^2}),\qquad V=\log(17)+\sqrt{2+\textup{Var}(f)},
    \end{align}
    and where we recall $\Pr[\Omega_{\PE}]_{\sigma}\geq \varepsilon_a\in\big(0,\frac{1}{2}\epssec\big]$.
\end{lemma}

\subsubsection{Guessing probability per testing round}\label{sec_guessing_probability}
In this section, we construct a simple affine function that bounds the amount of min-entropy generated during each measurement round with a sufficiently high winning probability and serves as the function $g:\mathcal{P}(\{0,1\})\rightarrow\mathbb{R}$ in \thref{th_GEAT_testing}.
The method presented here applies to all suitable non-local games as in the framework of Section \ref{sec_NLG} that satisfy the monogamy-of-entanglement property $\othree<\otwo$, where $\otwo$ is the quantum value of $\mathcal{G}_2$ and $\othree$ is the quantum value of the tripartite extension.

Fix a strategy for the $2$-player non-local game 
\begin{align}
    \mathcal{S}_2 = \Big\{\nu_{Q_AQ_B}, \big\{\{P^x_a\}_{a\in \mathcal{A}}\big\}_x, \big\{\{Q^y_b\}_{b\in \mathcal{B}}\big\}_y \Big\}
\end{align}
which allows Alice and Bob to win with probability at least $\oexp$.
That is, letting $\WAB$ be the event that Alice and Bob win an instance of the non-local game, we have $\Pr[\WAB]_{\mathcal{S}_2}\geq \oexp$ for some $\oexp\in[\omega_3,\omega_2]$, where the subscript ${\mathcal{S}_2}$ denotes that the probability is over the output distribution using strategy ${\mathcal{S}_2}$.

Add the third player Eve and fix a strategy 
\begin{align}
    \mathcal{S}_3 = \Big\{\nu_{Q_AQ_BE}, \big\{\{F^{xy}_c\}_{c\in\mathcal{C}}\big\}_{(x,y)}, \big\{\{P^x_a\}_{a\in \mathcal{A}}\big\}_x, \big\{\{Q^y_b\}_{b\in \mathcal{B}}\big\}_y \Big\},
\end{align}
which extends $\mathcal{S}_2$.
Namely, Alice and Bob's measurements are the same in $\mathcal{S}_2$ and $\mathcal{S}_3$, and $\nu_{Q_AQ_BE}\in\mathcal{S}_3$ is a quantum extension of $\nu_{Q_AQ_B}\in\mathcal{S}_2$.
Let $\GE$ be the event that Eve guesses the bit $S_A$ sifted by Alice from her output and both inputs to Alice and Bob in one instance of the non-local game. 
Since strategies of the form ${\mathcal{S}_3}$ are a subset of all possible tripartite strategies, we have that $\Pr[W_\textup{AB}\wedge \GE]_{\mathcal{S}_3} \leq\omega_3$.
As before, the subscript ${\mathcal{S}_3}$ indicates that the probability is over the output distribution using strategy ${\mathcal{S}_3}$.

Now, consider the probability of event $G^c_\textup{E}$ conditioned on $W_\textup{AB}$:
\begin{align}\label{GEcompliment}
    \Pr[\GEc|\WAB]_{\mathcal{S}_3}=\frac{\Pr[\GEc\wedge\WAB]_{\mathcal{S}_3}}{\Pr[\WAB]_{\mathcal{S}_2}}\leq\frac{\Pr[\GEc]_{\mathcal{S}_3}}{\oexp},
\end{align}
where we have noted that $\Pr[\WAB]_{\mathcal{S}_3}=\Pr[\WAB]_{\mathcal{S}_2}$.
It follows that
\begin{align}
    \oexp(1-\Pr[\GE |\WAB]_{\mathcal{S}_3})\leq 1-\Pr[\GE]_{\mathcal{S}_3},
\end{align}
and after rearranging we get
\begin{align}
    \Pr[\GE]_{\mathcal{S}_3} \leq 1- \oexp + \Pr[\GE\wedge\WAB]_{\mathcal{S}_3} \leq 1-\oexp + \omega_3,
\end{align}
which readily implies $-\log \Pr[\GE]_{\mathcal{S}_3}\geq -\log(1-\oexp+\omega_3)$.

We wish to find an affine lower bound to $-\log(1-\oexp+\omega_3)$.
Since this is convex, a tangent line at any point $\beta\in[\omega_3,\omega_2]$ suffices. 
This yields the following affine lower bound $g:[\othree,\otwo]\rightarrow\mathbb{R}$ to $-\log\Pr[\GE]_{\mathcal{S}_3}$:
\begin{align}\label{affine}
    {g}(\oexp):=\frac{\oexp-\beta}{\ln2\cdot(1-\beta+\omega_3)} - \log\big(1-\beta+\omega_3\big).
\end{align}
Then, the function $g(p)$ linearly extended to the domain $p\in[0,1]$ is such that $-\log\Pr[\GE]_{\mathcal{S}_3}\geq g(p)$, and further $g$ achieves its maximum when $p=1$.
In particular, note that \thref{th_infrequent_sampling} implies that the function $f:\mathcal{P}(\{0,1,\perp\})\rightarrow\mathbb{R}$ defined in Eq.~\eqref{affine_f} satisfies
\begin{align}
    \textup{Max}(f) &= g(1),\\
    \textup{Min}(f) &= \Big(1-\frac{1}{\gamma}\Big)g(1)+\frac{1}{\gamma}g(0),\\
    \textup{Min}_{\Sigma}(f)&\geq g(0),\\
    \textup{Var}(f) &\leq\frac{1}{\gamma}\big(g(1)-g(0)\big)^2.
\end{align}
Substituting these values in the constants $d_1,d_0$ of \thref{th_GEAT_testing} gives
\begin{align}\label{d1_def}
    d_1 
    &\leq \sqrt{\frac{2\ln (2)V^2}{\eta}\left( \vartheta({\varepsilon_s'}) + (2-\eta)\log(\frac{1}{\varepsilon_a})\right)},\\
    d_0
    &\leq \frac{(2-\eta)\eta^2\log(\frac{1}{\varepsilon_a})+\eta^2 \vartheta({\varepsilon_s'})}{3(\ln 2)^2 V^2 (2\eta-1)^3} 2^{\frac{1-\eta}{\eta}(2 +g(1)-g(0))} \ln^3\Big(2^{2+ g(1)-g(0) } + \textrm{e}^2\Big)\label{d0_def},
\end{align}
where we recall $\varepsilon_a\in\big(0,\frac{1}{2}\epssec\big]$ is such that $\Pr[\Omega_{\PE}]_{\sigma}\geq \varepsilon_a$.
Further, 
\begin{align}
    \eta = \frac{2\ln2}{1+2\ln2},\qquad \vartheta(\varepsilon_s')=-\log(1-\sqrt{1-(\varepsilon_s')^2}),\qquad V\leq\log(17)+\sqrt{2+\frac{1}{\gamma}\big(g(1)-g(0)\big)^2},
\end{align}
and from $g$ as defined in Eq.~\eqref{affine} we have:
\begin{align}
    g(1) &= \frac{1-\beta}{\ln2\cdot(1-\beta+\omega_3)} - \log\big(1-\beta+\omega_3\big),\\
    g(0) &= \frac{-\beta}{\ln2\cdot(1-\beta+\omega_3)} - \log\big(1-\beta+\omega_3\big),
\end{align}
for a variable $\beta\in[\othree,\otwo]$ to be optimised.

The assumption $\Pr[\WAB\wedge\GE]_{\mathcal{S}_3}\leq\omega_3$ marks a point where this analysis may be improved.
Specifically, when dealing with non-local games that exhibit monogamy of entanglement it is expected that the probability of Alice, Bob and Eve winning an instance of $\mathcal{G}_3$ decreases monotonically if we increase the expected bipartite winning probability $\oexp\in[\omega_3, \omega_2]$.
Further, if the optimal strategy for Alice and Bob in an instance of $\mathcal{G}_2$ self-tests maximally entangled states, then in the extreme $\Pr[\WAB]_{\mathcal{S}_2}\geq\omega_2$ we expect $\Pr[\WAB\wedge\GE]_{\mathcal{S}_3}\leq\frac{\omega_2}{2}$, as Eve's best strategy is to produce a uniform guess.

Explicitly, suppose instead that whenever $\Pr[\WAB]_{\mathcal{S}_2}\geq \oexp$ for $\oexp\in[\omega_3,\omega_2]$ we have that $\Pr[\WAB\wedge\GE]_{\mathcal{S}_3}\leq h(\oexp)$ for a monotonically decreasing, piece-wise differentiable function $h$.
The same analysis as in Eq.~\eqref{GEcompliment} shows
\begin{align}\label{min_ent_bound}
    \Pr[\GE]_{\mathcal{S}_3}\leq 1-\oexp + h(\oexp),
\end{align}
so a suitable affine lower bound to $-\log\Pr[\GE]$ is given by a tangent at some $\beta\in[\omega_3,\omega_2]$
\begin{align}\label{affine_constrained}
    {g}_\textup{cons}(\oexp):=\frac{(1-h'(\beta))\cdot(\oexp - \beta)}{\ln2\cdot(1-\beta+h(\beta))} - \log\big(1-\beta+h(\beta)\big),
\end{align}
where the subscript `cons' emphasizes that $\mathcal{G}_3$ was constrained by Alice and Bob's winning probability in an instance of $\mathcal{G}_2$.
In practice, the function $h$ bounding the winning probability of Alice, Bob and Eve in an instance of $\mathcal{G}_3$ constrained on Alice and Bob's winning probability in $\mathcal{G}_2$ may be computed numerically with the NPA hierarchy \cite{NPA08,PNA10,Wittek15} ---see Section \ref{sec_msg_min_ent}.

Now, we conclude that the functions $g$ and $g_\textup{cons}$ satisfy the definition of Eq.~\eqref{g_function_def}.
Indeed, letting $S_A = SK_A(A,X,Y)$ be the bit Alice generates from her output $A$ and inputs $X,Y$ in one instance of the non-local game primitive, we have that for any probability distribution $q=[1-\oexp,\oexp]$:
\begin{align}\label{pseudo_mintradeoff1}
    &\min_{\omega_{RE}}\big\{H(S_A|E)_{\mathcal{N}_i(\omega_{RE})}: \Tr_{S_ARE}[\mathcal{N}_i^\textup{test}(\omega_{RE})] = (1-\oexp)\ketbra{0}_{C_i}+\oexp\ketbra{1}_{C_i}\big\} \\
    &\;\; \geq\min_{\omega_{RE}}\big\{H_{\min}(S_A|E)_{\mathcal{N}_i(\omega_{RE})}: \Tr_{S_ARE}[\mathcal{N}_i^\textup{test}(\omega_{RE})] = (1-\oexp)\ketbra{0}_{C_i}+\oexp\ketbra{1}_{C_i}\big\} \\
    &\;\; \geq g_\textup{cons}(\oexp) \geq g(\oexp)\label{pseudo_mintradeoff2},
\end{align}
where we have identified register $R=R_AR_B$ with the registers $Q_AQ_B$ shared amongst Alice and Bob in strategy $\mathcal{S}_2$ and noted that since the strategy $\mathcal{S}_3$ is over arbitrary extensions $\nu_{Q_AQ_BE}$ of $\nu_{Q_AQ_B}$ so in particular includes the case where $E=E'F$ for a purifying register $F$ that gets acted on trivially by Eve's instruments, as per the definition in Eq.~\eqref{g_function_def}.
Further, the channel $\mathcal{N}_i$ used above is as in Eq.~\eqref{infrequent_sampling_channel} such that the strategy $\mathcal{S}_2$ corresponds to the quantum instruments performed by the devices of Alice and Bob in an application of $\mathcal{N}_i$.
Lastly, on the last inequality above we shift the argument of $g$ from a probability distribution on $\mathcal{P}(\{0,1\})$ to the winning probability.

\subsubsection{Final secrecy bound}\label{sec_final_bounds}
Finally, we are ready to lower bound the smooth min-entropy $H_{\min}^{\varepsilon_s'}(S_A\bar{S}_B|\bar{E})_{\sigma_{|{\Omega}_{\PE}}}$ in the right hand side of Eq.~\eqref{min_max_entropy_chain_rule} and conclude the secrecy proof.

\thref{th_GEAT_testing} with $g$ as in the previous section gives
\begin{align}
    H_{\min}^{\varepsilon_s'}(S_A\bar{S}_B|\bar{E})_{\sigma_{|{\Omega}_{\PE}}} \geq n{g}(\omega_\textup{exp}-\delta_\textup{tol}) -d_1\sqrt{n}-d_0,
\end{align}
for the constants $d_1,d_0$ as in Eqs.~\eqref{d1_def} and~\eqref{d0_def} respectively.
Together with Eqs.~\eqref{min_entropy_bound_1} and \eqref{min_max_entropy_chain_rule} we have 
\begin{align}
    H_{\min}^{\varepsilon_s}(S_A|{E}_n)_{\sigma_{|{\Omega}_{\PE}\wedge\Omega_{\EC}}} 
    \geq\; &n{g}(\omega_\textup{exp}-\delta_\textup{tol}) -d_1\sqrt{n}-d_0 - (2-\oexp+\delta_\textup{tol})n(\gamma+\kappa)  \\
    &- 2\vartheta(\delta) -l_{\EC}-\lambda_{\EC}-2,
\end{align}
where we recall $\varepsilon_s\in\big(0,\frac{1}{2}\epssec\big)$, $\varepsilon_s'\in\big[0,\frac{1}{2}\epssec\big)$ and $\varepsilon_s''\in\big[0,\frac{1}{4}\epssec\big)$ are such that $\delta = \varepsilon_s-\varepsilon_s'-2\varepsilon_s'' >0$, and $\kappa\geq\frac{1}{\sqrt{n}}\sqrt{-\ln{\varepsilon_s''}}$.
Then, the Leftover Hashing Lemma (\thref{th_leftover_hashing}) implies
\begin{align}
    &\norm{\rho_{K_AE | {\Omega}_{\PE}\wedge\Omega_{\EC} } - \tau_{K_A}\otimes\rho_{E| {\Omega}_{\PE}\wedge\Omega_{\EC} }}_{\Tr} \\
    &\qquad\quad\quad\;\;\leq \frac{1}{2}\cdot 2^{-\frac{1}{2}\big(n{g}(\omega_\textup{exp}-\delta_\textup{tol}) -d_1\sqrt{n}-d_0 - (2-\oexp+\delta_\textup{tol})n(\gamma+\kappa)   - 2\vartheta(\delta)
    -l_{\EC}-\lambda_{\EC}-2-\lkey\big)}+2{\varepsilon_s}.
\end{align}
Finally, merging with Eq.~\eqref{sececy_trivial} and the conditions $\Pr[\Omega_{\g}|{\Omega}_{\PE}]\geq \varepsilon_s^2$ and $\Pr[{\Omega}_{\PE}]\geq\varepsilon_a\in\big(0,\frac{1}{2}\epssec\big]$, we have secrecy with parameter $\max\big\{2\varepsilon_a, 2\varepsilon_s^2, \varepsilon_\textup{sec}\big\}= \epssec$ provided 
\begin{align}
    2^{-\frac{1}{2}\big(n{g}(\omega_\textup{exp}-\delta_\textup{tol}) -d_1\sqrt{n}-d_0 - (2-\oexp+\delta_\textup{tol})n(\gamma+\kappa)   - 2\vartheta(\varepsilon_s-\varepsilon_s'-2\varepsilon_s'')-l_{\EC}-\lambda_{\EC}-\lkey\big)}+2{\varepsilon_s} \leq \epssec,
\end{align}
which implies 
\begin{align}\label{lkey_proof}
    l_\textup{key} \leq \;&n{g}(\omega_\textup{exp}-\delta_\textup{tol}) -d_1\sqrt{n}-d_0 - (2-\oexp+\delta_\textup{tol})n(\gamma+\kappa)   - 2\vartheta(\varepsilon_s-\varepsilon_s'-2\varepsilon_s'')\\
        &-l_{\EC}-\lambda_{\EC} 
        - 2\log\frac{1}{\epssec-2\varepsilon_s}.
\end{align} 

In order to maximize the finite key rates given by $\frac{1}{n}\lkey$, we wish to minimize all the terms in the right hand side which depend on the free variables $\varepsilon_a,\varepsilon_s,\varepsilon_s',\varepsilon_s''$.

First, we tackle the coefficients $d_1,d_0$ respectively as in Eqs.~\eqref{d1_def}, and~\eqref{d0_def}.
Note that $d_1$ asymptotically scales as $O(n^{-1/2})$ in the expression for the key rate, whereas $d_0$ scales as $O(n^{-1})$. 
Without loss of generality we can fix $\varepsilon_a=\frac{1}{2}\epssec$ since $\argmin_{\varepsilon_a}(d_0)=\argmin_{\varepsilon_a}(d_1)=\frac{1}{2}\epssec$, as $\varepsilon_a$ does not appear elsewhere in Eq.~\eqref{lkey_proof}.
To further minimize $d_1,d_0$, we need $\vartheta(\varepsilon_s')\in\mathbb{R}^+\cup\{0\}$ to be small, which requires $\varepsilon_s'\in\big[0,\frac{1}{2}\epssec\big)$ to be large.
Moreover, the term $\kappa$ which also scales as $O(n^{-1/2})$ in the key rate is minimized when $\varepsilon_s''\in\big[0,\frac{1}{4}\epssec\big)$ is large.
Further, $\vartheta(\varepsilon_s-\varepsilon_s'-2\varepsilon_s'')$ which scales as $O(n^{-1})$ in $\frac{1}{n}\lkey$, is smallest when $\varepsilon_s-\varepsilon_s'-2\varepsilon_s''\in\big(0,\frac{1}{2}\epssec\big)$ is large. 
This requires both $\varepsilon_s\in\big(0,\frac{1}{2}\epssec\big)$ to be large and $\varepsilon_s',2\varepsilon_s''\in\big[0,\frac{1}{2}\epssec\big)$ to be small.
Lastly, the $\log\frac{1}{\epssec-2\varepsilon_s}$ term vanishes as $O(n^{-1})$ in the key rate and is minimized when $\varepsilon_s$ is close to $0$.

These conditions impose a tradeoff between $\varepsilon_s,\varepsilon_s',\varepsilon_s''$ in order to maximize the finite key rates.
For pedagogical purposes, it suffices to set $\varepsilon_a=\frac{1}{2}\epssec$, $\varepsilon_s'=\frac{1}{2}\varepsilon_s$ and $\varepsilon_s''=\frac{1}{8}\varepsilon_s$, in which case we obtain the expression from the statement of \thref{th_secrecy}
\begin{align}
    l_\textup{key} \leq \;&n{g}(\omega_\textup{exp}-\delta_\textup{tol}) -d_1\sqrt{n}-d_0 - (2-\oexp+\delta_\textup{tol})n(\gamma+\kappa)   - 2\vartheta\left(\frac{\varepsilon_s}{4}\right)\\
        &-l_{\EC}-\lambda_{\EC} 
        - 2\log\frac{1}{\epssec-2\varepsilon_s},
\end{align} 
wherein $\kappa\geq\frac{1}{\sqrt{n}}\sqrt{\ln{8}-\ln{\varepsilon_s}}$, and
\begin{align}\label{d1_def_final}
    d_1 
    &\leq \sqrt{\frac{2\ln (2)V^2}{\eta}\left( \vartheta\left({\frac{\varepsilon_s}{2}}\right) + (2-\eta)\log(\frac{2}{\epssec})\right)},\\
    d_0
    &\leq \frac{(2-\eta)\eta^2\log(\frac{2}{\epssec})+\eta^2 \vartheta({\frac{\varepsilon_s}{2}})}{3(\ln 2)^2 V^2 (2\eta-1)^3} 2^{\frac{1-\eta}{\eta}(2 +g(1)-g(0))} \ln^3\Big(2^{2+ g(1)-g(0) } + \textrm{e}^2\Big)\label{d0_def_final}.
\end{align}
Further, 
\begin{align}
    \eta = \frac{2\ln2}{1+2\ln2},\qquad V\leq\log(17)+\sqrt{2+\frac{1}{\gamma}\big(g(1)-g(0)\big)^2},
\end{align}
and from $g$ as defined in Eq.~\eqref{affine} we have:
\begin{align}
    g(1) &= \frac{1-\beta}{\ln2\cdot(1-\beta+\omega_3)} - \log\big(1-\beta+\omega_3\big),\\
    g(0) &= \frac{-\beta}{\ln2\cdot(1-\beta+\omega_3)} - \log\big(1-\beta+\omega_3\big),
\end{align}
for a variable $\beta\in[\othree,\otwo]$ to be optimised.
This concludes the proof.

We remark that the finite key rates for the magic square game in Figure \ref{fig:MSG-DIQKD_keyrates} are generated through numerical optimization of all free variables $\varepsilon_a,\varepsilon_s,\varepsilon_s',\varepsilon_s''$, but the particular choice of variables above does not incur discernible changes in the finite key rates.

\section{Numerical optimizations for the magic square game}\label{sec_msg}
In this section, we show that the magic square game satisfies $\omega_3<\omega_2$, and subsequently construct bounds on the min-entropy and von Neumann entropy of Alice's generated bit conditioned on an eavesdropper's information.

\subsection{Tripartite value of the magic square game}

Recall the bipartite magic square game $\msgt$, in which Alice and Bob respectively receive uniform inputs $x,y\in\{0,1,2\}$,  produce bit triples $a$ of even Hamming weight and $b$ of odd Hamming weight, and win whenever the $y$-th bit of $a$ equals the $x$-th bit of $b$ ---symbolically $a[y]=b[x]$.

It is known \cite{Mermin90, Peres90} that the maximal winning probability using classical strategies is $\omega_\textup{C}(\msgt)=\frac{8}{9}$ and that the maximal winning probability using a quantum strategy is $\omega_\textup{Q}(\msgt)=1$. 
The strategies achieving each value are outlined at the beginning of Section \ref{sec_msg_keyrates}.

We consider a tripartite extension to the standard magic square game, which we label $\msg$, in which the third player (for now, let us blissfully call him Charlie) receives both inputs $x,y$ and produces a single bit $c\in\{0,1\}$ required to satisfy $c=a[y]=b[x]$.

In this section we compute bounds on the min-entropy and on the von Neumann entropy of Alice's sifted bit $a[y]$ conditioned on the third player's side information.
The main numerical tools used in this section is the NPA hierarchy \cite{NPA08,PNA10,Wittek15} which we briefly describe in the following.
Consider a polynomial optimization problem with non-commuting variables of the form
\begin{align}
    p_\textup{opt} \; := \; &\underset{\mathcal{H}, X, \ket{\phi}}{\text{inf }} \bra{\phi}p(X)\ket{\phi} \label{poly_opt}\\
    \text{s.t.:} \qquad &g_i(X)\geq 0 \quad i={1,...,m}, \\
    &\norm{\phi} = 1,
\end{align}
where $p$ and $g_i$ are Hermitian polynomials in the algebra generated by $2n$ non-commuting variables $(x_1,...,x_n,x^\dagger_1,...,x_n^\dagger)$. 
The optimization is over the Hilbert space $\mathcal{H}$, the state $\kettext{\phi}\in \mathcal{H}$ and the variables $X=(X_1,...,X_n)$ which are bounded operators in $\mathcal{H}$. 
The NPA hierarchy is a sequence of semidefinite programs (SDP) with increasing relaxations which monotonically converge to the optimal value in Eq.~\eqref{poly_opt}.
In short, if $d\in\mathbb{N}$ is the level of the hierarchy\footnote{We note that in principle it is possible to optimise SDPs that lie between the discrete levels of the hierarchy, depending on the available computational power.} and $p_d\in\mathbb{R}$ is the optimum value of the SDP at level $d$ of the hierarchy, then 
\begin{align}
    &p_{d'} \geq p_d \quad \forall d'\geq d \\
    &\underset{d\rightarrow\infty}{\lim}p_d = p_\textup{opt}.
\end{align}
We also reference the \textbf{ncpol2sdpa} python library \cite{Wittek15} written by Wittek and maintained by Brown, which implements the NPA hierarchy and which we use throughout this section (see our python scripts in \cite{Cervero23}).

We begin with an analysis on the classical, quantum and no-signalling values of the tripartite $\msg$.
The classical strategy achieving $\omega_\textup{C}(\msgt)=\frac{8}{9}$ can be used to infer that $\omega_\textup{C}(\msg)=\frac{8}{9}$.
Further, we use the linear program in \cite[Section 3]{Toner08} to show $\omega_\textup{NS}(\msg)=1$.
Now, a general quantum strategy for $\msg$ is given by the tuple
\begin{align}
    \mathcal{S} = \Big\{\nu_{Q_AQ_BQ_C}, \big\{\{F^{xy}_c\}_{c\in\{0,1\}}\big\}_{(x,y)}, \big\{\{P^x_a\}_{a\in \mathcal{A}}\big\}_x, \big\{\{Q^y_b\}_{b\in \mathcal{B}}\big\}_y \Big\},
\end{align}
where $\nu_{Q_AQ_BQ_C}\in S(Q_AQ_BQ_C)$ is a tripartite state, $\{\{P^x_a\}_{a\in \mathcal{A}}\}_x,$ and $\{\{Q^y_b\}_{b\in \mathcal{B}}\}_y$ are POVMs in Alice and Bob's spaces indexed over respective alphabets $\mathcal{A} =\{000,011,101,110\}$ and $\mathcal{B} = \{001,010,100,111\}$, and $ \{\{F^{xy}_c\}_{c\in\{0,1\}}\}_{(x,y)}$ are POVMs in Charlie's space.

The quantum value of MSG$_3$ is then
\begin{align}\label{qvalue_msg}
    \omega_\textup{Q}(\msg) = \underset{\mathcal{S}}{\text{sup}} \sum_{x,y}\frac{1}{9}\text{Tr}\Big[\nu_{Q_AQ_BQ_C} \Pi^{xy}\Big],
\end{align}
where
\begin{align}\label{pi_xy}
    \Pi^{xy} := \sum_{c\in\{0,1\}} \bigg[ F_c^{xy} \otimes \sum_{a\in \mathcal{A}: a[y] = c} P^x_a \otimes \sum_{b\in \mathcal{B}: b[x] = c} Q^y_b \bigg].
\end{align}
Without loss of generality we can assume that all POVMs in $\mathcal{S}$ are projective by Naimark's Theorem, and further that $\nu_{Q_AQ_BQ_C}=\ketbratext{\phi}{\phi}_{Q_AQ_BQ_C}$ is pure since the dimension of $Q_AQ_BQ_C$ is arbitrary.
Eq.~\eqref{qvalue_msg} therefore becomes
\begin{align}\label{qvalue_pure}
    \omega_\textup{Q}(\msg) = \sup_{\mathcal{S}}\frac{1}{9} \sum_{x,y} \bra{\phi} \Pi^{xy}\ket{\phi},
\end{align}
and the optimization is over the Hilbert space $Q_AQ_BQ_C$, the state $\kettext{\phi}\in Q_AQ_BQ_C$ and the PVMs on Alice, Bob and Charlie's systems.

We use the NPA hierarchy to show (see Appendix \ref{app_msg_npa} and \cite{Cervero23}):
\begin{theorem}\thlabel{th_msg}
    For the $\msg$ game considered above, we have 
    \begin{align}
        \omega_\textup{Q}(\msg)\lesssim \frac{8.00077}{9}.
    \end{align}
\end{theorem}

{This is a numerical upper bound on the winning probability, supporting the conjecture that the value is in fact exactly $\frac89$.}
The intuition behind this result is that the quantum strategy yielding $\omega_\textup{Q}(\msgt)=1$ requires Alice and Bob to share maximally entangled states ---in fact, the optimal quantum bipartite strategy self-tests a pair of maximally entangled states \cite{WBMS16}.
Monogamy of entanglement then suggests that the correlations shared by Alice, Bob and Charlie in an instance of $\msg$ are not enough to obtain a value $\omega_\textup{Q}(\msg)$ larger than $\omega_\textup{C}(\msg)$.

\subsection{Bound on min-entropy}\label{sec_msg_min_ent}

The goal of this section is to compute the function $h:[8/9,1]\rightarrow\mathbb{R}$ in Eq.~\eqref{affine_constrained} to lower bound the min-entropy of Alice's bit $A[Y]$ conditioned on the third player's information, generated in one round of the MSG.
The function $h$ is used in the min-tradeoff function in Eq.~\eqref{affine_constrained} to provide the finite key length in Eq.~\eqref{l_key} of the MSG-DIQKD protocol.

To bound the min-entropy of Alice's bit conditioned on third player (now turned Eve), we are interested in the related optimization problem given by
\begin{align}\label{msg_constrained}
    &  \sup \frac{1}{9} \sum_{x,y} \bra{\phi} \Pi^{xy}\ket{\phi} \\
    &\text{s.t.: } \{\Pax\}_a, \{\Qby\}_b, \{\Fcxy\}_c \text{ are PVMs in different labs for all }x,y, \\
    &\quad\quad \bra{\phi}\frac{1}{9}\sum_{xy}\Pi^{xy}_{AB}\ket{\phi}\geq \oexp ,
\end{align}
where $\oexp\in[8/9,1]$ and 
\begin{align}
    \Pi^{xy}_{AB} := \sum_{c\in\{0,1\}} \mathbb{I}_{Q_C}\otimes \sum_{a\in \mathcal{A}: a[y] = c} P^x_a \otimes \sum_{b\in \mathcal{B}: b[x] = c} Q^y_b
\end{align}
is the projection yielding a winning outcome for Alice and Bob.
This is precisely the optimization problem resulting from fixing a bipartite strategy 
\begin{align}
    \mathcal{S}_2 = \Big\{\nu_{Q_AQ_B}, \big\{\{P^x_a\}_{a\in \mathcal{A}}\big\}_x, \big\{\{Q^y_b\}_{b\in \mathcal{B}}\big\}_y \Big\}
\end{align}
which allows Alice and Bob to win with probability at least $\oexp$ and subsequently considering a tripartite strategy 
\begin{align}
    \mathcal{S}_3 = \Big\{\nu_{Q_AQ_BE}, \big\{\{F^{xy}_c\}_{c\in\mathcal{C}}\big\}_{(x,y)}, \big\{\{P^x_a\}_{a\in \mathcal{A}}\big\}_x, \big\{\{Q^y_b\}_{b\in \mathcal{B}}\big\}_y \Big\},
\end{align}
which extends $\mathcal{S}_2$ ---that is $\mathcal{S}_3 = \mathcal{S}_2\cup\big\{\{\{F^{xy}_c\}_c\}_{(x,y)}\big\}$ and $\Tr_{Q_C}[\nu_{Q_AQ_BE}]=\nu_{Q_AQ_B}$.
We use the NPA hierarchy to bound these probabilities constrained on a range of values of $\oexp\in[8/9,1]$ and display the results in Figure \ref{fig:MSG_cons_w_upperbound}\footnote{
We remark that the values in the figure are obtained via relaxation of the NPA containing less extramonomials than the value in \thref{th_msg}.
For reference, each semidefinite program in Figure \ref{fig:MSG_cons_w_upperbound} takes roughly 12h on a desktop with 256GB of RAM using the scripts in \cite{Cervero23}, whereas the $8.00077/9$ bound obtained in \thref{th_msg} takes about $3$ days.
}.

We remark that the optimums in the optimization problem in Eq.~\eqref{msg_constrained} form a monotonically decreasing and concave curve.
Monotonically decreasing because increasing $\oexp\in[8/9,1]$ decreases the size of the feasible set of the optimization problem, and concave because the straight line connecting the optimums given by any pair of constraints $\omega_\textup{exp}^1 <\omega_\textup{exp}^2$ on Alice and Bob's winning probability is a lower bound on the optimal value for any $\oexp\in(\omega_\textup{exp}^1,\omega_\textup{exp}^2)$\footnote{Specifically, a convex combination of the strategies given by the pair of constraints $\omega_\textup{exp}^1 <\omega_\textup{exp}^2$ is a valid strategy, but not necessarily the optimal one}.

The monotonicity is a direct consequence of the monogamy-of-entanglement property.
Indeed, the magic square game self-tests a pair of maximally entangled states \cite{WBMS16}, so it is expected that the value $\sup \frac{1}{9} \sum_{x,y} \bra{\phi} \Pi^{xy}\kettext{\phi}$ decreases from $8/9$ when $\oexp=8/9$ to $1/2$ when $\oexp=1$.
Consequently, Eve's best strategy if Alice and Bob are required to win with unit probability is to make a uniform guess, which brings the tripartite winning probability to $1/2$ ---as can be verified by the rightmost datapoint in Figure \ref{fig:MSG_cons_w_upperbound}.

\begin{figure}[h]
    \centering
     \includegraphics[width=0.6\textwidth]{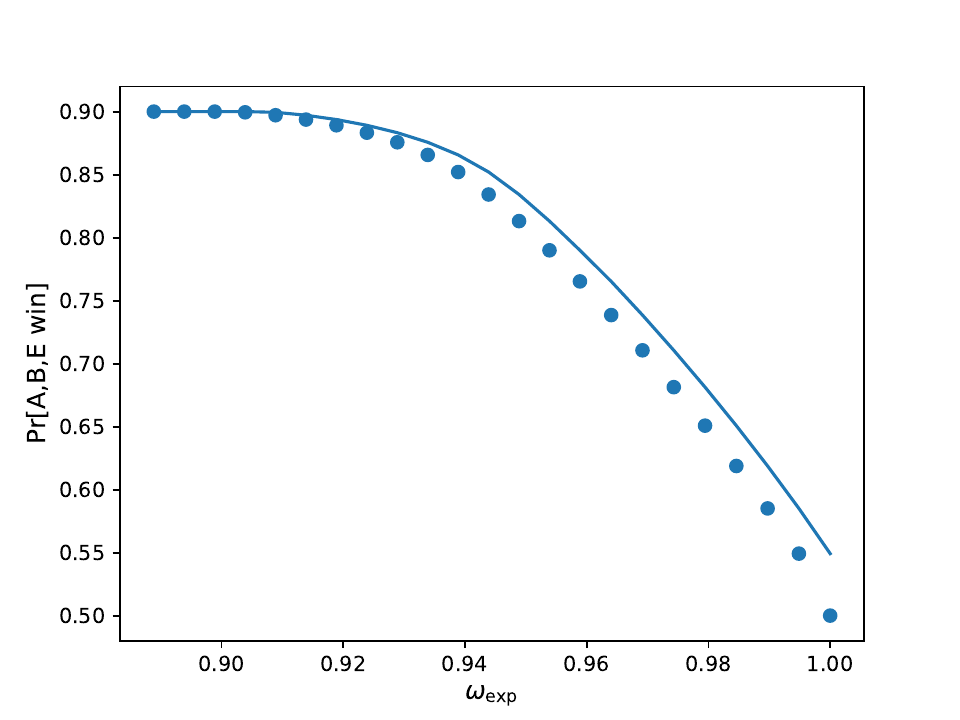}
     \caption{Upper bounds on the winning probability of Alice, Bob and Eve on non-local games constrained on the winning probability of Alice and Bob as given by the problem in Eq.~\eqref{msg_constrained}.
     The solid line corresponds to the upper bound defined in Eq.~\eqref{msg_prob_upper_bound} from the NPA points.
    See \cite{Cervero23} for Python script used to generate these figures.}
     \label{fig:MSG_cons_w_upperbound}
\end{figure}

The upper bound $h:[8/9,1]\rightarrow \mathbb{R}$ can be constructed from the $N$ points computed via the NPA hierarchy as follows.
Given a sequence of $N$ constraints $[x_k]_{k=0}^{N-1}$ on $\oexp$ with corresponding upper bounds $[y_k]_{k=0}^{N-1}$ to the optimization in Eq.~\eqref{msg_constrained} obtained via the NPA hierarchy, let $L_k$ be the straight line segment connecting points $(x_{k},y_{k})$ to $(x_{k+1},y_{k+1})$.
For any $x\in [x_k,x_{k+1}]$, $L_k$ is given by the following explicit form
\begin{align}
    f_k(x) = \frac{y_{k-1}-y_k}{x_{k-1}-x_k}(x-x_k) - y_k.
\end{align}
Hence, a continuous upper bound to the tripartite winning probability $h$ is defined as follows
\begin{align}\label{msg_prob_upper_bound}
    h(\oexp) := 
    \begin{cases}
        8/9 & \text{if } x_0 \leq \oexp \leq x_1 \\
        f_k(\oexp) & \text{if } x_{k+1} \leq \oexp \leq x_{k+2} \text{ for } k=0,...,N-3.
    \end{cases}
\end{align}
On a slight technicality, we remark that this upper bound is differentiable everywhere except at each point $x_k$. 
However, this is not an issue as one can simply choose the gradient of either the line segment preceding or succeeding it whilst guaranteeing that the resulting tangent is still an upper bound to the tripartite winning probability.

Intuitively, $h$ is constructed by placing the line segment $L_k$ connecting the pairs $(x_{k},y_{k})$ and $(x_{k+1},y_{k+1})$ above the pairs $(x_{k+1},y_{k+1})$ to $(x_{k+2},y_{k+2})$.
The resultant function is clearly an upper bound by the monotonicity and concavity of the optimums of Eq.~\eqref{msg_constrained} for constraints $\oexp\in[8/9,1]$ and further converges as $N\rightarrow\infty$.
The function $h$ is plotted in Figure \ref{fig:MSG_cons_w_upperbound}.

Using $h$ in the the bound for Eve's guessing probability in Eq.~\eqref{min_ent_bound} directly yields a bound to the min-entropy of Alice's key bit in one instance of the magic square game $H_{\min}(S_A|E)_\nu$ for $\nu\in\mathcal{S}_2$. 
An affine lower bound to this min-entropy is subsequently given by Eq.~\eqref{affine_constrained}.

\subsection{Bound on von-Neumann entropy}\label{sec_msg_vn}
In this section, we additionally employ the device independent bounds introduced in \cite{BFF23} to obtain a bound of the von Neumann entropy $H(S_A | E)_\nu$ of Alice's bit $S_A=A[Y]$ in one instance of the MSG over states $\nu$ belonging to a strategy of Alice and Bob $\mathcal{S}_2$ which achieves $\oexp\in[8/9,1]$.
This is done in order to obtain a tighter min-tradeoff function for improved key rates in our MSG-DIQKD protocol.

Firstly, note that it is not possible to use \cite[Lemma 2.3]{BFF23} directly for the magic square game.
This is because in the MSG, the key bit $S_A$ generated each round is not interchangeable with the register $A$ in the lemma, which is to be interpreted as the outcome of the non-local game.
Indeed, Alice's measurement output $A$ in the magic square non-local game does not correspond to Alice's key bit for that round. 
Nonetheless, it is possible to use \cite[Theorem 2.1]{BFF23} to overcome this subtlety for the MSG and prove a similar bound:

\begin{theorem}\thlabel{MSG_di_vn}
    Let $m\in\mathbb{N}$, and let $t_1,...,t_m$ and $w_1,...,w_m$ be the nodes and weights of an $m$-point Gauss-Radau quadrature on $[0,1]$ with endpoint $t_m=1$ (see \cite{BFF23} for relevant definitions). 
    Let $\nu_{Q_AQ_BE}$ be the initial quantum state shared between the devices of Alice, Bob and Eve in an instance of the MSG and let $\{P_a^x\}_a$ denote the measurement operators performed by Alice's device in the MSG on input $x\in\mathcal{X}$. 
    Further, for $i=1,...,m-1$ let $\alpha_i=\frac{3}{2}\max\{\frac{1}{t_i}, \frac{1}{t_1-1}\}$. 
    Then $H(S_A|XYE)_\nu$ is lower bounded by
    \begin{align}
         c_m + \sum_{i=1}^{m-1}\frac{w_i}{t_i\ln2}\sum_{s,x,y} \pi(x,y)\inf_{Z_{s,x,y}} &\Tr\Bigg[\nu_{Q_AE}\Bigg(\sum_{a: a[y]=s}P_a^x\otimes (Z_{s,x,y} + Z_{s,x,y}^\dagger + (1-t_i)Z^\dagger_{s,x,y}Z_{s,x,y})\Bigg)\Bigg]\\
        &\quad+ \Tr[\rho_{Q_AQ_E}t_i(\mathbb{I}_{S_A}\otimes Z_{s,x,y}Z^\dagger_{s,x,y})],
    \end{align}
    such that $\norm{Z_{s,x,y}}_\infty\leq \alpha_i$, where $c_m=\sum_{i=1}^{m-1}\frac{w_i}{t_i\ln 2}$, the infimum is over operators $Z_{s,x,y}$ in Eve's space and $\pi(x,y)$ is the probability that Alice and Bob select inputs $x,y$.
    Moreover, these bounds converge to $H(S_A|XYE)_\nu$ as $m\rightarrow\infty$.
\end{theorem}

The proof of this theorem is analogous to that of \cite[Lemma 2.3]{BFF23} albeit more involved by necessity of notation, but requires us to first establish some notation about the magic square game.

Suppose that $\nu_{Q_AQ_BE}\in S(Q_AQ_BE)$ is the state shared by Alice, Bob and Eve at the beginning of the game chosen such that Alice and Bob win $\msgt$ with probability $\oexp\in[8/9,1]$.
On input of $x\in\{0,1,2\}$ for Alice and $y\in\{0,1,2\}$ for Bob, they will perform measurements $\{P_a^x\}_a$ and $\{Q_b^y\}_b$ on their respective portions of $\nu$.
Including the input registers, the post-measurement state is $\nu_{ABXYE} = \sum_{a,b,x,y}\pi(x,y)\ketbratext{abxy}{abxy}\otimes \nu_{E}^{abxy}$, where $\nu_{E}^{abxy}=\Tr_{Q_AQ_B}[\nu_{Q_AQ_BE}(P_a^x\otimes Q_b^y\otimes\mathbb{I}_{E})]$ is Eve's post-measurement subnormalised state.
Now, to generate her key bit Alice applies the map 
\begin{align}
    &\mathcal{E}\in\CPTP(AY,SY), \\
    &\mathcal{E}:\ketbra{a_0y_0}{a_1y_1} \mapsto \ketbra{a_0[y_0]y_0}{a_1[y_1]y_1},
\end{align}
for any pairs $a_0,a_1\in\{000,011,101,110\}$ and $y_0,y_1\in\{0,1,2\}$.
We include the register $Y$ in the codomain of $\mathcal{E}$ to emphasize that Eve knows it. 
This results in the state
\begin{align}
    \nu_{S_ABXYE}
    &= \sum_{a,b,x,y}\pi(x,y)\ketbra{a[y]bxy}\otimes \nu_{E}^{abxy}\\
    &= \sum_{s,b,x,y}\pi(x,y)\ketbra{sbxy}\otimes \nu_{E}^{sbxy},
\end{align}
where we have introduced the subnormalized state $\nu_{E}^{sbxy}:=\sum_{a: a[y]=s}\nu_{E}^{abxy}$.
Note that 
\begin{align}
    \sum_{a: a[y]=s}\nu_{E}^{abxy} = \Tr_{Q_AQ_B}\Bigg[\nu_{Q_AQ_BE}\Bigg(\sum_{{a: a[y]=s}} P_a^x \otimes \Qby \otimes\mathbb{I}_{E}\Bigg)\Bigg],
\end{align}
so $\nu_{E}^{sbxy}$ is analogously obtained by using $P_s^{xy}=\sum_{{a: a[y]=s}} P_a^x$ instead.

We continue by exploring how this rewriting of $P_s^{xy}$ affects the expression for Alice's and Bob's winning probability constraint.
Recall that $\nu_{Q_AQ_BE}$ is chosen so that
\begin{align}
    \omega_\textup{exp} = \sum_{x,y}\frac{1}{9}\Tr[\nu_{Q_AQ_BE}\Pi^{xy}_{AB}],
\end{align}
where
\begin{align}
    \Pi^{xy}_{AB} 
    &= \frac{1}{9}\sum_{s\in\{0,1\}} \Bigg[\sum_{a: a[y] = s} P^x_a \otimes \sum_{b: b[x] = s} Q^y_b \otimes \mathbb{I}_{E}\Bigg] \\
    &= \frac{1}{9}\sum_{s\in\{0,1\}} \Bigg[P^{xy}_s \otimes \sum_{b: b[x] = s} Q^y_b \otimes \mathbb{I}_{E}\Bigg].
\end{align}
Replacing Alice's projections with $\{P_s^{xy}\}$ and accounting for inputs $x,y$ in the proof of \cite[Lemma 2.3]{BFF23} yields precisely the statement in \thref{MSG_di_vn}.

In practice, the expression in \thref{MSG_di_vn} bounding the entropy $H(S_A|XYE)$ in the MSG contains too many operators to feasibly optimize using the NPA hierarchy.
However, note $H(S_A|XYE)_\nu = \sum_{x,y}\pi(x,y)H(S_A|X=x,Y=y;E)_{\nu^{xy}}$ where $\nu^{xy}_{S_AXYE}= \sum_{s}\pi(x,y)[s]\otimes \nu_{E}^{sxy}$ is the final state with fixed inputs $(x,y)$ and $\pi(x,y)$ is uniform.
This implies $H(S_A|XYE)_\nu$ is lower bounded by $H(S_A|X\in\mathcal{X}',Y\in\mathcal{Y}';E)_\nu$ for any subsets $\mathcal{X}'\subset\mathcal{X}, \mathcal{Y}'\subset\mathcal{Y}$ which may be chosen to make the NPA optimization feasible.
In particular, $H(S_A|X\in\mathcal{X}',Y\in\mathcal{Y}';E)_\nu$ is lower bounded by replacing the inner summation in the expression of \thref{MSG_di_vn} with
\begin{align}\label{MSG_di_vn_smol}
    \sum_{\substack{s\\ x\in\mathcal{X}' \\ y\in\mathcal{Y}'}} \frac{1}{|\mathcal{X}'\times\mathcal{Y}'|}\inf_{Z_{s,x,y}} &\Tr\Bigg[\nu_{Q_AQ_E}\Bigg(\sum_{a: a[y]=s}P_a^x\otimes (Z_{s,x,y} + Z_{s,x,y}^\dagger + (1-t_i)Z_{s,x,y}^\dagger Z_{s,x,y})\Bigg)\Bigg]\\
    &\quad+ \Tr[\nu_{Q_AQ_E}t_i(\mathbb{I}_{S_A}\otimes Z_{s,x,y}Z_{s,x,y}^\dagger)].
\end{align}

Results from these optimizations choosing $\mathcal{X}'=\{0\}$ and $\mathcal{Y}'=\{0\}$, and $\mathcal{X}'=\{0,1,2\}$ and $\mathcal{Y}'=\{0\}$ are presented in Figures \ref{fig:MSG_asymptotics} (alongside the corresponding Devetak-Winter rates) and \ref{fig:MSG_ent_comparison}. 
In the latter, we also compare the min-entropy bound derived in the previous section: $-\log(1-\oexp+h(\oexp))$ for $\oexp\in[\frac{8}{9},1]$ and $h:[\frac{8}{9},1]\rightarrow\mathbb{R}$ the continuous upper bound defined in Eq.~\eqref{msg_prob_upper_bound}.

It is apparent in Figure \ref{fig:MSG_ent_comparison} that as a consequence of having to restrict the set of inputs in Eq.~\eqref{MSG_di_vn_smol} to make the numerical optimization feasible, the resultant von Neumann entropy bounds are not optimal ---in particular, for $\oexp\lesssim 0.97$ the min-entropy bound gives a higher value. 
When reporting the finite and asymptotic statistics for the MSG-DIQKD in Section \ref{sec_msg_keyrates} we use whichever entropy bound yields better key rates. 

Further, we remark that when Alice and Bob require a winning probability of $1$, $H_{\min}(S_A|XYE)=H(S_A|XYE)=1$ as opposed to the roughly $H_{\min}(S_A|XYE)\approx 0.85$ in Figure \ref{fig:MSG_ent_comparison}.
We attribute this error to the amount of points in Figure \ref{fig:MSG_cons_w_upperbound} from which $h$ is constructed ---indeed the function $h$ converges to the tripartite winning probability bound as the number of points increases. 

\begin{figure}[h]
    \centering
    \includegraphics[width=0.6\textwidth]{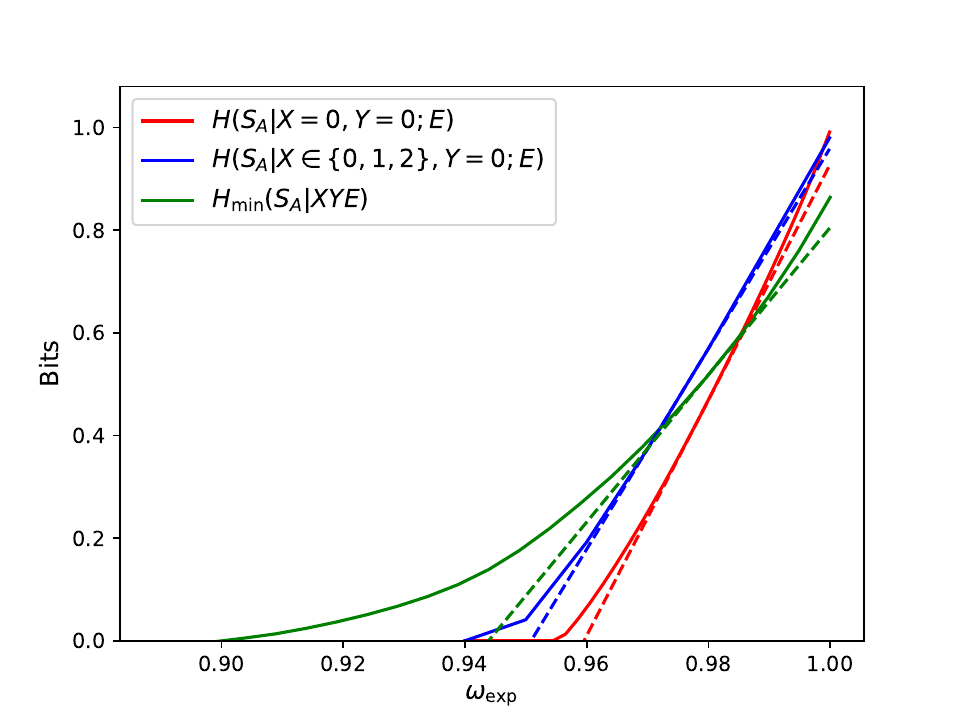}
    \caption{Comparison of bounds to the entropy of Alice's raw key conditioned on Eve's information. 
    Solid green line corresponds to the negative logarithm of Eq.~\eqref{min_ent_bound}, whilst the solid red line and solid blue line correspond to the bounds on $H(S|X\in\mathcal{X}',Y=0;Q_E)$ for $\mathcal{X}'=\{0\}, \{0,1,2\}$ respectively, found via \thref{MSG_di_vn} and Eq.~\eqref{MSG_di_vn_smol}.
    The corresponding dashed lines are affine lower bounds to each entropy at the point $0.98$, computed using Eq.~\eqref{affine_constrained} for the min entropy (green) and Eq.~\eqref{msg_vn_lower_bound} for the von Neumann entropies (red, blue).}
    \label{fig:MSG_ent_comparison}
\end{figure}

In the remainder of this section, we use the points computed using the optimization from \thref{MSG_di_vn} to construct an affine lower bound to the von Neumann entropy $H(S_A|XYE)$ in one instance of the magic square game. 
Analogous to the previous section, given a sequence of constraints $[x_k]_{k=0}^{N-1}\subset[8/9,1]$ with corresponding optimums $[y_k]_{k=0}^{N-1}\subset[0,1]$ of the optimization problem in \thref{MSG_di_vn}, then a lower bound to $H(S_A|XYE)$ is given by the continuous function $h_\textup{vN}:[8/9,1]\rightarrow\mathbb{R}$ defined by

\begin{align}\label{msg_vn_lower_bound}
    h_\textup{vN}(\oexp) := 
    \begin{cases}
        0 & \text{if } x_0 \leq \oexp \leq x_1 \\
        \frac{y_{k-1}-y_k}{x_{k-1}-x_k}(\oexp-x_k) - y_k & \text{if } x_{k+1} \leq \oexp \leq x_{k+2} \text{ for } k=0,...,N-3.
    \end{cases}
\end{align}

Since the von Neumann entropy is convex and monotically increasing in the constraint $\oexp\in[8/9,1]$, the function $h_\textup{vN}$ obtained via the lower bounds of \thref{MSG_di_vn} is indeed a lower bound which converges as the number of points $N$ goes to infinity.
Hence, any tangent $g_\textup{vN}:[8/9,1]\rightarrow\mathbb{R}$ to $h_\textup{vN}$ is a suitable affine lower bound to the von Neumann entropy which may replace the affine function $g$ in the finite key length of Eq.~\eqref{l_key}.
This is because $g_\textup{vN}$ is a lower bound to 
\begin{align}
    \min_{\omega_{REF}}\{H(S_A|E)_{\mathcal{N}_i(\omega_{REF})}: \Tr_{S_AREF}[\mathcal{N}_i^\textup{test}(\omega_{REF})] = (1-\oexp)\ketbra{0}_{C_i}+\oexp\ketbra{1}_{C_i}\}
\end{align}
for $\mathcal{N}_i$ as in Eq.~\eqref{infrequent_sampling_channel}, following the same logic as in Eqs.~\eqref{pseudo_mintradeoff1}-\eqref{pseudo_mintradeoff2} and the remarks thereafter.

\section{Conclusions}
In this work we introduce a framework for device independent {quantum key distribution} {leveraging on} any two party non-local game {with a winning condition that} requires Alice and Bob to produce a matching bit. 
The security of this protocol is derived from a monogamy-of-entanglement property which states that the winning probability of the non-local game played solely by Alice and Bob is higher than the winning probability of the same game extended to three parties, in which the third party receives all inputs and is required to guess the bit produced by Alice and Bob. 
Using the gap in the winning probabilities in the two and three party games we provide a very simple affine bound to the min-entropy of each round which is sufficient to bound an eavesdropper's knowledge of the sifted keys in the quantum key distribution protocol proposed in algorithm \ref{prot_qkd_actual}.

We exemplify our general framework with the magic square game, for which we provide a numerical bound on the maximal tripartite winning probability of about $\frac{8.00077}{9}$ which is close to the maximal classical bipartite winning probability of $\frac{8}{9}$.
We use the NPA hierarchy to compute bounds on the min-entropy of Alice's bit using the tripartite winning probabilities (Section \ref{sec_msg_min_ent}) of the magic square game.
Further, we also use the the techniques from \cite{BFF23} to compute bounds on the respective von Neumann entropy (Section \ref{sec_msg_vn}).
Using a combination of the min- and von Neumann entropy bounds we provide a simple numerical construction for min-tradeoff functions suitable for the EAT/GEAT.
We use these bounds to numerically optimize the finite and asymptotic key rates of our protocol using the magic square game as primitive, obtaining robustness up to roughly $2.88\%$ in depolarizing noise ---the highest tolerance reported for MSG based DIQKD.

Lastly, we remark that our proof of secrecy is mostly independent of our monogamy-of-entanglement framework and the subsequent choice of min-tradeoff function $g$. 
In fact, as can be inferred from Sections \ref{sec_min_tradeoff} and \ref{sec_final_bounds}, any affine function $g:\mathcal{P}(\{0,1\})\rightarrow\mathbb{R}$ satisfying Eq.~\eqref{affine_g} derived for a specific primitive may be substituted into \thref{th_secrecy} to yield a secrecy guarantee.
In particular, the secrecy proof remains valid when considering games where the set of inputs in testing rounds differs from those in generation rounds, provided $g$ is chosen appropriately.
For example, \cite[Eq.~(28)]{TSB+22} may be directly used in Eq.~\eqref{l_key} to bound $\lkey$ when considering a DIQKD protocol based on the CHSH game where Alice always produces inputs in $\{0,1\}$ whereas Bob produces $\{0,1\}$ in generation rounds but $\{2,3\}$ in testing rounds.

\section*{Acknowledgements}
We would like to thank Armando Bellante for early discussions on this work, as well as Mario Berta, Christian Boghiu, Peter Brown and Javier Rivera-Dean for tips on the usage of the ncpol2sdpa python library.
We are additionally grateful to Ernest Y.-Z. Tan for valuable discussions and tips on entropy accumulation based proofs, and also Roberto Rubboli for useful comments on earlier versions of this manuscript.
Lastly we thank Acharya Tejas and Ian George for pointing out a mistake in the original manuscript's calculation of the magic square game's QBER.
E.C.M.\ and M.T.\ are funded by the National Research Foundation, Singapore and A*STAR under its CQT Bridging Grant. This research is also supported by the National Research Foundation, Singapore and A*STAR under its Quantum Engineering Programme (NRF2021-QEP2-01-P06).

\bibliographystyle{quantum}
\bibliography{refs}

\begin{thebibliography}{10}

\bibitem{BHK05}
Jonathan Barrett, Lucien Hardy, and Adrian Kent.
\newblock ``No signaling and quantum key distribution''.
\newblock \href{https://dx.doi.org/10.1103/PhysRevLett.95.010503}{Phys. Rev. Lett. {\bf 95}, 010503}~(2005).

\bibitem{AMP06}
Antonio Acín, Serge Massar, and Stefano Pironio.
\newblock ``Efficient quantum key distribution secure against no-signalling eavesdroppers''.
\newblock \href{https://dx.doi.org/10.1088/1367-2630/8/8/126}{New Journal of Physics {\bf 8}, 126--126}~(2006).

\bibitem{Werner89}
Reinhard~F. Werner.
\newblock ``An application of bell's inequalities to a quantum state extension problem''.
\newblock \href{https://dx.doi.org/10.1007/BF00399761}{Letters in Mathematical Physics {\bf 17}, 359--363}~(1989).

\bibitem{DPS04}
Andrew~C. Doherty, Pablo~A. Parrilo, and Federico~M. Spedalieri.
\newblock ``Complete family of separability criteria''.
\newblock \href{https://dx.doi.org/10.1103/PhysRevA.69.022308}{Phys. Rev. A {\bf 69}, 022308}~(2004).

\bibitem{Renner05}
Renato Renner.
\newblock ``Security of quantum key distribution''.
\newblock PhD thesis.
\newblock ETH Zurich.
\newblock ~(2006).
\newblock  url:~\href{https://arxiv.org/abs/quant-ph/0512258}{arxiv.org/abs/quant-ph/0512258}.

\bibitem{TL17}
Marco Tomamichel and Anthony Leverrier.
\newblock ``A largely self-contained and complete security proof for quantum key distribution''.
\newblock \href{https://dx.doi.org/10.22331/q-2017-07-14-14}{{Quantum} {\bf 1}, 14}~(2017).

\bibitem{TCR09}
Marco Tomamichel, Roger Colbeck, and Renato Renner.
\newblock ``A fully quantum asymptotic equipartition property''.
\newblock \href{https://dx.doi.org/10.1109/TIT.2009.2032797}{IEEE Transactions on Information Theory {\bf 55}, 5840--5847}~(2009).

\bibitem{DFR20}
Frédéric Dupuis, Omar Fawzi, and Renato Renner.
\newblock ``Entropy accumulation''.
\newblock \href{https://dx.doi.org/10.1007/s00220-020-03839-5}{Communications in Mathematical Physics {\bf 379}, 867--913}~(2020).

\bibitem{DF19}
Frédéric Dupuis and Omar Fawzi.
\newblock ``Entropy accumulation with improved second-order term''.
\newblock \href{https://dx.doi.org/10.1109/TIT.2019.2929564}{IEEE Transactions on Information Theory {\bf 65}, 7596--7612}~(2019).

\bibitem{MFSR22}
Tony Metger, Omar Fawzi, David Sutter, and Renato Renner.
\newblock ``Generalised entropy accumulation''.
\newblock \href{https://dx.doi.org/10.1109/FOCS54457.2022.00085}{2022 IEEE 63rd Annual Symposium on Foundations of Computer Science (FOCS)Pages 844--850}~(2022).

\bibitem{T21}
Ernest Y.~Z. Tan.
\newblock ``Prospects for device-independent quantum key distribution''.
\newblock PhD thesis.
\newblock ETH Zurich.
\newblock ~(2021).
\newblock  url:~\href{https://arxiv.org/abs/2111.11769}{arxiv.org/abs/2111.11769}.

\bibitem{VR08}
Valerio Scarani and Renato Renner.
\newblock ``Security bounds for quantum cryptography with finite resources''.
\newblock In Theory of Quantum Computation, Communication, and Cryptography.
\newblock \href{https://dx.doi.org/https://doi.org/10.1007/978-3-540-89304-2_8}{Pages 83--95}.
\newblock Berlin, Heidelberg~(2008). Springer Berlin Heidelberg.

\bibitem{SPV10}
Lana Sheridan, Thinh~Phuc Le, and Valerio Scarani.
\newblock ``Finite-key security against coherent attacks in quantum key distribution''.
\newblock \href{https://dx.doi.org/10.1088/1367-2630/12/12/123019}{New Journal of Physics {\bf 12}, 123019}~(2010).

\bibitem{TLGR12}
Marco Tomamichel, Charles C.-W. Lim, Nicolas Gisin, and Renato Renner.
\newblock ``Tight finite-key analysis for quantum cryptography''.
\newblock \href{https://dx.doi.org/10.1038/ncomms1631}{Nature Communications~{\bf 3}}~(2012).

\bibitem{TFKW13}
Marco Tomamichel, Serge Fehr, Jędrzej Kaniewski, and Stephanie Wehner.
\newblock ``A monogamy-of-entanglement game with applications to device-independent quantum cryptography''.
\newblock \href{https://dx.doi.org/10.1088/1367-2630/15/10/103002}{New Journal of Physics {\bf 15}, 103002}~(2013).

\bibitem{VV14}
Umesh Vazirani and Thomas Vidick.
\newblock ``Fully device-independent quantum key distribution''.
\newblock \href{https://dx.doi.org/10.1103/physrevlett.113.140501}{Physical Review Letters~{\bf 113}}~(2014).

\bibitem{AFDF+18}
Rotem Arnon-Friedman, Fr{\'e}d{\'e}ric Dupuis, Omar Fawzi, Renato Renner, and Thomas Vidick.
\newblock ``Practical device-independent quantum cryptography via entropy accumulation''.
\newblock \href{https://dx.doi.org/10.1038/s41467-017-02307-4}{Nature Communications {\bf 9}, 459}~(2018).

\bibitem{AFRV19}
Rotem Arnon-Friedman, Renato Renner, and Thomas Vidick.
\newblock ``Simple and tight device-independent security proofs''.
\newblock \href{https://dx.doi.org/10.1137/18M1174726}{SIAM Journal on Computing {\bf 48}, 181--225}~(2019).

\bibitem{TSG+21}
Ernest Y.-Z. Tan, René Schwonnek, Koon~Tong Goh, Ignatius~William Primaatmaja, and Charles C.-W. Lim.
\newblock ``Computing secure key rates for quantum cryptography with untrusted devices''.
\newblock \href{https://dx.doi.org/10.1038/s41534-021-00494-z}{npj Quantum Information~{\bf 7}}~(2021).

\bibitem{STP+21}
René Schwonnek, Koon~Tong Goh, Ignatius~W. Primaatmaja, Ernest Y.-Z. Tan, Ramona Wolf, Valerio Scarani, and Charles C.-W. Lim.
\newblock ``Device-independent quantum key distribution with random key basis''.
\newblock \href{https://dx.doi.org/10.1038/s41467-021-23147-3}{Nature Communications~{\bf 12}}~(2021).

\bibitem{TSB+22}
Ernest Y.-Z. Tan, Pavel Sekatski, Jean-Daniel Bancal, René Schwonnek, Renato Renner, Nicolas Sangouard, and Charles C.-W. Lim.
\newblock ``Improved {DIQKD} protocols with finite-size analysis''.
\newblock \href{https://dx.doi.org/10.22331/q-2022-12-22-880}{Quantum {\bf 6}, 880}~(2022).

\bibitem{NDN+22}
D.~P. Nadlinger, P.~Drmota, B.~C. Nichol, G.~Araneda, D.~Main, R.~Srinivas, D.~M. Lucas, C.~J. Ballance, K.~Ivanov, E.~Y.-Z. Tan, P.~Sekatski, R.~L. Urbanke, R.~Renner, N.~Sangouard, and J.-D. Bancal.
\newblock ``Experimental quantum key distribution certified by bell{\textquotesingle}s theorem''.
\newblock \href{https://dx.doi.org/10.1038/s41586-022-04941-5}{Nature {\bf 607}, 682--686}~(2022).

\bibitem{MR22}
Tony Metger and Renato Renner.
\newblock ``Security of quantum key distribution from generalised entropy accumulation''.
\newblock \href{https://dx.doi.org/10.1038/s41467-023-40920-8}{Nature Communications {\bf 14}, 5272}~(2023).

\bibitem{PAB+20}
Stefano Pirandola, Ulrik~L. Andersen, Leonardo Banchi, Mario Berta, Darius Bunandar, Roger Colbeck, Dirk Englund, Tobias Gehring, Cosmo Lupo, Carlo Ottaviani, Jason~L. Pereira, Mohsen Razavi, Jesni~Shamsul Shaari, Marco Tomamichel, Vladyslav~C. Usenko, Giuseppe Vallone, Paolo Villoresi, and Petros Wallden.
\newblock ``Advances in quantum cryptography''.
\newblock \href{https://dx.doi.org/10.1364/aop.361502}{Advances in Optics and Photonics {\bf 12}, 1012}~(2020).

\bibitem{ZLetal23}
Víctor Zapatero, Tim van Leent, Rotem Arnon-Friedman, Wen-Zhao Liu, Qiang Zhang, Harald Weinfurter, and Marcos Curty.
\newblock ``Advances in device-independent quantum key distribution''.
\newblock \href{https://dx.doi.org/10.1038/s41534-023-00684-x}{npj Quantum Information~{\bf 9}}~(2023).

\bibitem{PGTKGL23}
Ignatius~W. Primaatmaja, Koon~Tong Goh, Ernest Y.-Z. Tan, John T.-F. Khoo, Shouvik Ghorai, and Charles C.-W. Lim.
\newblock ``Security of device-independent quantum key distribution protocols: a review''.
\newblock \href{https://dx.doi.org/10.22331/q-2023-03-02-932}{Quantum {\bf 7}, 932}~(2023).

\bibitem{Zetal22}
Tim Zhang, Weiand van~Leent, Kai Redeker, Robert Garthoff, René Schwonnek, Florian Fertig, Sebastian Eppelt, Wenjamin Rosenfeld, Valerio Scarani, Charles C.-W. Lim, and Harald Weinfurter.
\newblock ``A device-independent quantum key distribution system for distant users''.
\newblock \href{https://dx.doi.org/10.1038/s41586-022-04891-y}{Nature {\bf 607}, 687--691}~(2022).

\bibitem{V17}
Thomas Vidick.
\newblock ``Parallel diqkd from parallel repetition''~(2017).
\newblock  \href{http://arxiv.org/abs/1703.08508}{arXiv:1703.08508}.

\bibitem{JMS20}
Rahul Jain, Carl~A. Miller, and Yaoyun Shi.
\newblock ``Parallel device-independent quantum key distribution''.
\newblock \href{https://dx.doi.org/10.1109/tit.2020.2986740}{{IEEE} Transactions on Information Theory {\bf 66}, 5567--5584}~(2020).

\bibitem{JK21}
Rahul Jain and Srijita Kundu.
\newblock ``A direct product theorem for quantum communication complexity with applications to device-independent cryptography''~(2023).
\newblock  \href{http://arxiv.org/abs/2106.04299}{arXiv:2106.04299}.

\bibitem{WBMS16}
Xingyao Wu, Jean-Daniel Bancal, Matthew McKague, and Valerio Scarani.
\newblock ``Device-independent parallel self-testing of two singlets''.
\newblock \href{https://dx.doi.org/10.1103/physreva.93.062121}{Physical Review A~{\bf 93}}~(2016).

\bibitem{CV22}
Eric Culf and Thomas Vidick.
\newblock ``A monogamy-of-entanglement game for subspace coset states''.
\newblock \href{https://dx.doi.org/10.22331/q-2022-09-01-791}{Quantum {\bf 6}, 791}~(2022).

\bibitem{CVV22}
Eric Culf, Thomas Vidick, and Victor~V. Albert.
\newblock ``Group coset monogamy games and an application to device-independent continuous-variable qkd''~(2022).
\newblock  \href{http://arxiv.org/abs/2212.0393}{arXiv:2212.03935v15}.

\bibitem{PK18}
Damián Pitalúa-García and Iordanis Kerenidis.
\newblock ``Practical and unconditionally secure spacetime-constrained oblivious transfer''.
\newblock \href{https://dx.doi.org/10.1103/physreva.98.032327}{Physical Review A{\bf 98}}~(2018).

\bibitem{CLLZ22}
Andrea Coladangelo, Jiahui Liu, Qipeng Liu, and Mark Zhandry.
\newblock ``Hidden cosets and applications to unclonable cryptography''.
\newblock In Advances in Cryptology -- CRYPTO 2021.
\newblock \href{https://dx.doi.org/https://doi.org/10.1007/978-3-030-84242-0_20}{Pages 556--584}.
\newblock Springer International Publishing~(2021).

\bibitem{KT22}
Srijita Kundu and Ernest Y.~Z. Tan.
\newblock ``Device-independent uncloneable encryption''~(2023).
\newblock  \href{http://arxiv.org/abs/2210.01058}{arXiv:2210.01058}.

\bibitem{PAB+09}
Stefano Pironio, Antonio Acín, Nicolas Brunner, Nicolas Gisin, Serge Massar, and Valerio Scarani.
\newblock ``Device-independent quantum key distribution secure against collective attacks''.
\newblock \href{https://dx.doi.org/10.1088/1367-2630/11/4/045021}{New Journal of Physics {\bf 11}, 045021}~(2009).

\bibitem{SBV+21}
Pavel Sekatski, Jean-Daniel Bancal, Xavier Valcarce, Ernest Y.-Z. Tan, Renato Renner, and Nicolas Sangouard.
\newblock ``Device-independent quantum key distribution from generalized {CHSH} inequalities''.
\newblock \href{https://dx.doi.org/10.22331/q-2021-04-26-444}{{Quantum} {\bf 5}, 444}~(2021).

\bibitem{MS16}
Carl~A. Miller and Yaoyun Shi.
\newblock ``Robust protocols for securely expanding randomness and distributing keys using untrusted quantum devices''.
\newblock \href{https://dx.doi.org/10.1145/2885493}{J. ACM~{\bf 63}}~(2016).

\bibitem{Mermin90}
N.~David Mermin.
\newblock ``Simple unified form for the major no-hidden-variables theorems''.
\newblock \href{https://dx.doi.org/10.1103/PhysRevLett.65.3373}{Phys. Rev. Lett. {\bf 65}, 3373--3376}~(1990).

\bibitem{Peres90}
Asher Peres.
\newblock ``Incompatible results of quantum measurements''.
\newblock \href{https://dx.doi.org/https://doi.org/10.1016/0375-9601(90)90172-K}{Physics Letters A {\bf 151}, 107--108}~(1990).

\bibitem{NPA08}
Miguel Navascués, Stefano Pironio, and Antonio Acín.
\newblock ``A convergent hierarchy of semidefinite programs characterizing the set of quantum correlations''.
\newblock \href{https://dx.doi.org/10.1088/1367-2630/10/7/073013}{New Journal of Physics {\bf 10}, 073013}~(2008).

\bibitem{PNA10}
Stefano Pironio, Miguel Navascués, and Antonio Acín.
\newblock ``Convergent relaxations of polynomial optimization problems with noncommuting variables''.
\newblock \href{https://dx.doi.org/10.1137/090760155}{SIAM Journal on Optimization {\bf 20}, 2157--2180}~(2010).

\bibitem{Wittek15}
Peter Wittek.
\newblock ``Algorithm 950''.
\newblock \href{https://dx.doi.org/10.1145/2699464}{{ACM} Transactions on Mathematical Software {\bf 41}, 1--12}~(2015).

\bibitem{BFF23}
Peter Brown, Hamza Fawzi, and Omar Fawzi.
\newblock ``Device-independent lower bounds on the conditional von neumann entropy''.
\newblock \href{https://dx.doi.org/10.22331/q-2024-08-27-1445}{Quantum {\bf 8}, 1445}~(2024).

\bibitem{ZMZ+23}
Yi-Zheng Zhen, Yingqiu Mao, Yu-Zhe Zhang, Feihu Xu, and Barry~C. Sanders.
\newblock ``Device-independent quantum key distribution based on the mermin-peres magic square game''.
\newblock \href{https://dx.doi.org/10.1103/PhysRevLett.131.080801}{Phys. Rev. Lett. {\bf 131}, 080801}~(2023).

\bibitem{PV16}
Carlos Palazuelos and Thomas Vidick.
\newblock ``Survey on nonlocal games and operator space theory''.
\newblock \href{https://dx.doi.org/10.1063/1.4938052}{Journal of Mathematical Physics {\bf 57}, 015220}~(2016).

\bibitem{MRC+14}
Lluís Masanes, Renato Renner, Matthias Christandl, Andreas Winter, and Jonathan Barrett.
\newblock ``Full security of quantum key distribution from no-signaling constraints''.
\newblock \href{https://dx.doi.org/10.1109/TIT.2014.2329417}{IEEE Transactions on Information Theory {\bf 60}, 4973--4986}~(2014).

\bibitem{RW04}
Renato Renner and Stefan Wolf.
\newblock ``The exact price for unconditionally secure asymmetric cryptography''.
\newblock In Advances in Cryptology - EUROCRYPT 2004.
\newblock \href{https://dx.doi.org/https://doi.org/10.1007/978-3-540-24676-3_7}{Pages 109--125}.
\newblock Berlin, Heidelberg~(2004). Springer Berlin Heidelberg.

\bibitem{DW09}
{Dodis, Yevgeniy and Wichs, Daniel}.
\newblock ``Non-malleable extractors and symmetric key cryptography from weak secrets''.
\newblock In {Proceedings forty-first annual acm symposium on theory computing}.
\newblock \href{https://dx.doi.org/10.1145/1536414.1536496}{Page 601–610}.
\newblock New York, NY, USA~(2009). Association for Computing Machinery.

\bibitem{PR14}
Christopher Portmann and Renato Renner.
\newblock ``Cryptographic security of quantum key distribution''~(2014).
\newblock  \href{http://arxiv.org/abs/1409.3525}{arXiv:1409.3525}.

\bibitem{RR12}
Joseph~M. Renes and Renato Renner.
\newblock ``One-shot classical data compression with quantum side information and the distillation of common randomness or secret keys''.
\newblock \href{https://dx.doi.org/10.1109/TIT.2011.2177589}{IEEE Transactions on Information Theory {\bf 58}, 1985--1991}~(2012).

\bibitem{TMPE17}
Marco Tomamichel, Jesus Martinez-Mateo, Christoph Pacher, and David Elkouss.
\newblock ``Fundamental finite key limits for one-way information reconciliation in quantum key distribution''.
\newblock \href{https://dx.doi.org/10.1007/s11128-017-1709-5}{Quantum Information Processing~{\bf 16}}~(2017).

\bibitem{DW05}
Igor Devetak and Andreas Winter.
\newblock ``Distillation of secret key and entanglement from quantum states''.
\newblock \href{https://dx.doi.org/10.1098/rspa.2004.1372}{Proceedings of the Royal Society A: Mathematical, Physical and Engineering Sciences {\bf 461}, 207--235}~(2005).

\bibitem{Cervero23}
Enrique Cervero-Martín.
\newblock ``Python script for obtaining tripartite winning probabilities, and finite and asymptotic keyrates''.
\newblock \url{https://github.com/EnriqueCerv/NLG_DI_QKD.git}~(2023).

\bibitem{Toner08}
Ben Toner.
\newblock ``Monogamy of non-local quantum correlations''.
\newblock \href{https://dx.doi.org/10.1098/rspa.2008.0149}{Proceedings of the Royal Society A: Mathematical, Physical and Engineering Sciences {\bf 465}, 59--69}~(2008).

\bibitem{Tomamichel16}
Marco Tomamichel.
\newblock ``Quantum information processing with finite resources''.
\newblock \href{https://dx.doi.org/10.1007/978-3-319-21891-5}{Springer International Publishing}. ~(2016).

\bibitem{ME05}
Michael Mitzenmacher and Eli Upfal.
\newblock ``Probability and computing: Randomized algorithms and probabilistic analysis''.
\newblock Cambridge University Press. USA~(2005).

\end{thebibliography}

\newpage
\appendix
\section{Simplifying $\omega_\textup{Q}(\msg)$}\label{app_msg_npa}
In this appendix we provide a simplification to the objective function in the optimization problems for the tripartite magic square game in Eq.~\eqref{qvalue_pure} which we used to compute the numerical results in \thref{th_msg} and Figure \ref{fig:MSG_cons_w_upperbound}.

Recall the optimization problem
\begin{align}\label{app_qvalue_pure}
    \omega_\textup{Q}(\msg) = \sup_{\mathcal{S}}\frac{1}{9} \sum_{x,y} \bra{\phi} \Pi^{xy}\ket{\phi},
\end{align}
where
\begin{align}\label{pi_xy_app}
    \Pi^{xy} := \sum_{c\in\{0,1\}} \bigg[ F_c^{xy} \otimes \sum_{a\in \mathcal{A}: a[y] = c} P^x_a \otimes \sum_{b\in \mathcal{B}: b[x] = c} Q^y_b \bigg].
\end{align}
We begin by providing a simplified form of Eq.~\eqref{pi_xy_app}.
Let 
\begin{align}
    P^x_{000}&=\mathbb{I}_{Q_A}-P^x_{011}-P^x_{101}-P^x_{110},\\
    Q^y_{111}&=\mathbb{I}_{Q_B}-Q^y_{001}-Q^y_{010}-Q^y_{100},\\
    F^{xy}_1&=\mathbb{I}_{Q_C}-F^{xy}_0
\end{align}
for all $x,y\in\{0,1,2\}$.
A straightforward calculation then shows that Eq.~\eqref{pi_xy_app} reduces to
\begin{align}\label{pi_xy_simplified}
    \Pxy = \sum_{a:a[y]=1}\sum_{b:b[x]=0}\Big[\frac{1}{2} F_0^{xy}\otimes\mathbb{I}_{Q_A}\otimes\Qby + \frac{1}{2}(1-F_0^{xy})\otimes\Pax\otimes\mathbb{I}_{Q_B}-\mathbb{I}_{Q_C}\otimes\Pax\otimes\Qby\Big].
\end{align}

Now, let
\begin{align}
    \Pi^{xy}_{AB} &:= \frac{1}{9}\sum_{c\in\{0,1\}} \mathbb{I}_{Q_C}\otimes \sum_{a\in \mathcal{A}: a[y] = c} P^x_a \otimes \sum_{b\in \mathcal{B}: b[x] = c} Q^y_b,\\
    \Pi^{xy}_{AC} &:= \frac{1}{9}\sum_{c\in\{0,1\}} \Fcxy \otimes \sum_{a\in \mathcal{A}: a[y] = c} P^x_a  \otimes \mathbb{I}_{Q_B},\\
    \Pi^{xy}_{BC} &:= \frac{1}{9}\sum_{c\in\{0,1\}} \Fcxy\otimes \mathbb{I}_{Q_A} \otimes\sum_{b\in \mathcal{B}: b[x] = c} Q^y_b.
\end{align}

Without changing the optimum we add the following constraints to the optimization problem in Eq.~\eqref{app_qvalue_pure}:
\begin{align}\label{constraint_one}
    &\bra{\phi}\sum_{xy}\Pi^{xy}_{AB}\ket{\phi}\geq \frac{8}{9},\\
    &\bra{\phi}\sum_{xy}\Pi^{xy}_{AC}\ket{\phi}\geq \frac{8}{9}, \label{constraint_two} \\ 
    &\bra{\phi}\sum_{xy}\Pi^{xy}_{BC}\ket{\phi}\geq \frac{8}{9}.\label{constraint_three}
\end{align}
Indeed, since the classical value of $\msg$ is $\frac{8}{9}$, we expect each pair of parties to succeed with probability at least $\frac{8}{9}$ when using a quantum strategy.

Applying the NPA hierarchy to the optimization problem in Eq.~\eqref{app_qvalue_pure} with the simplified expression in Eq.~\eqref{pi_xy_app} and the additional constraints from Eq.~\eqref{constraint_one}-\eqref{constraint_three} yields the approximate value of $\frac{8.00077}{9}$ reported in \thref{th_msg} (see the code in supplemental material \cite{Cervero23}).
We found that these heuristical changes significantly improved the value output by the NPA hierarchy.
We used a mixed relaxation of the hierarchy, combining level $2$ and two thirds of the monomials of the form $P\otimes Q\otimes F$ (which were themselves picked arbitrarily by selecting two out of every three in the list $[P,Q,F]_{P,Q,F}$ in order).

Lastly, we remark that the constraints in Eq.~\eqref{constraint_two} and \eqref{constraint_three} need to be removed when optimizing the constrained problem in Eq.~\eqref{msg_constrained} to generate Figure \ref{fig:MSG_cons_w_upperbound}.

\end{document}